\documentclass[10pt,a4paper, reqno]{article}
\usepackage{mysetup}
\begin{document}

\title{Second-order fast-slow dynamics of non-ergodic Hamiltonian systems: Thermodynamic interpretation and simulation}

\author[1]{Matthias Klar\footnote{M.Klar@bath.ac.uk}}
\author[1]{Karsten Matthies\footnote{K.Matthies@bath.ac.uk}}
\author[2]{Celia Reina\footnote{creina@seas.upenn.edu}}
\author[3]{Johannes Zimmer\footnote{jz@ma.tum.de}}
\affil[1]{Department of Mathematical Sciences, University of Bath, Bath BA2 7AY, United Kingdom}
\affil[2]{Department of Mechanical Engineering and Applied Mechanics, University of Pennsylvania, Philadelphia PA 19104, USA}
\affil[3]{Department of Mathematics, Technische Universit\"at M\"unchen, Boltzmannstr.~3, 85748 Garching, Germany}

\maketitle

\begin{abstract}
    A class of fast--slow Hamiltonian systems with potential \(U_\eps\) describing the interaction of non-ergodic
    fast and slow degrees of freedom is studied. The parameter \(\eps\) indicates the typical timescale ratio of
    the fast and slow degrees of freedom. It is known that the Hamiltonian system converges for \(\eps\to0\) to
    a homogenised Hamiltonian system. We study the situation where \(\eps\) is small but positive.
    First, we rigorously derive the second-order corrections to the homogenised (slow)
    degrees of freedom. They can be decomposed into explicitly given terms that oscillate rapidly around zero
    and terms that trace the average motion of the corrections, which are given as the solution to an
    inhomogeneous linear system of differential equations.
    Then, we analyse the energy of the fast degrees of freedom expanded to second-order from a thermodynamic
    point of view. In particular, we define and expand to second-order a temperature, an entropy and external
    forces and show that they satisfy to leading-order, as well as on average to second-order, thermodynamic
    energy relations akin to the first and second law of thermodynamics.
    Finally, we analyse for a specific fast--slow Hamiltonian system the second-order
    asymptotic expansion of the slow degrees of freedom from a numerical point of view. Their approximation quality for short and long time frames and their total computation
    time are compared with those of the solution to the original fast--slow Hamiltonian system of similar
    accuracy.
\end{abstract}

\paragraph{Keywords:} Two-scale Hamiltonian, Asymptotic expansion, Coarse-graining, Far-from-equilibrium, Many-degrees-of-freedom interaction

\section{Introduction}

Many scientists in physics, chemistry and materials science resort to computer
simulations to study real-world
dynamical processes. These simulations open up the possibility to quickly and
inexpensively iterate through
different experimental setups, thus hugely reducing cost in the form of time and
labour and allow a level of
insight into small- and large-scale processes that were out of reach decades
ago. In chemical physics or
materials science, for example, scientists frequently analyse large scale
molecular dynamics simulations to
predict properties of large dynamical systems based on mathematical models that
aim to describe the dynamical
evolution of the constituent particles. In simulating these systems, one
typically encounters two problems
that severely impede their scalability. Firstly, the simulation of molecular
structures requires a step size
in the numerical integration scheme that ranges in the order of \(10^{-15}\)
seconds to accurately replicate the fast
molecular vibrations in the system. Secondly, even small macroscopic systems of
interest require the
integration of a potentially large number of particles. Even more, the two
problems often compound and pose a
challenging obstacle in the scalability and utility of molecular dynamic
simulations.

From an applications point of view, one is often not interested in analysing the
exact evolution of the fast
molecular vibrations, but in the slow conformal motion that embodies the macroscale
dynamics of the system. Here
lies an opportunity to bypass at least partly the scalability issues by
advancing the understanding of these
systems and a subsequent development of numerical integration schemes that
describe only the average evolution
of the dynamical system without fully resolving the small-scale vibrations.

Fast--slow Hamiltonian systems provide a simplified fundamental description of
large-scale interacting
particles systems, where the system's degrees of freedom evolve on different
scales in time and space. They
can be used, for example, to model the evolution of molecules where the slow
degrees of freedom represent the
conformal motion and the fast degrees of freedom represent the high-frequency
molecular vibrations~\cite{Bornemann}. There is
a vast body of literature for averaging general dynamical systems, not
necessarily of Hamiltonian type, for
example, using Young measures~\cite{Chatterjee2018a}. Applications of such
fast--slow systems arise, for example,
in models of plasticity~\cite{Chatterjee2020a}. Recent work on averaging of
Hamiltonian fast--slow systems and connections with adiabatic invariants
include~\cite{LERMAN20161219,Neishtadt_2017,Neishtadt_2019} and references
therein. Similar averaging techniques are also relevant to understand
equilibration in springy billiards~\cite{ShahE10514}.  For general references
to averaging, we refer the reader
to~\cite{SVM07,Kuehn2015a,Pavliotis2008a}.

With a mathematical description of a dynamical system in the form of a fast--slow
Hamiltonian system, we are able
to derive the conformal motion through homogenisation in time. The theory laid
out by Bornemann
in~\cite{Bornemann} enables us to derive the homogenised evolution of a
specific class of fast--slow
Hamiltonian systems. More precisely, Bornemann considers a family of mechanical
systems, parametrised by a
scale parameter \(\eps\), whose Lagrangian is of the form
\begin{equation}
    \label{eq:6}
    \mathscr{L}_\eps(x,\dot{x}) = \tfrac{1}{2}\left\langle \dot{x},\dot{x} \right\rangle - W_\eps(x),\qquad
    \dot{x}\in T_{x} M,
\end{equation}
on a Euclidean configuration space \(M=\R^m\) with a potential given by
\begin{equation*}
    W_\eps(x)=V(x)+\eps^{-2}U(x).
\end{equation*}
Here, the potential \(U\) characterises the fast dynamics of the system.
By splitting the coordinates according
to \(x=(y,z)\in \R^n\times\R^r=\R^m\), where \(y\) represents the slow and \(z\) the fast
degrees of freedom,
Bornemann showed that system~\eqref{eq:6} converges as \(\eps\to0\) to a system
on a slow manifold
\(N=U^{-1}(0)\), governed by the Lagrangian
\begin{equation}
    \label{eq:02}
    \mathscr{L}_{\hom}(x,\dot{x}) = \tfrac{1}{2} \left\langle \dot{x}, \dot{x} \right\rangle -V(x) -U_{\hom}(x),
    \qquad \dot{x}\in T_{x} N,
\end{equation}
where \(U_{\hom}\) can be derived from the Hessian of \(U\) and the initial
conditions of \(x\).

System~\eqref{eq:02} describes the slow, leading-order dynamics of the original
system~\eqref{eq:6}. As such,
it allows for the integration of the corresponding equations of motion with a
larger step size than what would
usually be required for the integration of the original system. This can speed
up the numerical integration
significantly.

A crucial aspect of approximating the fast--slow solution of system~\eqref{eq:6} by
a slow solution of
system~\eqref{eq:02} is that the approximation error depends on the scale
parameter \(\eps\). This scale
parameter is a critical element in the dynamics described by~\eqref{eq:6}. It
is determined by the underlying
true natural system that the model aims to represent and indicates the ratio of
the typical timescales of the
fast and slow degrees of freedom. In the case of a very small \(\eps\), a
description of the fast--slow
solution of system~\eqref{eq:6} by a slow solution of system~\eqref{eq:02}
might be an acceptable trade-off in
order to deal with the scalability issue mentioned earlier. However, a problem
arises if \(\eps\) is small,
so that the microscale oscillations severely affect the numerical integration,
but not small enough so that
the dynamics of system~\eqref{eq:6} cannot be sufficiently approximated by the
homogenised dynamics given by
system~\eqref{eq:02}. In this case, the dynamics of the fast degrees of freedom
contribute much more to the
evolution of the whole system than in the case of a very small \(\eps\). For
example, in~\cite[Chapter~III
    \textsection 2]{Bornemann}, the author applies the homogenisation process to
derive the conformal motion of a butane
molecule, where in the united atom representation, the scale parameter is
\(\eps\approx 0.25\), which cannot be
considered as small. It is thus natural to extend the theory presented
in~\cite{Bornemann} to describe the
slow dynamics of the original system on a finer scale, potentially revealing
microscale properties in the
case of a scale parameter away from the limit \(\eps\to0\). This line of research
begins already
in~\cite{Bornemann}, where formal asymptotic expansions are derived in
Appendix~C.

A further step was developed in~\cite{Klar2020}, where the authors derive a
second-order asymptotic expansion
to the solution of system~\eqref{eq:6} in the case of one fast and one slow
degree of freedom, i.e.,~\(n=r=1\). Although the model in~\cite{Klar2020} is
rather simple and the
fast subsystem is ergodic, the
fast--slow character is sufficient to derive properties of the fine-scale
dynamics that are characteristic for
thermodynamic processes. More precisely, for \(V(x)\equiv 0\) and \(U(x)=\frac{1}{2}\omega^2(y)z^2\), where
\(\omega>0\) is a smooth frequency function, the thermodynamic character of the
model in~\cite{Klar2020}
becomes evident by analysing the fast subsystem, which models the dynamics of
the fast degree of freedom
\(z\), as a motion that is perturbed by the evolution of the slow degree
of freedom \(y\). This setting allows
an interpretation of the fast subsystem from a thermodynamic point of view. By
applying the thermodynamic
theory for ergodic Hamiltonian systems, first developed by Boltzmann and
Gibbs~\cite{Gibbs1901}, and later
specified by Hertz~\cite{Hertz1910}, one derives expressions for temperature,
entropy and external force in
the fast subsystem. Utilising the second-order asymptotic expansion, one can
determine the leading-order terms
of these thermodynamic expressions and show that they satisfy a thermodynamic
energy relation akin to the
first and second law of thermodynamics. It turns out that the entropy expression
to \emph{leading-order} is
\emph{constant}, suggesting an interpretation of the leading-order dynamics as
an adiabatic thermodynamic
process. Remarkably, although away from the limit \(\eps\to0\), one finds a
similar energy relation for the
averaged second-order terms of the expansion. Most importantly, the entropy
expression to second-order is
\emph{not} constant. The dynamics to second-order can therefore be interpreted
as a non-adiabatic
thermodynamic process.

In this article, we carry out a comparable study for the case of more than one
fast and slow degrees of
freedom, with the important difference that the higher dimensional fast
subsystem is non-ergodic. We extend
the theory presented in~\cite{Bornemann} and derive the second-order asymptotic
expansion to the solution of
system~\eqref{eq:6} in the case of an arbitrary finite number of fast and slow
degrees of freedom, i.e.,~\(n,r\in \mathbb{N}\). Specifically, we analyse the mechanical
system~\eqref{eq:6} with a smooth potential
\(V=V(y)\) and \(U(x)=\tfrac{1}{2}\left\langle H(y)z,z \right\rangle\), where
\(H(y)=\mathrm{diag}(\omega_1^2(y),\ldots,\omega_r^2(y))\) for smooth frequency functions \(\omega_\lambda>0\)
\((\lambda=1,\ldots,r)\). Unlike in~\cite{Klar2020}, we have to impose certain non-resonance
conditions to
derive the second-order asymptotic expansion. Following the strategy presented
in~\cite{Klar2020}, a key
element in the derivation of the second-order asymptotic expansion is a
transformation of the fast degrees of
freedom into action--angle variables. By using weak convergence methods we show
that the second-order
asymptotic expansion of the \(\eps\)-dependent transformed variables is
given, for instance in the case of
\(y_\eps\), as \(y_\eps = y_0 + \eps^2 (\bar{y}_2 + [y_2]^\eps) + \eps^2 y_3^\eps\), where \(y_3^\eps \to 0\)
in \(C([0,T], \R^n)\) as \(\eps \to 0\). Here, the function \(y_0\) is the
leading-order term derived from
system~\eqref{eq:02}, the function \(\bar{y}_2\) is the slow component of the
second-order correction, which
can be derived as the solution to an inhomogeneous linear system of differential
equations, and the function
\([y_2]^\eps\) is the fast component of the second-order correction, which consists
of explicitly given
rapidly oscillating terms that converge weakly\(^\ast\) to zero in
\(L^\infty([0,T],\R^n)\).

Furthermore, we interpret the dynamics of the fast subsystem, which is composed
of the fast degrees of
freedom, from a thermodynamic point of view. This is based on the thermodynamic
theory for Hamiltonian systems
formalised by Hertz~\cite{Hertz1910} and used in~\cite{Klar2020}. More
precisely, by decomposing the
total energy \(E_\eps\) into the energies \(E_\eps^\parallel\) and \(E_\eps^\perp\) such that
\(E_\eps = E_\eps^\parallel + E_\eps^\perp\), we consider \(E_\eps^\perp(z_\eps, \dot{z}_\eps; y_\eps)\) as
the energy describing the evolution of the fast degrees of freedom \(z_\eps\)
under a slow, external influence
described by the dynamics of \(y_\eps\). As the fast subsystem is not
necessarily ergodic, we follow along the
lines of~\cite{Berdichevsky} and replace time averages in the thermodynamic
theory by ensemble averages, i.e.,~averages over uniformly distributed initial
values on the energy surface. With this modification, we apply
Hertz' thermodynamic formalism and derive a temperature \(T_\eps\), an entropy
\(S_\eps\) and an external
force \(F_\eps\) for the fast subsystem. By applying the asymptotic expansion
results from the first part of
this article, we similarly expand the energy
\(E_\eps^\perp = E_0^\perp + \eps [E_1^\perp]^\eps + \eps^2 \left( \bar{E}_2^\perp + [E_2^\perp]^\eps \right)
+ \eps^2 E_3^{\perp\eps}\), the temperature \(T_\eps = T_0 + \mathcal{O}(\eps)\), the entropy
\(S_\eps = S_0 + \eps [S_1]^\eps + \eps^2 \left( \bar{S}_2 + [S_2]^\eps \right) + \eps^2 S_3^\eps\) and the
external force \(F_\eps = F_0 + \mathcal{O}(\eps)\), where \(E_3^{\perp\eps}, S_3^\eps \to 0\) in
\(C([0,T])\). We find that to leading-order the thermodynamic quantities satisfy
an energy relation akin to
the first and second law of thermodynamics (in the sense of
Carath\'eodory~\cite{Weiner2012})
\begin{equation*}
    d E_0^\perp = \sum_{j=1}^n F_0^j d y^j_0+T_0 dS_0.
\end{equation*}
In contrast to the work in~\cite{Klar2020}, the leading-order entropy
\(S_0\) can be \emph{constant} or
\emph{non-constant}, depending on the characteristics of the weighted frequency
ratios
\(\theta_\ast^\lambda\omega_\lambda(y_0)/\omega_\mu(y_0)\) \((\lambda,\mu = 1,\ldots, r)\). Here, we use the definition of the entropy as the logarithm of the phase space volume where the latter does not have to change slowly; the analysis shows that even in this situation, a meaningful thermodynamic setting exists. As a consequence,
we interpret the dynamics to leading-order as an adiabatic or non-adiabatic
thermodynamic process,
respectively. Moreover, by considering the average dynamics to second-order for
fixed \(y_0\) and \(p_0=\dot{y}_0\), we
similarly find, although away from the limit \(\eps\to0\), a comparable energy
relation of the form
\begin{equation*}
    d \bar{E}_2^\perp = \sum_{j=1}^n F_0^j d \bar{y}_2^j + T_0 d \bar{S}_2.
\end{equation*}
Likewise, with a non-constant second-order entropy expression \(\bar{S}_2\), we
can interpret the averaged
second-order dynamics as a non-adiabatic thermodynamic process.

Finally, we analyse the viability of the second-order asymptotic expansion as a
suitable approximation to the
slow degrees of freedom of system~\eqref{eq:6} from a numerical point of view.
More precisely, we choose a
specific model from the class of fast--slow Hamiltonian systems represented
by~\eqref{eq:6} and compare the
numerical solution of \(y_\eps\) with \(y_0 + \eps^2 (\bar{y}_2 + [y_2]^\eps)\) in terms of its short- and
long-term approximation quality and computation time. The maximal time frame for
which an approximation of
\(y_\eps\) can be considered sufficiently accurate significantly increases by
using
\(y_0 + \eps^2 (\bar{y}_2 + [y_2]^\eps)\) instead of \(y_0\) alone. Moreover, we show that the computation
of
\(y_0 + \eps^2 (\bar{y}_2 + [y_2]^\eps)\) is up to two orders of magnitude faster (depending on the scale
parameter \(\eps\)) than a computation of \(y_\eps\) to comparable accuracy
as a solution to
system~\eqref{eq:6}. As described earlier, the reason is that fast oscillations
severely affect the runtime
for numerically computing \(y_\eps\) from~\eqref{eq:6}. In contrast, the
problematic oscillatory term at
second-order \([y_2]^\eps\) is given explicitly, and the derivation of \(y_0\)
and \(\bar{y}_2\) only require a
numerical integration of two slow systems of differential equations, which can
be solved, in parallel, using a
relatively large step size.

An application of the theory presented in this article may not only improve
large-scale molecular dynamics
simulations. It can also find applications in cases where the homogenisation
theory outlined above and related
work as in~\cite{Reich98} are applicable. Some examples are given by the
description of quantum--classical
models in quantum-chemistry~\cite{Bornemann}, the problem of deriving the
guiding centre motion in plasma
physics~\cite{Bornemann1997} or, more recently, the derivation of a
coarse-grained description of the coupled
thermoelastic behaviour from an atomistic model in materials
science~\cite{Li2019}.

Finally, we want to point out other thermodynamic analyses based on the
fast--slow system governed by the Lagrangian~\eqref{eq:6}. In~\cite{Jia2005},
the authors extend system~\eqref{eq:6} by coupling the fast and slow degrees
of freedom to an external Nos\'e--Hoover thermostat and analyse the
thermodynamic equilibration of the system on the fast and slow scale. In a
similar line of thought, the authors in~\cite{Reich2000} expand
system~\eqref{eq:6} by embedding it into an external heat bath and
subsequently analysing the resulting slow dynamics, analogous to the
homogenisation procedure introduced above, in the limit \(\eps\to 0\).

\subsection{Outline of the paper}

In  Section~\ref{sec:8} we introduce the model problem, which establishes the
foundation for the analysis in
this article, and state necessary non-resonance conditions, which ensure that
the subsequently derived
second-order expansion of the solution to the model problem is well-defined. A
summary of our main results is
provided in  Section~\ref{sec:9}. We start the analysis of the model problem by
introducing a transformation of
the fast degrees of freedom into action--angle variables in  Section~\ref{sec:7}, where
we also prove the
existence and uniqueness of a solution of the transformed system. In
Section~\ref{sec:5} we introduce some
notation that simplifies the governing equations of motion and derive the
second-order asymptotic expansion
for the transformed degrees of freedom. Subsequently, in  Section~\ref{sec:4}, we
define expressions for the
temperature, the entropy and the external force for the fast subsystem and
interpret the model from a
thermodynamic point of view. For a test model, the global error for
approximating \(y_\eps\) by
\(y_0 + \eps^2 (\bar{y}_2 + [y_2]^\eps)\) are analysed on short and long time intervals in  Section~\ref{sec:3},
where
we also compare the runtimes for computing \(y_\eps\), \(y_0\) and
\(\bar{y}_2 + [y_2]^\eps\).  Section~\ref{sec:10}
provides a short conclusion of this article. In Appendix~\ref{app:1} we summarise how
the thermodynamic
expressions can be derived for the fast subsystem. Finally, in Appendix~\ref{app:2} we
present some data on
the computation times corresponding to the maximal step sizes used in the numerical
simulations presented in this article.

\section{The model problem}

\label{sec:8}
For a small scale parameter \(0<\eps<\eps_0<\infty\), we study the family of mechanical
systems given by the
Lagrangian
\begin{equation}
    \label{eq:67}
    \mathscr{L}_\eps(x,\dot{x}) = \tfrac{1}{2}\left\langle \dot{x},\dot{x} \right\rangle - W_\eps(x),\qquad
    \dot{x}\in T_{x} M,
\end{equation}
on a Euclidean configuration space \(M=\R^m\). Here and in the following,
\(\left\langle \cdot, \cdot \right\rangle \) denotes Euclidean inner products and
\(\left\vert \cdot \right\vert\) denotes Euclidean norms. Splitting the coordinates according to
\(x=(y,z)\in \R^n\times\R^r=\R^m\), we specify, following~\cite{Bornemann}, the potential
\(W_\eps=V+\eps^{-2}U\) by a smooth potential \(V=V(y)\), which is assumed to be bounded
from below and
\begin{equation}
    \label{eq:U-ass}
    U(x)=\tfrac{1}{2}\left\langle H(y)z,z \right\rangle \qquad \text{with} \qquad H(y)
    =\mathrm{diag}(\omega_1^2(y),\ldots,\omega_r^2(y)).
\end{equation}
We assume that the smooth functions \(\omega_\lambda \in C^\infty(\R^n)\) are uniformly positive,
i.e.,~there
exists a constant \(\omega_\ast>0\) such that
\begin{equation}
    \label{eq:1}
    \omega_\lambda(y)\geq \omega_\ast, \qquad y\in \R^n,\quad\lambda=1,\ldots,r.
\end{equation}
A componentwise formulation of the equations of motion for the
\(\eps\)-dependent coordinates \(y_\eps\) and
\(z_\eps\) in~\eqref{eq:67} yields
\begin{IEEEeqnarray}{rCl}
    \eqlabel{eq:2}
    \IEEEyesnumber \IEEEyessubnumber*
    \ddot{y}_\eps^j&=& -\partial_j V(y_\eps)-\tfrac{1}{2}\eps^{-2}\left\langle \partial_j H(y_\eps)z_\eps,
    z_\eps \right\rangle,
    \qquad j = 1,\ldots, n, \\
    \label{eq:3}
    \ddot{z}_\eps &=& -\eps^{-2}H(y_\eps)z_\eps.
\end{IEEEeqnarray}
Moreover, we consider the \(\eps\)-independent initial values
\begin{equation}
    \label{eq:7}
    y_\eps(0)=y_\ast, \qquad \dot{y}_\eps(0)=p_\ast, \qquad z_\eps(0)=0,\qquad \dot{z}_\eps(0)=u_\ast.
\end{equation}
We notice that the energy \(E_\eps\) of the system is independent of
\(\eps\) due to the particular choice
\(z_\eps(0)=0\),
\begin{equation}
    \label{eq:4}
    E_\eps = \tfrac{1}{2}\left\vert \dot{y}_\eps \right\vert^2 + \tfrac{1}{2}\left\vert \dot{z}_\eps \right\vert^2
    +V(y_\eps)+ \eps^{-2}U(y_\eps, z_\eps)
    = \tfrac{1}{2}\left\vert p_\ast \right\vert^2 + \tfrac{1}{2} \left\vert u_\ast \right\vert^2 + V(y_\ast)= E_\ast.
\end{equation}

\begin{remark}For the equations in~\eqref{eq:2}--\eqref{eq:4} and below, we will
    simultaneously make use of the vector
    notation for the coordinates \(y_\eps\in\R^n\) and \(z_\eps\in\R^r\) (and related
    expressions) as well as
    their componentwise representation \(y_\eps^j\) \((j=1,\ldots,n)\) and \(z_\eps^\lambda\)
    \((\lambda = 1,\ldots,r)\). The index in the superscript should not be confused with an
    exponent.
\end{remark}

We are primarily interested in the evolution of the slow degrees of freedom
\(y_\eps^j\)
\((j=1,\ldots,n)\). The following theorem by Bornemann shows that \(y_\eps\)
converges in the limit
\(\eps\to0\) to a function \(y_0\) which is given as the solution to a
second-order differential equation.
\begin{theorem}[Bornemann,~\cite{Bornemann}]
    \label{thm:1}
    For
    \begin{equation*}
        U_{\hom}(y_0)=\sum_{\lambda=1}^r\theta_\ast^\lambda \omega_\lambda(y_0), \qquad \text{where} \qquad
        \theta_\ast^\lambda = \frac{\left\vert u_\ast^\lambda \right\vert^2}{2\omega_\lambda(y_\ast)}, \qquad \lambda=1,\ldots,r,
    \end{equation*}
    let \(y_0\) be the solution to the second-order differential equation
    \begin{equation}
        \label{eq:62}
        \ddot{y}_0^j = -\partial_j V(y_0) - \partial_j U_{\hom}(y_0),\qquad j=1,\ldots, n,
    \end{equation}
    with initial values \(y_0(0)=y_\ast\), \(\dot{y}_0(0)=p_\ast\). Then, for every finite time
    interval
    \([0,T]\), we obtain the strong convergence
    \begin{equation*}
        y_\eps \to y_0 \quad \text{in} \quad C^1([0,T],\R^n)
    \end{equation*}
    and the weak\(^\ast\) convergences \(\eps^{-1}z_\eps \weak{\ast}0\) and \(\dot{z}_\eps \weak{\ast}0\)
    in \(L^\infty([0,T],\R^r)\).
\end{theorem}
Theorem~\ref{thm:1} shows that the family of mechanical systems~\eqref{eq:67}
converges as \(\eps\to0\) to a
mechanical system which is again Hamiltonian.

\subsection{Non-resonance conditions}

As the interaction of multiple oscillating degrees of freedom can lead to
resonance effects in the system, we
will, similar to~\cite{Bornemann}, impose suitable non-resonance conditions on
the frequencies
\(\omega_\lambda\) to ensure that the second-order asymptotic expansions, which we
will derive in
Section~\ref{sec:5}, are well-defined. We say, referring to the definition stated for
example in~\cite[Section
    14.6]{Wiggins2003}, that a resonance of order \(j\in \mathbb{N}\) at \(y\in \R^n\) is
given by the relation
\begin{equation}
    \label{eq:5}
    \gamma_1 \omega_1(y) + \cdots+ \gamma_r \omega_r(y) = 0, \qquad \left\vert \gamma_1 \right\vert + \cdots
    +\left\vert \gamma_r \right\vert = j,
\end{equation}
with integer coefficients \(\gamma_\lambda\in \mathbb{Z}\) for \(\lambda=1,\ldots,r\). Note that the
non-degeneracy condition~\eqref{eq:1} implies that there is no resonance of
order one.

\begin{assumption}\label{ass:1}We assume that the homogenised solution in Theorem~\ref{thm:1} is non-resonant of
    order two, i.e.,~we assume
    that
    \begin{equation*}
        \gamma_1 \omega_1(y_0(t)) + \cdots+ \gamma_r \omega_r(y_0(t)) \neq 0, \qquad \left\vert \gamma_1 \right\vert
        + \cdots +\left\vert \gamma_r \right\vert = 2,
    \end{equation*}
    for all \(t\in[0,T]\).
\end{assumption}

\begin{assumption}
    \label{ass:2}Moreover, we assume that the homogenised solution in Theorem~\ref{thm:1} is not flatly
    resonant up to order
    three. More precisely, we assume that
    \begin{equation*}
        \frac{d}{dt}\left( \gamma_1 \omega_1(y_0(t_i)) + \cdots+ \gamma_r \omega_r(y_0(t_i)) \right)\neq 0, \qquad
        \left\vert \gamma_1 \right\vert + \cdots +\left\vert \gamma_r \right\vert \leq 3,
    \end{equation*}
    for all impact times \(t_i\in [0,T]\) \((i\in I\subset \mathbb{N}, I \text{ finite})\) such that the
    non-resonance condition~\eqref{eq:5} holds at \(y_0(t_i)\).
\end{assumption}

We remark that Assumption~\ref{ass:1} is intentionally chosen to simplify the derivation
of the second-order
asymptotic expansions (see remark following Lemma~\ref{lemma:6}). Under these
simplifications the assumption also ensures, that the
second-order asymptotic expansions derived in  Section~\ref{sec:5} are well-defined.
Assumption~\ref{ass:2} is,
analogous to~\cite{Bornemann}, a necessary prerequisite for the theory
developed below. It ensures that
rapidly oscillating functions of the form
\(\exp \left( \pm i \eps^{-1}\left( \omega_\lambda(y_\eps) - \omega_\mu(y_\eps) \right) \right)\) and
\(\exp \left( \pm i \eps^{-1}\left( \omega_\lambda(y_\eps) + \omega_\nu(y_\eps) - \omega_\mu(y_\eps) \right)
\right)\) where \(\lambda,\mu, \nu=1,\ldots, r\), \(\lambda\neq \mu\) converge weakly\(^\ast\) to
zero in
\(L^\infty([0,T])\). In~\cite{Bornemann} these functions appear due to interactions of
the fast degrees of
freedom caused by the structure of a more general potential \(U(x)\) as well
as a more general metric
\(\left\langle \cdot, \cdot \right\rangle \) and Assumption~\ref{ass:2} is used to derive the leading-order
asymptotic expansion of the system's degrees of freedom. Here, however, these
functions appear only due to
small-scale interactions in the second-order asymptotic expansions.

\section{Summary of the main results}

\label{sec:9}
The goal of this article is to extend the theory developed in~\cite{Bornemann}
by deriving the second-order asymptotic expansion rigorously for the solution of the equations of
motion~\eqref{eq:2} and interpret the
corresponding second-order asymptotic expansion of the energy~\eqref{eq:4} from
a thermodynamic point of
view. Note that the mechanical system~\eqref{eq:67} is not a classical
thermodynamic system. In particular,
the fast subsystem, consisting of the fast degrees of freedom \(z_\eps^\lambda\)
\((\lambda=1,\ldots,r)\),
which we will consider in  Section~\ref{sec:4} as the thermodynamic part of the whole
system, is in general not
ergodic. Finally, we will discuss the numerical implications of the second-order
asymptotic expansion of \(y_\eps\) in terms
of its approximation error and computational cost.

Our main findings in this article can be summarised as follows.

\begin{enumerate}
    \item \label{thm:p1} After transforming the rapidly oscillating degrees of
          freedom into action--angle variables
          \((z_\eps,\dot{z}_\eps)\mapsto (\theta_\eps, \phi_\eps)\), which also involves a transformation of the
          generalised momentum \(\dot{y}_\eps\mapsto p_\eps\), we derive the second-order asymptotic
          expansion of
          \(y_\eps, p_\eps, \theta_\eps, \phi_\eps\). This takes the form
          \begin{IEEEeqnarray*}{rClClClCl}
              y_\eps &=& y_0 &+& \eps[\bar{y}_1]^\eps&+&\eps^2[\bar{y}_2]^\eps &+& \eps^2y_3^\eps,\\
              p_\eps &=& p_0 &+& \eps[\bar{p}_1]^\eps&+&\eps^2[\bar{p}_2]^\eps &+& \eps^2p_3^\eps,\\
              \theta_\eps &=& \theta_\ast &+& \eps[\bar{\theta}_1]^\eps&+&\eps^2[\bar{\theta}_2]^\eps &+& \eps^2\theta_3^\eps,\\
              \phi_\eps &=& \phi_0 &+& \eps[\bar{\phi}_1]^\eps&+&\eps^2[\bar{\phi}_2]^\eps &+& \eps^2\phi_3^\eps,
          \end{IEEEeqnarray*}
          where for \(i\in \{1,2\}\),
          \begin{IEEEeqnarray*}{rClClClCl+rClCl}
              [\bar{y}_i]^\eps &\coloneqq& \bar{y}_i &+& [y_i]^\eps&\weak{\ast}& \bar{y}_i \quad &\text{in}
              &\quad L^\infty([0,T], \R^n),&y_3^\eps &\to& 0 \quad &\text{in}&\quad C([0,T], \R^n),\\
              \phantom{}[\bar{p}_i]^\eps &\coloneqq& \bar{p}_i &+& [p_i]^\eps&\weak{\ast}& \bar{p}_i \quad &\text{in}
              &\quad L^\infty([0,T], \R^n),&p_3^\eps &\to& 0 \quad &\text{in}&\quad C([0,T], \R^n),\\
              \phantom{}[\bar{\theta}_i]^\eps &\coloneqq& \bar{\theta}_i &+& [\theta_i]^\eps&\weak{\ast}& \bar{\theta}_i \quad
              &\text{in}&\quad L^\infty([0,T],\R^r),&\theta_3^\eps &\to& 0 \quad &\text{in}&\quad C([0,T],\R^r),\\
              \phantom{}[\bar{\phi}_i]^\eps &\coloneqq& \bar{\phi}_i &+& [\phi_i]^\eps&\weak{\ast}& \bar{\phi}_i \quad
              &\text{in}&\quad L^\infty([0,T],\R^r),&\phi_3^\eps &\to& 0 \quad &\text{in}&\quad C([0,T],\R^r).
          \end{IEEEeqnarray*}
          In other words, for each degree of freedom the second-order asymptotic expansion
          is characterised -- to
          leading-order by the theory developed in~\cite{Bornemann} (Theorem~\ref{thm:1}) -- to \(i\)th order by a
          decomposition into a slow term, indicated by an overbar, which constitutes the
          average motion of the
          \(i\)th order expansion, and a fast term, indicated by square brackets,
          which oscillates rapidly and
          converges weakly\(^\ast\) to zero --- and by a residual term, indicated with
          a subscript three, which
          converges uniformly to zero. In particular, we show that
          \begin{equation*}
              [\bar{y}_1]^\eps = 0,\qquad [\bar{p}_1]^\eps = 0,\qquad  [\bar{\theta}_1]^\eps =
              [\theta_1]^\eps,\qquad [\bar{\phi}_1]^\eps = 0,\end{equation*}
          and that \((\bar{\phi}_2, \bar{\theta}_2, \bar{y}_2, \bar{p}_2)\) is given as the solution to an
          inhomogeneous linear system of differential equations (Theorem~\ref{thm:2}). Moreover,
          the rapidly
          oscillating functions \([\theta_1]^\eps\), \([y_2]^\eps\), \([p_2]^\eps\), \([\theta_2]^\eps\) and
          \([\phi_2]^\eps\) are explicitly given in Definition~\ref{def:1}.
    \item In~\cite{Hertz1910}, Hertz formalises a thermodynamic theory for fast
          Hamiltonian systems which are
          perturbed by slow external agents. We regard the fast subsystem \((z_\eps, \dot{z}_\eps)\) as
          such a
          thermodynamic system, perturbed by the slow motion of \((y_\eps, \dot{y}_\eps)\). Since the
          fast subsystem
          is not ergodic, we follow along the lines
          of~\cite[Chapter~1.10]{Berdichevsky} and replace the time average,
          which is an essential component in the thermodynamic theory, by the ensemble
          average, i.e.,~the average over
          uniformly distributed initial values on the energy surface (see Appendix~\ref{app:1}),
          and define, based on
          Hertz' formulation, a temperature \(T_\eps\), an entropy \(S_\eps\) and an
          external force \(F_\eps\) for the
          fast subsystem.

          In combination with the analytic result discussed under~\ref{thm:p1}, we
          decompose the total energy
          \(E_\eps\) into the energy associated with the fast subsystem \(E_\eps^\perp\) and
          its residual energy
          \(E_\eps^\parallel= E_\eps - E_\eps^\perp\), and expand, similar to above \(E_\eps^\perp\),
          \(E_\eps^\parallel\), \(T_\eps\), \(S_\eps\) and \(F_\eps\) into the form
          \begin{IEEEeqnarray*}{rClClClCl}
              E_\eps^\perp &=& E_0^\perp &+& \eps[\bar{E}_1^\perp]^\eps &+&\eps^2[\bar{E}_2^\perp]^\eps&+& \eps^2E_3^{\perp\eps},\\
              E_\eps^\parallel &=& E_0^\parallel &+& \eps[\bar{E}_1^\parallel]^\eps &+&\eps^2[\bar{E}_2^\parallel]^\eps&+& \eps^2E_3^{\parallel\eps},\\
              S_\eps &=& S_0 &+&\eps[\bar{S}_1]^\eps &+& \eps^2[\bar{S}_2]^\eps&+& \eps^2S_3^\eps,\\
              T_\eps &=& T_0 &+& \mathcal{O}(\eps),\\
              F_\eps &=& F_0 &+& \mathcal{O}(\eps),
          \end{IEEEeqnarray*}
          where for \(i\in \{1,2\}\),
          \begin{IEEEeqnarray*}{rClClClCl+rClCl}
              [\bar{E}_i^\perp]^\eps &\coloneqq& \bar{E}_i^\perp &+& [E_i^\perp]^\eps&\weak{\ast}& \bar{E}_i^\perp \quad
              &\text{in}&\quad L^\infty([0,T]),&E_3^{\perp\eps} &\to& 0 \quad &\text{in}&\quad C([0,T]),\\
              \phantom{}[\bar{E}_i^\parallel]^\eps &\coloneqq& \bar{E}_i^\parallel &+& [E_i^\parallel]^\eps&\weak{\ast}
              & \bar{E}_i^\parallel \quad &\text{in}&\quad L^\infty([0,T]),&E_3^{\parallel\eps} &\to& 0 \quad &\text{in}&\quad C([0,T]),\\
              \phantom{}[\bar{S}_i]^\eps &\coloneqq& \bar{S}_i &+& [S_i]^\eps&\weak{\ast}& \bar{S}_i \quad
              &\text{in}&\quad L^\infty([0,T]),&S_3^\eps &\to& 0 \quad &\text{in}&\quad C([0,T]).
          \end{IEEEeqnarray*}

          The characterisation of the \(i\)th order expansion is similar
          to~\ref{thm:p1} and is already discussed,
          for the case of \(n=r=1\), in~\cite{Klar2020}. In  Section~\ref{sec:4} we
          interpret these asymptotic
          expansions from a thermodynamic point of view. In particular, we show that, to
          leading-order, the dynamics can
          be interpreted as a thermodynamic process characterised by the energy relation
          \begin{equation*}
              d E_0^\perp = \sum_{j=1}^n F_0^j d y^j_0+T_0 dS_0.
          \end{equation*}
          In contrast to the analysis in~\cite{Klar2020}, we find, provided that
          \(\theta_\ast^\lambda \neq 0\) for at
          least one \(\lambda=1,\ldots, r\), that the entropy expression to leading-order is
          constant, \(dS_0=0\),
          \emph{if and only if} all pairwise weighted frequency ratios
          \(\theta_\ast^\lambda\omega_\lambda(y_0)/\omega_\mu(y_0)\) \((\lambda, \mu = 1, \ldots, r)\) are
          constant. In this case, the leading-order dynamics can be interpreted as an
          adiabatic thermodynamic
          process. Yet, if any of the weighted frequency ratios is non-constant, the
          entropy is non-constant and thus
          the leading-order dynamics can be interpreted as a non-adiabatic thermodynamic process. Here we use the definition of entropy given by Hertz in a context where the
          entropy is not necessarily the logarithm of an adiabatic invariant. Nevertheless, we show that a meaningful thermodynamic interpretation can be given.

          Furthermore, we show that the averaged second-order dynamics, i.e.,~the dynamics
          in the weak\(^\ast\) limit
          of the second-order terms, indicated by an overbar, represents for fixed
          \((y_0, p_0)\) a non-adiabatic thermodynamic process with an averaged \emph{non-constant entropy}, \(d \bar{S}_2\neq 0\), which also
          satisfies relations
          akin to equilibrium thermodynamics, despite being beyond the limit \(\eps \to 0\),
          \begin{equation*}
              d \bar{E}_2^\perp = \sum_{j=1}^n F_0^j d \bar{y}_2^j + T_0 d \bar{S}_2.
          \end{equation*}
          Finally, we show in Theorem~\ref{thm:3} that the evolution of \((\bar{y}_2, \bar{p}_2)\) is governed
          by
          equations which resemble Hamilton's canonical equations,
          \begin{equation*}
              \frac{d \bar{y}_2 }{dt} = \frac{\partial \bar{E}_2}{\partial p_0},\qquad \frac{\partial \bar{p}_2}{dt}
              = -\frac{\partial \bar{E}_2 }{\partial y_0},
          \end{equation*}
          for \(\bar{E}_2 = \bar{E}_2^\perp + \bar{E}_2^\parallel\), which are complemented by the
          \(\eps\)-independent initial values
          \begin{equation*}
              \bar{y}_2(0)= -[y_2]^\eps(0), \qquad \bar{p}_2(0)= -[p_2]^\eps(0).
          \end{equation*}

    \item Finally, we compare in numerical simulations the second-order asymptotic
          expansion of the slow degrees
          of freedom \(y_0 + \eps^2 (\bar{y}_2+[y_2]^\eps)\) with simulations for \(y_\eps\) of the original
          system~\eqref{eq:2}. The latter is computationally expensive, as it requires a
          numerical integration of the
          fast degrees of freedom \(z_\eps\). To this end, we derive numerically the
          slow motion \(y_0\) of the
          leading-order system~\eqref{eq:62} and the average motion \(\bar{y}_2\) of the
          second-order
          system~\eqref{eq:27} and combine them with the explicitly given rapidly
          oscillating components
          \([y_2]^\eps\) of the second-order expansion as specified in Definition~\ref{def:1}. We
          find, depending on
          the value of the scale parameter \(\eps\), that the computation time for the
          second-order expansion is up to
          two orders of magnitude faster than the computation time for the slow degrees of
          freedom of the original
          system. Moreover, we show that \(y_0 + \eps^2 (\bar{y}_2+[y_2]^\eps)\) provides an approximation of
          \(y_\eps\) which has significantly better global error bounds on long time
          intervals than an approximation
          by \(y_0\) alone.
\end{enumerate}

\section{The model problem in action--angle variables}

\label{sec:7}
To study the dynamics of \(y_\eps\) and \(z_\eps\) on different scales, a
detailed asymptotic analysis is
required. Such an analysis was already presented for the model problem as
introduced in  Section~\ref{sec:8} in
the case of one fast and one slow degree of freedom (i.e.,~\(n=r=1\))
in~\cite{Klar2020}, which extends the
analysis given in~\cite[Appendix C]{Bornemann}. To derive the second-order
asymptotic expansion of the
solution to the model problem for arbitrary \(n,r\in \mathbb{N}\), we analogously start
by rephrasing the
governing system of Newtonian equations~\eqref{eq:2} by transforming the fast
degrees of freedom
\((z_\eps, \dot{z}_\eps)\) into action--angle variables \((\theta_\eps, \phi_\eps)\).

We denote the canonical momenta corresponding to the positions \((y_\eps, z_\eps)\) as
\((\eta_\eps, \zeta_\eps)\). Then, the equations of motion~\eqref{eq:2}, together with the
velocity relations
\begin{equation*}
    \dot{y}_\eps = \eta_\eps,\qquad \dot{z}_\eps =\zeta_\eps,\end{equation*}
are given by the canonical equations of motion belonging to the energy function
\begin{equation*}
    E_\eps = \frac{1}{2}\vert \eta_\eps \vert^2 + \frac{1}{2}\vert \zeta_\eps \vert^2 +V(y_\eps)
    + \frac{1}{2}\eps^{-2}\sum_{\lambda=1}^r \omega_\lambda^2(y_\eps)(z_\eps^\lambda)^2.
\end{equation*}
The transformation \((z_\eps, \zeta_\eps)\mapsto (\theta_\eps, \phi_\eps)\) can be found by the theory of
generating functions~\cite{Arnold1989} as presented
in~\cite[Appendix~C]{Bornemann}. For fixed \(y_\eps\), the
generating function is given by
\begin{equation*}
    S_0(z_\eps, \phi_\eps;y_\eps)
    =  \frac{1}{2\eps}\sum_{\lambda=1}^r\omega_\lambda(y_\eps)(z_\eps^\lambda)^2 \cot(\eps^{-1}\phi_\eps^\lambda),\end{equation*}
via \(\zeta_\eps = \partial S_0/\partial z_\eps\) and \(\theta_\eps = -\partial S_0/\partial \phi_\eps\). With
this transformation, the fast degrees of freedom \((z_\eps, \zeta_\eps)\) can be written as
\begin{equation*}
    z_\eps^\lambda = \eps \sqrt{ \frac{2\theta_\eps^\lambda}{\omega_\lambda(y_\eps)}}\sin(\eps^{-1}\phi_\eps^\lambda),
    \qquad \zeta_\eps^\lambda =\sqrt{2\theta_\eps^\lambda \omega_\lambda(y_\eps)}\cos(\eps^{-1}\phi_\eps^\lambda).
\end{equation*}
It turns out, however, that the transformation \((z_\eps, \zeta_\eps)\mapsto (\theta_\eps,\phi_\eps)\) is
symplectic only for fixed \(y_\eps\). To derive a transformation that
preserves the symplectic structure on
the whole phase-space, one introduces the generalised momenta \(p_\eps\)
through another transformation
\(\eta_\eps \mapsto p_\eps\). To this end, we define the extended generating function
\(S(y_\eps, p_\eps, z_\eps, \phi_\eps) = p_\eps^T y_\eps+S_0(z_\eps, \phi_\eps;y_\eps)\) which does not
transform the position \(y_\eps = \partial S /\partial p_\eps\), but changes the momentum \(\eta_\eps\) such
that the transformation remains symplectic on the whole phase-space. The missing
transformation of the
momentum \(\eta_\eps\) is given componentwise for \(j=1,\ldots, n\) by
\begin{equation*}
    \eta_\eps^j = \frac{\partial S}{\partial y^j_\eps} = p^j_\eps +\eps \sum_{\lambda=1}^r\frac{\theta_\eps^\lambda \cdot
        \partial_j \omega_\lambda(y_\eps)}{2\omega_\lambda(y_\eps)}\sin(2\eps^{-1}\phi^\lambda_\eps).
\end{equation*}
By construction, the resulting transformation
\((y_\eps, \eta_\eps;z_\eps, \zeta_\eps)\mapsto (y_\eps, p_\eps; \phi_\eps, \theta_\eps)\) is symplectic.

The energy can be expressed in the new coordinates as
\begin{IEEEeqnarray*}{rCl}
    E_\eps &=& \frac{1}{2}\vert p_\eps \vert^2 + V(y_\eps)+\sum_{\lambda=1}^r \theta_\eps^\lambda \omega_\lambda(y_\eps)
    + \eps\sum_{j=1}^n\sum_{\lambda=1}^r \frac{ \theta_\eps^\lambda p_\eps^j
        \cdot\partial_j \omega_\lambda(y_\eps)}{2\omega_\lambda(y_\eps)}\sin(2\eps^{-1}\phi^\lambda_\eps)\\
    &&+ \> \frac{\eps^2}{8} \sum_{j=1}^n \left( \sum_{\lambda=1}^r \frac{\theta_\eps^\lambda\cdot
            \partial_j \omega_\lambda(y_\eps)}{\omega_\lambda(y_\eps)} \sin(2\eps^{-1}\phi_\eps^\lambda)\right)^2.
\end{IEEEeqnarray*}
Thus, by the canonical formalism, the equations of motion take the form
\begin{equation*}
    \dot{\phi}_\eps^\lambda = \frac{\partial E_\eps}{\partial \theta_\eps^\lambda},\qquad \dot{\theta}_\eps^\lambda
    = -\frac{\partial E_\eps}{\partial \phi_\eps^\lambda},\qquad \dot{y}_\eps^j
    = \frac{\partial E_\eps}{\partial p_\eps^j},\qquad \dot{p}_\eps^j = - \frac{\partial E_\eps}{\partial y^j_\eps},\end{equation*}
for \(\lambda=1,\ldots, r\) and \( j= 1, \ldots, n\). After some calculations, we find that these
equations are given by
\begin{IEEEeqnarray}{rCl}
    \eqlabel{eq:8}
    \IEEEyesnumber\IEEEyessubnumber*
    \IEEEnonumber
    \dot{\phi}^\lambda_\eps &=& \omega_\lambda(y_\eps)+\eps \sum_{j=1}^n \frac{p^j_\eps
        \cdot \partial_j \omega_\lambda(y_\eps)}{2\omega_\lambda(y_\eps)}\sin(2\eps^{-1}\phi^\lambda_\eps)\\
    \label{eq:8a}
    && +\frac{\eps^2}{8}\sum_{j=1}^n\sum_{\mu=1}^r\frac{\theta^\mu_\eps \cdot \partial_j \omega_\mu(y_\eps) \cdot \partial_j \omega_\lambda(y_\eps) }{\omega_\mu(y_\eps)\omega_\lambda(y_\eps)}\left( \cos\left( 2\eps^{-1}\left( \phi_\eps^\mu - \phi^\lambda_\eps \right) \right) - \cos\left(2\eps^{-1}\left( \phi_\eps^\mu + \phi_\eps^\lambda \right)\right) \right),\\
    \IEEEnonumber
    \dot{\theta}^\lambda_\eps &=& -\sum_{j=1}^n \frac{\theta^\lambda_\eps p^j_\eps\cdot \partial_j \omega_\lambda(y_\eps)}{\omega_\lambda(y_\eps)}\cos\left( 2\eps^{-1}\phi^\lambda_\eps \right) \\
    \label{eq:8b}
    && - \frac{\eps}{4} \sum_{j=1}^n\sum_{\mu=1}^r \frac{\theta^\mu_\eps \theta^\lambda_\eps \cdot \partial_j \omega_\mu(y_\eps) \cdot \partial_j \omega_\lambda(y_\eps)}{\omega_\mu(y_\eps)\omega_\lambda(y_\eps)}\left( \sin\left( 2\eps^{-1}\left( \phi_\eps^\mu -\phi_\eps^\lambda \right) \right) + \sin\left( 2\eps^{-1}\left( \phi_\eps^\mu +\phi_\eps^\lambda \right) \right) \right),\\
    \label{eq:8c}
    \dot{y}^j_\eps &=& p^j_\eps + \eps  \sum_{\lambda=1}^r \frac{\theta^\lambda_\eps \cdot \partial_j \omega_\lambda(y_\eps)}{2\omega_\lambda(y_\eps)}\sin\left( 2\eps^{-1}\phi^\lambda_\eps \right),\\
    \IEEEnonumber
    \dot{p}^j_\eps &=& -\partial_j V(y_\eps)-\sum_{\lambda=1}^r \theta^\lambda_\eps \cdot\partial_j \omega_\lambda(y_\eps) - \eps \sum_{k=1}^n\sum_{\lambda=1}^r \frac{\theta^\lambda_\eps p^k_\eps }{2}\left(  \frac{\partial_j \partial_k \omega_\lambda(y_\eps)}{\omega_\lambda(y_\eps)}- \frac{ \partial_k \omega_\lambda(y_\eps) \cdot \partial_j \omega_\lambda(y_\eps)}{\omega_\lambda^2(y_\eps)} \right)\sin\left( 2\eps^{-1}\phi^\lambda_\eps \right)\IEEEeqnarraynumspace\\
    \IEEEnonumber
    && -\frac{\eps^2}{8} \sum_{k=1}^n  \sum_{\lambda=1}^r\sum_{\mu=1}^r \frac{\theta^\lambda_\eps \theta^\mu_\eps \cdot \partial_k \omega_\mu(y_\eps)}{\omega_\mu(y_\eps)} \left( \frac{\partial_j \partial_k \omega_\lambda(y_\eps)}{\omega_\lambda(y_\eps)} -\frac{ \partial_k \omega_\lambda(y_\eps) \cdot \partial_j \omega_\lambda(y_\eps)}{\omega_\lambda^2(y_\eps)}\right)\\
    \label{eq:8d}
    &&  \times \left( \cos\left(2\eps^{-1} \left(  \phi^\mu_\eps - \phi^\lambda_\eps \right)\right) - \cos\left(2\eps^{-1} \left(  \phi^\mu_\eps + \phi^\lambda_\eps \right)\right)\right).
\end{IEEEeqnarray}
The initial values as given in~\eqref{eq:7} transform to
\begin{equation}
    \label{eq:9}
    \phi_\eps(0)=0,\qquad  \theta_\eps^\lambda(0)= \theta_\ast^\lambda
    = \frac{\vert u_\ast^\lambda \vert^2}{2\omega_\lambda(y_\ast)},\qquad y_\eps(0)=y_\ast,\qquad p_\eps(0)=p_\ast.
\end{equation}

\subsection{Existence and uniqueness of a solution to the transformed model
    problem}

\label{sec:Exist-uniq-solut}

Let us denote the right-hand side of~\eqref{eq:8} as
\(\mathcal{F}_\eps \colon \mathbb{R}^{2m}\to \mathbb{R}^{2m}\). By assumption
\(\omega_\lambda\in C^\infty(\R^n)\) for \(\lambda = 1, \ldots, r\) and therefore
\(\mathcal{F}_\eps\in C^\infty(\R^{2m},\R^{2m})\) for \(0<\eps<\eps_0< \infty\). In particular,
\(\mathcal{F}_\eps\) is locally Lipschitz continuous. Hence, by the standard existence
and uniqueness theory
for ordinary differential equations (see for example~\cite{Teschl2012}), there
exists a \(T >0\) such that for
fixed \(0<\eps<\eps_0<\infty\) the initial value problem~\eqref{eq:8}--\eqref{eq:9} has a
unique solution
\begin{equation}
    \label{eq:10}
    (\phi_\eps, \theta_\eps,y_\eps, p_\eps)\in C^\infty([0,T],\R^{2m}).
\end{equation}

\section{Asymptotic expansion}

\label{sec:5}

In this section, we rigorously derive the second-order asymptotic expansion of \(\phi_\eps\),
\(\theta_\eps\), \(y_\eps\) and \(p_\eps\). We will see, that the leading-order
expansion follows directly
from the evolution equations~\eqref{eq:8}. To simplify these equations for the
subsequent analysis, we
introduce in Section~\ref{subsec:1} some suitable new notation. We then derive the
first- and second-order
asymptotic expansion in Section~\ref{subsec:2}.

\subsection{Leading-order expansion}

\label{sec:Lead-order-expansion}

We consider a sequence of solutions~\eqref{eq:10} for \(\eps\to 0\). The
right-hand side of the evolution
equations~\eqref{eq:8} is oscillatory and has rapidly oscillating terms of
leading-order. As a consequence,
the sequences \(\{\dot{\phi}_\eps\}\) and \(\{\theta_\eps\}\) are bounded in \(C^{0,1}([0,T], \R^r)\), and the
sequences \(\{\dot{y}_\eps\}\) and \(\{\dot{p}_\eps\}\) are bounded in \(C^{0,1}([0,T], \R^n)\), while
sequences of higher-order derivatives (in particular \(\{\ddot{\theta}_\eps\}\), which will
thus require
special attention in the later part of this analysis) become unbounded as
\(\eps\to 0\). It follows from the
extended Arzel\`a--Ascoli theorem~\cite[Chapter~I~\textsection 1]{Bornemann} that
we can
extract a subsequence, not
relabelled, and functions \(\theta_0 \in C^{0,1}([0,T], \R^r)\), \({\phi_0\in C^{1,1}([0,T], \R^r)}\) and
\(y_0, p_0 \in C^{1,1}([0,T], \R^n)\), such that
\begin{IEEEeqnarray}{rCl+rCl}
    \eqlabel{eq:13}
    \IEEEyesnumber\IEEEyessubnumber*
    \phi_\eps \to \phi_0& \quad \text{in} \quad& C^1([0,T],\R^r),& \ddot{\phi}_\eps \weak{\ast}\ddot{\phi}_0& \quad \text{in}
    \quad& L^\infty([0,T],\R^r),\\
    \theta_\eps \to \theta_0& \quad \text{in} \quad& C([0,T],\R^r),& \dot{\theta}_\eps \weak{\ast}\dot{\theta}_0& \quad \text{in}
    \quad& L^\infty([0,T],\R^r),\\
    y_\eps \to y_0& \quad \text{in} \quad& C^1([0,T],\R^n),&\ddot{y}_\eps \weak{\ast}\ddot{y}_0& \quad \text{in} \quad&
    L^\infty([0,T],\R^n),\\
    p_\eps \to p_0& \quad \text{in} \quad& C^1([0,T],\R^n),& \ddot{p}_\eps \weak{\ast}\ddot{p}_0& \quad \text{in} \quad&
    L^\infty([0,T],\R^n).
\end{IEEEeqnarray}
By taking the limit \(\eps\to0\) in Equations~\eqref{eq:8a},~\eqref{eq:8c} and~\eqref{eq:8d} and the
weak\(^\ast\) limit in~\eqref{eq:8b} we deduce that
\begin{equation*}
    \dot{\phi}^\lambda_0 = \omega_\lambda(y_0),\qquad \dot{\theta}^\lambda_0 =0,\qquad \dot{y}^j_0 = p_0^j,
    \qquad \dot{p}^j_0 = -\partial_j V(y_0)-\sum_{\lambda=1}^r\theta_\ast^\lambda \cdot\partial_j \omega_\lambda(y_0),\end{equation*}
for \(\lambda=1,\ldots,r\) and \(j=1,\ldots, n\), and in particular that
\(\theta_0^\lambda \equiv \theta_\ast^\lambda\) (compare with~\eqref{eq:9}). Moreover, since the right-hand
side of the limit equation
\begin{equation*}
    \ddot{y}^j_0 = -\partial_j V(y_0)-\sum_{\lambda=1}^r\theta_\ast^\lambda \cdot\partial_j \omega_\lambda(y_0)\end{equation*}
does not depend on a chosen subsequence, we can discard the extraction of a
subsequence altogether (see~\cite[Principle~5, Chapter~I \textsection 1]{Bornemann}). Note that the
above
convergence results extend Theorem~\ref{thm:1}.

\subsection{Reformulation of the governing equations}

\label{subsec:1}

It will be convenient to introduce some notation to simplify the system of
differential
equations~\eqref{eq:8}. To this end, we define for \(f\in C^\infty(\R^n)\), where \(f = f(y)\)
and
\(y\in C^\infty([0,T],\R^n)\), the expression
\begin{equation*}
    D_t^k D_j^l f \coloneqq \frac{d^k}{dt^k}\frac{\partial^l f}{\partial y_j^l},\end{equation*}
for \( k,l\in \mathbb{N}_0\) and \(j=1,\ldots, n\). We will often apply this notation in
combination with the function
\begin{equation*}
    L_\eps^\lambda \coloneqq \log (\omega_\lambda(y_\eps)),\end{equation*}
where \(\lambda = 1,\ldots, r\). Then, we can conveniently write, for instance,
\begin{equation}
    \label{eq:755}
    D L_\eps^\lambda = \sum_{j=1}^n D_j L_\eps^\lambda \cdot e_j \qquad \text{or} \qquad  D_t L_\eps^\lambda
    = \left\langle \dot{y}_\eps, D L_\eps^\lambda \right\rangle = \sum_{j=1}^n \dot{y}_\eps^j \cdot D_j L_\eps^\lambda,
\end{equation}
with \(e_j\) as the \(j\)th standard basis vector in \(\R^n\).
With these definitions, the equations
in~\eqref{eq:8} read

\begin{IEEEeqnarray}{rCl}
    \label{eq:11}
    \IEEEyesnumber\IEEEyessubnumber*
    \IEEEnonumber
    \dot{\phi}_\eps^\lambda &=& \omega_\lambda(y_\eps) + \frac{\eps}{2}\left\langle p_\eps, D L_\eps^\lambda \right\rangle\sin(2\eps^{-1}\phi_\eps^\lambda)\\
    \label{eq:11a}
    &&+\>\frac{\eps^2}{8}\sum_{\mu=1}^r \theta_\eps^\mu \left\langle DL_\eps^\lambda, DL_\eps^\mu \right\rangle \left( \cos\left( 2\eps^{-1}\left( \phi_\eps^\mu - \phi^\lambda_\eps \right) \right) - \cos\left(2\eps^{-1}\left( \phi_\eps^\mu + \phi_\eps^\lambda \right)\right) \right),\\
    \IEEEnonumber
    \dot{\theta}_\eps^\lambda&=& - \theta_\eps^\lambda \left\langle p_\eps, DL_\eps^\lambda \right\rangle \cos(2\eps^{-1}\phi_\eps^\lambda)\\
    \label{eq:11b}
    &&-\>\frac{\eps}{4}\sum_{\mu=1}^r \theta_\eps^\mu \theta_\eps^\lambda \left\langle DL_\eps^\lambda, DL_\eps^\mu \right\rangle \left( \sin\left( 2\eps^{-1}\left( \phi_\eps^\mu -\phi_\eps^\lambda \right) \right) + \sin\left( 2\eps^{-1}\left( \phi_\eps^\mu +\phi_\eps^\lambda \right) \right) \right),\\
    \label{eq:11c}
    \dot{y}_\eps^j &=& p_\eps^j+\frac{\eps}{2}\sum_{\lambda=1}^n \theta_\eps^\lambda \cdot D_j L_\eps^\lambda\sin(2\eps^{-1}\phi_\eps^\lambda),\\
    \IEEEnonumber
    \dot{p}_\eps^j &=& - D_j V(y_\eps) - \sum_{\lambda=1}^r \theta_\eps^\lambda\cdot D_j \omega_\lambda(y_\eps)-\frac{\eps}{2} \sum_{\lambda=1}^r \theta_\eps^\lambda \left\langle p_\eps, DD_jL_\eps^\lambda \right\rangle \sin(2\eps^{-1}\phi_\eps^\lambda)\\
    \label{eq:11d}
    &&-\> \frac{\eps^2}{8}\sum_{\lambda=1}^r \sum_{\mu=1}^r \theta_\eps^\lambda\theta_\eps^\mu \left\langle DL_\eps^\mu, DD_j L_\eps^\lambda \right\rangle \left( \cos\left(2\eps^{-1} \left(  \phi^\mu_\eps - \phi^\lambda_\eps \right)\right) - \cos\left(2\eps^{-1} \left(  \phi^\mu_\eps + \phi^\lambda_\eps \right)\right)\right).
\end{IEEEeqnarray}
Moreover, solving~\eqref{eq:11c} with respect to \(p_\eps^j\) and inserting the
result
into~\eqref{eq:11a},~\eqref{eq:11b} and~\eqref{eq:11d} brings the equations of motion to
their final form
\begin{IEEEeqnarray}{rCl}
    \eqlabel{eq:12}
    \IEEEyesnumber\IEEEyessubnumber*
    \label{eq:12a}
    \dot{\phi}_\eps^\lambda &=& \omega_\lambda(y_\eps) +\frac{\eps}{2} D_t L_\eps^\lambda \sin(2\eps^{-1}\phi_\eps^\lambda),\\
    \label{eq:12b}
    \dot{\theta}_\eps^\lambda &=& -\theta_\eps^\lambda \cdot D_t L_\eps^\lambda \cos(2\eps^{-1}\phi_\eps^\lambda),\\
    \label{eq:12c}
    \dot{y}_\eps^j &=& p_\eps^j +\frac{\eps}{2} \sum_{\lambda=1}^r \theta_\eps^\lambda
    \cdot D_j L_\eps^\lambda \sin(2\eps^{-1}\phi_\eps^\lambda),\\
    \label{eq:12d}
    \dot{p}_\eps^j &=& -D_j V(y_\eps)-\sum_{\lambda=1}^r \theta_\eps^\lambda \cdot D_j \omega_\lambda(y_\eps)
    - \frac{\eps}{2} \sum_{\lambda=1}^r \theta_\eps^\lambda \cdot D_t D_j L_\eps^\lambda \sin(2\eps^{-1}\phi_\eps^\lambda).
\end{IEEEeqnarray}

\subsection{First- and second-order expansion}

\label{subsec:2}

We now define functions that will appear throughout this work and then state the
first main
result~\ref{thm:p1}.
\begin{definition}
    \label{def:1}
    Let \((\phi_\eps, \theta_\eps, y_\eps, p_\eps)\) be the solution to~\eqref{eq:8}--\eqref{eq:9} and
    \((\phi_0,\theta_0,y_0,p_0)\) be as in~\eqref{eq:13}. With Assumption~\ref{ass:1} and the notation
    introduced above we define for \(\lambda=1,\ldots,r\) and \(j=1,\ldots,n\) the functions
    \begin{equation*}
        \theta_1^{\lambda\eps}\coloneqq\frac{\theta_\eps^\lambda -\theta_\ast^\lambda}{\eps},\qquad
        \phi_2^{\lambda\eps}\coloneqq \frac{\phi_\eps^\lambda-\phi_0^\lambda}{\eps^2},\qquad
        y_2^{j\eps}\coloneqq \frac{y_\eps^j - y_0^j}{\eps^2},\qquad p_2^{j\eps}\coloneqq\frac{p_\eps^j
            - p_0^j}{\eps^2},\qquad\theta_2^{\lambda\eps} \coloneqq \frac{\theta_1^{\lambda\eps} -[\theta_1^\lambda]^\eps}{\eps},
    \end{equation*}
    \begin{IEEEeqnarray*}{rCl+rCl}
        [\theta_1^\lambda]^\eps &\coloneqq& -\frac{\theta_\ast^\lambda\cdot
            D_t L_0^\lambda}{2\omega_\lambda(y_0)}\sin(2\eps^{-1}\phi^\lambda_0),& [\phi^\lambda_2]^\eps
        &\coloneqq& -\frac{D_t L_0^\lambda}{4\omega_\lambda(y_0)}\cos(2\eps^{-1}\phi_0^\lambda),\\
        \phantom{}[y^j_2]^\eps&\coloneqq& -\sum_{\lambda=1}^r \frac{\theta_\ast^\lambda \cdot D_j
            L_0^\lambda}{4\omega_\lambda(y_0)}\cos(2\eps^{-1}\phi^\lambda_0),&  [p^j_2]^\eps &\coloneqq& \sum_{\lambda=1}^r \frac{d}{dt}\left( \frac{\theta_\ast^\lambda \cdot D_j L_0^\lambda}{4\omega_\lambda(y_0)} \right)\cos(2\eps^{-1}\phi^\lambda_0)
    \end{IEEEeqnarray*}
    and
    \begin{IEEEeqnarray*}{rCl}
        [\theta_2^\lambda]^\eps &\coloneqq& \sum_{\mu=1}^r\frac{\theta_\ast^\lambda \theta_\ast^\mu \left\langle D\omega_\mu(y_0), DL_0^\lambda \right\rangle }{4\omega^2_\lambda(y_0)}\cos(2\eps^{-1}\phi^\lambda_0)- \frac{\theta_\ast^\lambda\left\langle D^2L_0^\lambda \dot{y}_0, \dot{y}_0 \right\rangle }{4\omega^2_\lambda(y_0)}\cos(2\eps^{-1}\phi^\lambda_0)+\frac{\theta_\ast^\lambda (D_t L_0^\lambda)^2}{4\omega^2_\lambda(y_0)}\cos(2\eps^{-1}\phi^\lambda_0)\\
        && + \> \frac{(\theta_\ast^\lambda)^2 \vert DL_0^\lambda \vert^2}{16 \omega_\lambda(y_0)}\cos(4\eps^{-1}\phi^\lambda_0)-\frac{\theta_\ast^\lambda \cdot D_t L_0^\lambda}{\omega_\lambda(y_0)}\bar{\phi}_2^\lambda \cos(2\eps^{-1}\phi^\lambda_0)+\frac{\theta_\ast^\lambda \left\langle DV(y_0), DL_0^\lambda \right\rangle }{4\omega_\lambda^2(y_0)}\cos(2\eps^{-1}\phi_0^\lambda)\\
        && +\> \sum_{\substack{\mu=1\\\mu \neq \lambda}}^r \frac{\theta_\ast^\lambda \theta_\ast^\mu \left\langle DL_0^\mu, DL_0^\lambda \right\rangle }{8}\left\{ \frac{\cos(2\eps^{-1}(\phi^\mu_0-\phi^\lambda_0))}{\omega_\mu(y_0)-\omega_\lambda(y_0)} + \frac{\cos(2\eps^{-1}(\phi^\mu_0+\phi^\lambda_0))}{\omega_\mu(y_0)+\omega_\lambda(y_0)}\right\}.
    \end{IEEEeqnarray*}
\end{definition}

\begin{theorem}\label{thm:2}The functions specified in Definition~\ref{def:1} satisfy
    \begin{IEEEeqnarray}{rClCl+rClCl}
        \label{eq:19}
        \theta_1^{\eps} - [\theta_1]^\eps &\to& 0& \quad \text{in} \quad& C([0,T],\R^r),& \frac{d}{dt}\left(  \theta_1^{\eps}
        - [\theta_1]^\eps \right)&\weak{\ast}& 0& \quad \text{in} \quad& L^\infty([0,T],\R^r),\\
        \label{eq:20}
        \phi_2^{\eps} -[\phi_2]^\eps &\to& \bar{\phi}_2& \quad \text{in} \quad& C([0,T],\R^r),& \frac{d}{dt}\left( \phi_2^{\eps}
        -[\phi_2]^\eps \right)& \weak{\ast}& \frac{d\bar{\phi}_2}{dt}& \quad \text{in} \quad& L^\infty([0,T],\R^r),\\
        \label{eq:21}
        y_2^{\eps} - [y_2]^\eps &\to& \bar{y}_2& \quad \text{in}\quad& C([0,T],\R^n),& \frac{d}{dt}\left( y_2^{\eps}
        - [y_2]^\eps \right)& \weak{\ast}& \frac{d \bar{y}_2}{dt}& \quad \text{in} \quad& L^\infty([0,T],\R^n),\\
        \label{eq:22}
        p_2^{\eps} - [p_2]^\eps &\to& \bar{p}_2& \quad \text{in} \quad& C([0,T],\R^n),& \frac{d}{dt}\left( p_2^{\eps}
        - [p_2]^\eps \right)&\weak{\ast}& \frac{d \bar{p}_2}{dt}& \quad \text{in} \quad& L^\infty([0,T],\R^n)
    \end{IEEEeqnarray}
    and
    \begin{equation}
        \label{eq:26}
        \theta_2^{\eps} -[\theta_2]^\eps\to \bar{\theta}_2 \quad \text{in} \quad C([0,T],\R^r),
    \end{equation}
    where \((\bar{\phi}_2, \bar{\theta}_2, \bar{y}_2, \bar{p}_2)\) is the unique solution to the inhomogeneous
    linear system of differential equations
    \begin{IEEEeqnarray*}{rCl}
        \eqlabel{eq:27}
        \IEEEyesnumber\IEEEyessubnumber*
        \label{eq:27a}
        \frac{d \bar{\phi}^\lambda_2}{dt} &=&  \left\langle  D\omega_\lambda(y_0),\bar{y}_2\right\rangle + \frac{\theta_\ast^\lambda\vert D_y L^\lambda_0 \vert^2}{8} - \frac{(D_t L^\lambda_0)^2}{8 \omega_\lambda(y_0)},\\
        \label{eq:27b}
        \frac{d \bar{\theta}^\lambda_2}{dt} &=& \frac{d}{dt}\frac{\theta_\ast^\lambda(D_t L_0^\lambda)^2}{8\omega_\lambda^2(y_0)},\\
        \label{eq:27c}
        \frac{d \bar{y}^j_2}{dt} &=& \bar{p}^j_2 - \sum_{\lambda=1}^r\frac{\theta^\lambda_\ast \cdot D_j L^\lambda_0 \cdot D_t L^\lambda_0 }{4\omega_\lambda(y_0)},\\
        \IEEEnonumber
        \frac{d \bar{p}^j_2}{dt} &=&-\left\langle \bar{y}_2, DD_j V(y_0) \right\rangle - \sum_{\lambda=1}^r \bar{\theta}_2^\lambda \cdot D_j \omega_\lambda(y_0) - \sum_{\lambda=1}^r\theta_\ast^\lambda \left\langle \bar{y}_2, D D_j \omega_\lambda(y_0) \right\rangle\\
        \label{eq:27d}
        &&- \>\sum_{\lambda=1}^r\frac{\left( \theta_\ast^\lambda \right)^2 \left\langle D L_0^\lambda, D D_j L_0^\lambda \right\rangle}{8} +\sum_{\lambda=1}^r\frac{\theta_\ast^\lambda \cdot D_t D_j L_0^\lambda \cdot D_t L_0^\lambda}{4\omega_\lambda(y_0)},
    \end{IEEEeqnarray*}
    for \(\lambda=1,\ldots,r\) and \(j=1,\ldots,n\), with \(\eps\)-independent initial values
    \begin{equation}
        \label{eq:30}
        \bar{\phi}_2(0)=-[\phi_2]^\eps(0),\qquad\bar{\theta}_2(0) = -[\theta_2]^\eps(0),\qquad \bar{y}_2(0)
        =-[y_2]^\eps(0), \qquad\bar{p}_2(0) = -[p_2]^\eps(0).
    \end{equation}
\end{theorem}

\subsection{Proof of Theorem~\ref{thm:2}}
\label{sec:6}
The proof of Theorem~\ref{thm:2} will use the following Lemmas~\ref{lemma:1}
to~\ref{lemma:10}.  We start by sketching the
general strategy of the proof.

Theorem~\ref{thm:2} states that the first- and second-order asymptotic expansions of
\(\phi_\eps\),
\(\theta_\eps\), \(y_\eps\) and \(p_\eps\) can be decomposed into rapidly
oscillating terms
\([\theta_1]^\eps\), \([\phi_2]^\eps\), \([y_2]^\eps\), \([p_2]^\eps\) and \([\theta_2]^\eps\), which
converge
weakly\(^\ast\) to zero, and slowly evolving terms \(\bar{\phi}_2\),
\(\bar{y}_2\), \(\bar{p}_2\) and
\(\bar{\theta}_2\), which describe the average motion of the second-order expansions
and are given as the
solution to an inhomogeneous linear system of ordinary differential equations.

To derive these second-order asymptotic expansions, we specified in Definition~\ref{def:1}
the scaled
first-order residual function \(\theta_1^\eps\) and the scaled second-order residual
functions
\(\phi_2^\eps\), \(y_2^\eps\), \(p_2^\eps\) and \(\theta_2^\eps\) by subtracting the
leading- and first-order
asymptotic expansion terms from the original solution to the model problem and
by scaling these residual terms
to appropriate order. The functions \(\phi_2^\eps\), \(y_2^\eps\), \(p_2^\eps\) and
\(\theta_2^\eps\) carry
all the information about the system's second-order asymptotic expansion in
their leading-order expression.
We thus analyse the limit \(\eps\to0\) of these terms.

In the proof of Theorem~\ref{thm:2}, we will repeatedly integrate by parts, which
requires us to regularly
divide by \(\dot{\phi}_\eps^\lambda\) and \(\dot{\phi}_\eps^\lambda - \dot{\phi}_\eps^\mu\)
\((\lambda\neq\mu)\).  Lemma~\ref{lemma:1} ensures that the resulting terms are well-defined,
provided that the
scale parameter \(\eps\) is small enough.

As the model problem is highly oscillatory, the interacting degrees of freedom
can exhibit resonances of
different types. Lemmas~\ref{lemma:2},~\ref{lemma:7} and~\ref{lemma:3} clarify how the interaction of
a
generic function \(u_\eps\) with a rapidly oscillating function \(\exp(i\eps^{-1}\psi_\eps)\)
affects their
interaction in the limit \(\eps\to0\). Here, \(u_\eps\) and \(\psi_\eps\) are
representatives of functions that
appear throughout the proof of Theorem~\ref{thm:2}. Lemmas~\ref{lemma:2},~\ref{lemma:7} and~\ref{lemma:3}
are
used in the derivation of the weak\(^\ast\) limit of specific rapidly
oscillating functions under the
non-resonance Assumptions~\ref{ass:1} and~\ref{ass:2}.

Similarly,  Lemma~\ref{lemma:5} provides information about the uniform convergence of
the term
\(u_\eps \exp(i\eps^{-1}\psi_\eps)- u_0 \exp(i\eps^{-1}\psi_0)\), which is a representation of functions that appear
throughout the proof of Theorem~\ref{thm:2}. Here, \(u_\eps \exp(i \eps^{-1}\psi_\eps)\) is rapidly
oscillating at leading-order. By subtracting the leading-order term
\(u_0 \exp(i\eps^{-1}\psi_0)\), their
difference converges uniformly under certain convergence assumptions on
\(u_\eps\) and \(\psi_\eps\).

In Lemmas~\ref{lemma:4} and~\ref{lemma:6} we show that the sequences of scaled residual
functions
\(\{\theta_1^\eps\}\), \(\{\phi_2^\eps\}\) and \(\{\theta_2^\eps\}\) are bounded in \(L^\infty([0,T], \R^r)\),
and \(\{y_2^\eps\}\) and \(\{p_2^\eps\}\) are bounded in \(L^\infty([0,T], \R^n)\). This is a
necessary
prerequisite for the analysis of the first- and second-order asymptotic
expansion.

In general, the rapidly oscillating terms \([\theta_1]^\eps\), \([\phi_2]^\eps\),
\([y_2]^\eps\),
\([p_2]^\eps\) and \([\theta_2]^\eps\), which do \emph{not} converge in the limit
\(\eps\to0\), can be found
through integration by parts. To find the evolution equation for the averaged
second-order expansion terms
\(\bar{\phi}_2\), \(\bar{y}_2\), \(\bar{p}_2\) and \(\bar{\theta}_2\), we analyse in Lemmas~\ref{lemma:8} and~\ref{lemma:9} the time derivatives of the terms \(\phi_2^\eps - [\phi_2]^\eps\),
\(y_2^\eps - [y_2]^\eps \), \(p_2^\eps - [p_2]^\eps\) and \(\theta_2^\eps - [\theta_2]^\eps\). They carry
information about the time derivative of \(\bar{\phi}_2\), \(\bar{y}_2\),
\(\bar{p}_2\) and \(\bar{\theta}_2\)
in their leading-order asymptotic expansion. Alaoglu's theorem~\cite[Principle
    3]{Bornemann} and the extended
Arzel\`a--Ascoli theorem~\cite[Principle 4]{Bornemann} justify the extraction
of a subsequence such that in
the weak\(^\ast\) limit an evolution equation for the average dynamics at
second-order emerges. However, since
the evolution equation has a unique solution,  Lemma~\ref{lemma:10} implies that the
extraction of a
subsequence can be discarded altogether, meaning the limit holds for the whole
sequence.

The following lemmas collectively proof Theorem~\ref{thm:2}. They are stated
separately for reference but
should be understood in the context of Theorem~\ref{thm:2}. As mentioned earlier, the
problem presented in
Section~\ref{sec:8} extends the model in~\cite{Klar2020}. More precisely, it
generalises the model
in~\cite{Klar2020} in two ways. Firstly, by describing the interaction of
\(r\) fast and \(n\) slow degrees of
freedom \((n,r\in \mathbb{N})\) instead of the interaction of one fast and one slow degree
of freedom. This
requires us to impose certain non-resonance conditions. Secondly, the model in
this article includes a slow
potential \(V=V(y)\) which is absent in~\cite{Klar2020}. These
generalisations make the following proof much
more involved, yet it mimics at its core the proof as presented
in~\cite{Klar2020}. As such, some of the
following preparatory lemmas, with model-specific alterations, can be found
in~\cite{Klar2020}. Nevertheless,
we will state and prove these lemmas here for the reader's convenience.

\begin{lemma}[Similar to Lemma~3.4 in~\cite{Klar2020}]
    \label{lemma:1}
    There exist constants \(0<C<\infty\) and \(0<\eps_0<\infty\) where
    \(\eps_0 = \eps_0(\phi_\ast, \theta_\ast, y_\ast, p_\ast, \omega,C)\) such that
    \(0 < C \leq \dot{\phi}^\lambda_\eps\) for \(\lambda=1,\ldots, r\) and
    \(0< C\leq \vert \dot{\phi}_\eps^\lambda - \dot{\phi}_\eps^\mu\vert\) for \(\lambda, \mu=1,\ldots,r\),
    \(\lambda\neq \mu\), for all \(0 < \eps <\eps_0\) small enough.
\end{lemma}
\begin{proof}
    The claim follows directly from Assumption~\ref{ass:1} and~\eqref{eq:12a}.
\end{proof}

\begin{remark}Henceforth, we assume that \(0<\eps<\eps_0\) is small enough so that the statements
    of  Lemma~\ref{lemma:1}
    apply.
\end{remark}

\begin{lemma}[Lemma 3.5 in~\cite{Klar2020}]
    \label{lemma:2}
    Let \(\{u_\eps\}\) be a bounded sequence in \(C^{0,1}([0,T])\) and \(\{\psi_\eps\}\) be a
    bounded sequence
    in \(C^{1,1}([0,T])\) with \(0<C\leq\dot{\psi}_\eps\). Then, for all \(a,b\in[0,T]\):
    \begin{IEEEeqnarray*}{rClrCl}
        \int_a^b u_\eps \sin(\eps^{-1}\psi_\eps)\ud t&=& \mathcal{O}(\eps),\qquad \int_a^b u_\eps \cos(\eps^{-1}\psi_\eps)\ud t
        &=& \mathcal{O}(\eps).
    \end{IEEEeqnarray*}
\end{lemma}

\begin{proof}Integration by parts gives for \(0< \eps<\eps_0\) small enough
    \begin{equation*}
        \left|  \int_a^b u_\eps\exp\left( \frac{i\psi_\eps}{\eps}\right)\ud t\right|   \leq\eps
        \left|   \frac{u_\eps(a)}{ \dot{\psi}_\eps(a)} \right|   + \eps \left|   \frac{u_\eps(b)}{\dot{\psi}_\eps(b)} \right|
        + \eps \left|  \int_a^b \frac{d}{dt} \left( \frac{u_\eps}{\dot{\psi}_\eps} \right)
        \exp\left( \frac{i\psi_\eps}{\eps}\right) \ud t \right|   = \mathcal{O}(\eps).
    \end{equation*}
    The claim follows by considering the real and imaginary parts separately and the
    isometric isomorphism \(C^{k-1,1}([0,T])\cong W^{k,\infty}([0,T])\) (see~\cite[p.~154]{Gilbarg1983}).
\end{proof}

\begin{lemma}
    \label{lemma:7}
    Let \(u_0 \in C^2([0,T])\) and \(\psi_0\in C^3([0,T])\). Let \(\{u_\eps\}\) be a sequence in \(C^2([0,T])\)
    and \(\{\psi_\eps\}\) be a sequence in \(C^3([0,T])\) such that the sequences
    \(\{\eps^{-1}(u_\eps - u_0)\}\), \(\{\eps^{-2}(\psi_\eps -\psi_0)\}\) are bounded in
    \(L^\infty([0,T])\). Moreover, let \(t_i\in[0,T]\) be an impact point in time (see Assumption~\ref{ass:2})
    with \(\dot{\psi}_0(t_i)=0\) and \(\ddot{\psi}_0(t_i)\neq 0\).  Then, for all \(a,b\in [0,T]\):
    \begin{equation*}
        \int_a^b u_\eps \sin(\eps^{-1}\psi_\eps)\ud t= \mathcal{O}(\eps^{1/2}),\qquad
        \int_a^b u_\eps \cos(\eps^{-1}\psi_\eps)\ud t= \mathcal{O}(\eps^{1/2}).
    \end{equation*}
\end{lemma}

\begin{proof}
    We treat the real and imaginary parts separately and write
    \begin{equation*}
        u_\eps \exp(i\eps^{-1}\psi_\eps) = (u_\eps -u_0)\exp(i\eps^{-1}\psi_\eps)
        -u_0 \exp(i\eps^{-1}\psi_0)\left( 1-\exp\left( i\eps^{-1}\left( \psi_\eps-\psi_0 \right) \right) \right)
        +u_0\exp(i\eps^{-1}\psi_0).
    \end{equation*}
    Since the sequences \(\{\eps^{-1}(u_\eps - u_0)\}\) and \(\{\eps^{-2}(\psi_\eps - \psi_0)\}\) are bounded in
    \(L^\infty([0,T])\), the claim is satisfied for the first two terms on the right-hand side. Moreover, let
    \(\eta \in C^\infty_0([0,T])\) and \(U_{t_i},V_{t_i}\) be small neighbourhoods around \(t_i\) such
    that
    \(\supp \eta = V_{t_i}\), \(U_{t_i}\subset V_{t_i}\) and \(\eta=1\) in \(U_{t_i}\) and write
    \begin{equation*}
        \int_a^b u_0 \exp(i\eps^{-1}\psi_0)\ud t=\int_{V_{t_i}} u_0 \exp(i\eps^{-1}\psi_0)\eta\ud t
        + \int_{[a,b]\setminus U_{t_i}} u_0 \exp(i\eps^{-1}\psi_0)(1-\eta)\ud t.
    \end{equation*}
    For the second integral we can apply  Lemma~\ref{lemma:2} since \(t_i \notin [a,b]\setminus U_{t_i}\), and
    obtain an error of order \(\mathcal{O}(\eps)\). For the first integral we use the method
    of stationary phase
    to derive
    \begin{equation*}
        \int_{V_{t_i}} u_0 \exp(i\eps^{-1}\psi_0)\eta\ud t= \mathcal{O}(\eps^{1/2}).
    \end{equation*}
    A detailed description of the method of stationary phase can be found, for
    example, in~\cite[\textsection 1,
        Proposition~3]{Stein1993}, where smoothness of \(u_0\) and \(\psi_0\) is
    assumed.  Here, we are only
    interested in the leading-order asymptotics, for which \(u_0\in C^2([0,T])\) and
    \(\psi_0\in C^3([0,T])\) is
    sufficient.
\end{proof}

\begin{lemma}[Generalisation of Lemma 3.6 in~\cite{Klar2020}]
    \label{lemma:3}
    Let \(u\in C^2(\mathbb{R}^r_{>0}\times\mathbb{R}^{r+2n})\) and \((\phi_\eps, \theta_\eps, y_\eps,p_\eps)\)
    be the solution to~\eqref{eq:8}--\eqref{eq:9}. Then, the sequence of functions
    \(\{u_\eps\}\) where
    \(u_\eps \coloneqq u(\dot{\phi}_\eps, \theta_\eps,\dot{y}_\eps, y_\eps)\) satisfies for all \(a,b\in[0,T]\)
    and \(k=1,2\):
    \begin{equation*}
        \int_a^b\dot{u}_\eps \cos(2k\eps^{-1}\phi^\lambda_\eps)\ud t\to \frac{2-k}{2}
        \int_a^b D_t L_0^\lambda\left( \omega_\lambda(y_0)\cdot\partial_\lambda u_0
        - \theta_\ast^\lambda\cdot\partial_{r+\lambda}u_0 \right)+ \theta_\ast^\lambda\sum_{j=1}^n D_j
        \omega_\lambda(y_0)\cdot \partial_{2r+j} u_0\ud t
    \end{equation*}
    and
    \begin{equation*}
        \int_a^b \dot{u}_\eps\sin(2\eps^{-1}\phi^\lambda_\eps)\ud t = \mathcal{O}(\eps).
    \end{equation*}
\end{lemma}

\begin{proof}
    The equations in~\eqref{eq:12} imply
    \begin{IEEEeqnarray*}{rCl}
        \dot{u}_\eps&=&\sum_{\mu=1}^r \partial_\mu u_\eps \cdot \left( D_t \omega_\mu(y_\eps) + \frac{\eps}{2} D^2_t
        L_\eps^\mu \sin(2\eps^{-1}\phi^\mu_\eps) +  D_t L_\eps^\mu \cdot \dot{\phi}_\eps^\mu \cos(2\eps^{-1}\phi^\mu_\eps) \right)\\
        && -\>\sum_{\mu=1}^r \partial_{r+\mu} u_\eps\cdot \theta_\eps^\mu \cdot D_t L_\eps^\mu \cos(2\eps^{-1}\phi_\eps^\mu)\\
        && -\> \sum_{j=1}^n\partial_{2r+j}u_\eps \cdot \left(D_j V(y_\eps)+\sum_{\mu=1}^r \theta_\eps^\mu \cdot D_j
        \omega_\mu(y_\eps)-\sum_{\mu=1}^r \theta_\eps^\mu\cdot D_j \omega_\mu(y_\eps)\cos(2\eps^{-1}\phi^\mu_\eps)  \right)\\
        && +\> \sum_{j=1}^n \partial_{2r+n+j}u_\eps \cdot \dot{y}^j_\eps.
    \end{IEEEeqnarray*}
    The claim follows from the uniform convergence results in~\eqref{eq:13}, Lemmas~\ref{lemma:2}
    and~\ref{lemma:7} and the trigonometric identities
    \begin{equation*}
        2\cos(x)\cos(y) =\cos\left( x+y \right)+\cos\left( x-y \right), \qquad
        2\cos(x)\sin(y) = \sin\left(x+y \right)-\sin\left( x-y \right).
    \end{equation*}
\end{proof}

\begin{lemma}[Similar to Lemma~3.9 in~\cite{Klar2020}]
    \label{lemma:5}
    Let \(u_0,\psi_0 \in C^1([0,T])\) and let \(\{u_\eps\},\{\psi_\eps\}\) be sequences in \(C^1([0,T])\) such
    that the sequences \(\{\dot{u}_\eps\}\), \(\{\eps^{-1}\left( u_\eps - u_0 \right)\}\),
    \(\{\eps^{-2}\left( \psi_\eps - \psi_0 \right)\}\) and \(\{\eps^{-1}( \dot{\psi}_\eps - \dot{\psi}_0 )\}\)
    are bounded in \(L^\infty([0,T])\). Then, for \(v_\eps\) such that
    \begin{equation*}
        v_\eps \coloneqq u_\eps\exp(i\eps^{-1}\psi_\eps) - u_0\exp(i\eps^{-1}\psi_0),
    \end{equation*}
    the sequence \(\{\eps^{-1}v_\eps\}\) is bounded in \(L^\infty([0,T], \mathbb{C})\) and in particular
    \begin{equation*}
        v_\eps \to 0 \quad \text{in} \quad C([0,T], \mathbb{C}),\qquad \dot{v}_\eps \weak{\ast} 0 \quad \text{in}
        \quad L^\infty([0,T], \mathbb{C}).
    \end{equation*}
\end{lemma}

\begin{proof}
    By writing
    \begin{equation*}
        v_\eps = (u_\eps - u_0)\exp(i\eps^{-1}\psi_\eps)-u_0\exp(i\eps^{-1}\psi_0)\left(1-\exp\left(i\eps^{-1}\left(\psi_\eps-\psi_0\right)\right)\right)
    \end{equation*}
    and
    \begin{IEEEeqnarray*}{rCl}
        \dot{v}_\eps &=& \left( (\dot{u}_\eps - \dot{u}_0) + i\eps^{-1}\dot{\psi}_\eps(u_\eps - u_0)
        +i\eps^{-1}u_0(\dot{\psi}_\eps- \dot{\psi}_0)  \right)\exp(i\eps^{-1}\psi_\eps)\\
        && -\left( \dot{u}_0 + i\eps^{-1}\dot{\psi}_0u_0 \right)\left(1-\exp\left(i\eps^{-1}\left(\psi_\eps-\psi_0\right)\right)\right)\exp(i\eps^{-1}\psi_0),
    \end{IEEEeqnarray*}
    the assumptions imply that the sequences \(\{\eps^{-1}v_\eps\}\) and \(\{\dot{v}_\eps\}\) are
    bounded in
    \(L^\infty([0,T], \mathbb{C})\). This implies directly the uniform convergence of \(v_\eps\) to
    zero. The
    weak\(^\ast\) convergence of \(\dot{v}_\eps\) follows from~\cite[Principle
        1]{Bornemann}.
\end{proof}

\begin{lemma}
    \label{lemma:4}
    The sequences \(\{\theta_1^{\eps}\}\) and \(\{\phi_2^{\eps}\}\) are uniformly bounded in
    \(L^\infty([0,T],\R^r)\), and the sequences \(\{y_2^{\eps}\}\) and \(\{p_2^{\eps}\}\) are uniformly bounded
    in \(L^\infty([0,T],\R^n)\).
\end{lemma}

\begin{proof}
    In this proof, the constant \(0 < C < \infty\) depends on \(T\) but is
    independent of \(\eps\) and can take
    different values from line to line. Let \(t\in [0,T]\). For \(0 < \eps <\eps_0\) small enough and
    \(\lambda=1,\ldots,r\) let
    \begin{IEEEeqnarray*}{rCl}
        M^\lambda_1 &\coloneqq& \sup_{0<\eps<\eps_0} \sup_{h\in [0,1]}\left\Vert D
        \omega_\lambda((1-h)y_\eps +hy_0) \right\Vert_{L^\infty([0,T],\R^n)},\\
        M^\lambda_2 &\coloneqq& \sup_{0<\eps<\eps_0} \sup_{h\in [0,1]}\left\Vert D^2
        \omega_\lambda((1-h)y_\eps +hy_0) \right\Vert_{L^\infty([0,T],\R^{n\times n})},\\
        M_3 &\coloneqq& \sup_{0<\eps<\eps_0} \sup_{h\in [0,1]}\left\Vert D^2 V((1-h)y_\eps +hy_0)
        \right\Vert_{L^\infty([0,T],\R^{n\times n})}.
    \end{IEEEeqnarray*}
    For \(\lambda =1,\ldots,r\) and \(j=1,\ldots,n\) we apply  Lemma~\ref{lemma:2} to Equation~\eqref{eq:12a},
    \begin{equation}
        \label{eq:14}
        \left\vert \phi_2^{\lambda\eps}(t) \right\vert = \left\vert \int_0^t \dot{\phi}_2^{\lambda\eps}\ud s \right\vert
        \leq \frac{1}{\eps^2}\left\vert \int_0^t \omega_\lambda(y_\eps)-\omega_\lambda(y_0)\ud s \right\vert
        + \frac{1}{2\eps}\left\vert \int_0^t D_s L_\eps^\lambda \sin(2\eps^{-1}\phi_\eps^\lambda)\ud s\right\vert\leq
        M_1^\lambda \sum_{j=1}^n\int_0^t \left\vert y_2^{j\eps}\right\vert\ud s +C,
    \end{equation}
    then to Equation~\eqref{eq:12b},
    \begin{equation}
        \label{eq:15}
        \left\vert \theta_1^{\lambda\eps}(t) \right\vert = \left\vert \int_0^t \dot{\theta}_1^{\lambda\eps}\ud s \right\vert
        = \frac{1}{\eps} \left\vert \int_0^t \theta_\eps^\lambda\cdot D_s L_\eps^\lambda \cos(2\eps^{-1}\phi_\eps^\lambda)
        \ud s \right\vert \leq C,
    \end{equation}
    to Equation~\eqref{eq:12c},
    \begin{equation}
        \label{eq:16}
        \left\vert y_2^{j\eps}(t) \right\vert = \left\vert \int_0^t \dot{y}_2^{{j\eps}}\ud s \right\vert \leq \int_0^t
        \left\vert p_2^{j\eps} \right\vert \ud s + \frac{1}{2\eps}\sum_{\lambda=1}^r \left\vert \int_0^t \theta_\eps^\lambda
        \cdot D_j L_\eps^\lambda \sin(2\eps^{-1}\phi_\eps^\lambda)\ud s\right\vert \leq \int_0^t \left\vert p_2^{j\eps}
        \right\vert \ud s +C,
    \end{equation}
    and finally to Equation~\eqref{eq:12d},
    \begin{IEEEeqnarray*}{rCl}
        \left\vert p_2^{j\eps}(t) \right\vert &=& \left\vert \int_0^t \dot{p}_2^{j\eps}(s)\ud s \right\vert \\
        &\leq& \frac{1}{\eps^2}\left\vert \int_0^t D_j V(y_\eps)-D_jV(y_0)\ud s \right\vert+\frac{1}{\eps^2}\sum_{\lambda=1}^r\left\vert \int_0^t \theta_\eps^\lambda \cdot D_j \omega_\lambda(y_\eps)-\theta_\ast^\lambda \cdot D_j \omega_\lambda(y_0)\ud s \right\vert \\
        &&+\> \frac{1}{2\eps}\sum_{\lambda=1}^r \left\vert \int_0^t \theta_\eps^\lambda \cdot D_t D_j L_\eps^\lambda \sin(2\eps^{-1}\phi_\eps^\lambda) \ud s\right\vert\\
        \IEEEyesnumber \label{eq:17}
        &\leq& M_3 \sum_{k=1}^n\int_0^t \left\vert y_2^{k\eps} \right\vert \ud s+\frac{1}{\eps} \sum_{\lambda=1}^r\left\vert \int_0^t \theta_1^{\lambda\eps}\cdot D_j \omega_\lambda(y_\eps)\ud s \right\vert+ \left\langle  \theta_\ast, M_2 \right\rangle \sum_{k=1}^n\int_0^t \left\vert y_2^{k\eps}\right\vert \ud s+ C.
    \end{IEEEeqnarray*}
    After integrating by parts, Equation~\eqref{eq:12b} and Lemmas~\ref{lemma:2} and~\ref{lemma:3}
    imply that
    \begin{IEEEeqnarray*}{rCl}
        \frac{1}{\eps}\sum_{\lambda=1}^r \left\vert \int_{0}^{t} \theta_1^{\lambda\eps}\cdot D_j \omega_\lambda(y_\eps)\ud s \right\vert &=& \frac{1}{\eps}\sum_{\lambda=1}^r \left\vert \int_{0}^{t} D_j \omega_\lambda(y_\eps) \int_{0}^{s} \dot{\theta}_1^{\lambda\eps}\ud r\ud s \right\vert\\
        &=& \frac{1}{\eps^2} \sum_{\lambda=1}^r \left\vert \int_{0}^{t} D_j \omega_\lambda(y_\eps)\int_{0}^{s}\theta_\eps^\lambda\cdot D_r L_\eps^\lambda \cos(2\eps^{-1}\phi_\eps^\lambda)\ud r \ud s \right\vert\\
        &\leq&\frac{1}{2\eps}\sum_{\lambda=1}^r \left\vert \int_{0}^{t}\frac{\theta_\eps^\lambda\cdot D_s L_\eps^\lambda \cdot D_j \omega_\lambda(y_\eps)}{\dot{\phi}_\eps^\lambda}\sin(2\eps^{-1}\phi_\eps^\lambda)\ud s \right\vert \\
        \IEEEyesnumber\label{eq:18}
        &&+\> \frac{1}{2\eps}\sum_{\lambda=1}^r \left\vert \int_{0}^{t} D_j \omega_\lambda(y_\eps)\int_{0}^{s} \frac{d}{dr}\left( \frac{\theta_\eps^\lambda \cdot D_r L_\eps^\lambda}{\dot{\phi}_\eps^\lambda} \right)\sin(2\eps^{-1}\phi_\eps^\lambda)\ud r \ud s \right\vert \leq C.
    \end{IEEEeqnarray*}
    By combining the inequalities~\eqref{eq:16}--\eqref{eq:18} we obtain
    \begin{equation*}
        \left\vert y_2^{j\eps}(t) \right\vert \leq C
        +\left(  M_3 + \left\langle \theta_\ast, M_2 \right\rangle  \right)\int_0^T \int_0^T \sum_{j=1}^n\left\vert y_2^{j\eps}
        \right\vert \ud s
    \end{equation*}
    and thus
    \begin{equation*}
        \sum_{j=1}^n \left\vert y_2^{j\eps}(t) \right\vert \leq n C+ n\left(M_3 + \left\langle \theta_\ast, M_2
        \right\rangle  \right)\int_0^T \int_0^T \sum_{j=1}^n\left\vert y_2^{j\eps} \right\vert \ud s.
    \end{equation*}
    Finally, a variation of the classical Gronwall inequality (see~\cite[p.~383]{Mitrinovic}) implies that
    \begin{equation*}
        \sum_{j=1}^n \left\vert y_2^{j\eps}(t) \right\vert \leq n C\exp\left( n\left(M_3
            + \left\langle \theta_\ast, M_2 \right\rangle  \right)T^2 \right),\end{equation*}
    for \(t\in[0,T]\), which together with~\eqref{eq:14}--\eqref{eq:18} yields the uniform bound for
    \(\{\theta_1^\eps\},\, \{\phi_2^\eps\},\,\{y_2^\eps\}\) and \(\{p_2^\eps\}\).
\end{proof}

\begin{lemma}
    \label{lemma:6}
    The sequence \(\{\theta_2^\eps\}\) is uniformly bounded in \(L^\infty([0,T],\R^r)\).
\end{lemma}

\begin{proof}
    We write \(\theta_1^\eps\) componentwise as
    \begin{equation*}
        \theta_1^{\lambda\eps} = \frac{1}{\eps}\int_{0}^{\cdot}\dot{\theta}_\eps^\lambda\ud
        t = -\int_{0}^{\cdot}\frac{\theta_\eps^\lambda \cdot D_t
            L_\eps^\lambda}{2\dot{\phi}_\eps^\lambda}\frac{d}{dt}\sin(2\eps^{-1}\phi_\eps^\lambda) \ud t
    \end{equation*}
    and integrate by parts to derive
    \(\theta_2^{\lambda\eps}= \eps^{-1}\left( \theta_1^{\lambda\eps}-[\theta_1^\lambda]^\eps
    \right)=\theta_{21}^{\lambda\eps} + \theta_{22}^{\lambda\eps}\),
    where
    \begin{equation}
        \label{eq:74}
        \theta_{21}^{\lambda\eps}\coloneqq \frac{1}{\eps}\left( \frac{\theta_\ast^\lambda\cdot
            D_t L_0^\lambda}{2\omega_\lambda(y_0)}\sin(2\eps^{-1}\phi^\lambda_0)- \frac{\theta_\eps^\lambda \cdot
            D_t L_\eps^\lambda}{2\dot{\phi}_\eps^\lambda}\sin(2\eps^{-1}\phi^\lambda_\eps) \right),
        \quad
        \theta_{22}^{\lambda\eps}\coloneqq \frac{1}{\eps}\int_0^\cdot \frac{d}{dt} \left( \frac{\theta_\eps^\lambda \cdot
            D_t L_\eps^\lambda}{2\dot{\phi}^\lambda_\eps} \right)\sin(2\eps^{-1}\phi^\lambda_\eps)\ud t.
    \end{equation}
    The claim then follows from Assumption~\ref{ass:1} and Lemmas~\ref{lemma:3} and~\ref{lemma:5}.
\end{proof}

\begin{remark}
    Lemmas~\ref{lemma:5} and~\ref{lemma:6} imply the convergence~\eqref{eq:19}. Moreover, without
    Assumption~\ref{ass:1}, the sequence \(\{\theta_{22}^{\lambda\eps}\}\) is not necessarily bounded in
    \(L^\infty([0,T])\). In fact, by Lemma~\ref{lemma:7} and the definition of \(\theta_{227}^{\lambda\eps}\) in
    the proof of Lemma~\ref{lemma:9}, we would have in this case
    \(\{\theta_{22}^{\lambda\eps}\}= \mathcal{O}(\eps^{-1/2})\).
\end{remark}

\begin{lemma}
    \label{lemma:8}
    There exist a subsequence \(\{\eps^\prime\}\) and functions \(\bar{\phi}_2\in C^{0,1}([0,T],\R^r)\),
    \(\bar{y}_2, \bar{p}_2\in C^{0,1}([0,T],\R^n)\) such that the convergences~\eqref{eq:20}--\eqref{eq:22}
    hold.
\end{lemma}
\begin{proof}
    By taking the time derivative of \(\phi_2^{\lambda\eps} -[\phi^\lambda_2]^\eps\),
    \(y_2^{j\eps}- [y^j_2]^\eps\) and \(p_2^{j\eps}-[p^j_2]^\eps\) for \(\lambda= 1,\ldots, r\) and
    \(j=1, \ldots, n\) we obtain
    \begin{IEEEeqnarray}{rcl}
        \label{eq:23}
        \frac{d}{dt}\left( \phi_2^{\lambda\eps} -[\phi^\lambda_2]^\eps \right)\>&=&\> \frac{\omega_\lambda(y_\eps) -\omega_\lambda(y_0)}{\eps^2} + \frac{d}{dt}\left( \frac{D_t L_\eps^\lambda}{4\dot{\phi}_\eps^\lambda} \right)\cos(2\eps^{-1}\phi_\eps^\lambda)- \frac{d}{dt}\left([\phi^\lambda_2]^\eps +\frac{D_t L_\eps^\lambda}{4\dot{\phi}_\eps^\lambda}\cos(2\eps^{-1}\phi_\eps^\lambda) \right),\\
        \label{eq:24}
        \frac{d}{dt}\left( y_2^{j\eps}- [y^j_2]^\eps \right)\>&=&\> \frac{p_\eps^j - p_0^j}{\eps^2} +\sum_{\lambda=1}^r\frac{d}{dt}\left( \frac{\theta_\eps^\lambda \cdot D_j L_\eps^\lambda}{4\dot{\phi}_\eps^\lambda} \right)\cos(2\eps^{-1}\phi^\lambda_\eps)- \frac{d}{dt}\left( [y^j_2]^\eps + \sum_{\lambda=1}^r \frac{\theta_\eps^\lambda \cdot D_j L_\eps^\lambda}{4\dot{\phi}_\eps^\lambda}\cos(2\eps^{-1}\phi^\lambda_\eps)\right),\IEEEeqnarraynumspace\\
        \IEEEnonumber
        \frac{d}{dt}\left( p_2^{j\eps}-[p^j_2]^\eps \right) \>&=&\>-\>\frac{D_j V(y_\eps)-D_jV(y_0)}{\eps^2}-\sum_{\lambda=1}^r\theta_\ast^\lambda \frac{D_j \omega_\lambda(y_\eps)-D_j\omega_\lambda(y_0)}{\eps^2}-\sum_{\lambda=1}^r \frac{\theta_1^{\lambda\eps}-[\theta_1^\lambda]^\eps}{\eps}D_j \omega_\lambda(y_\eps)  \\
        \IEEEnonumber
        && \> -\>\sum_{\lambda=1}^r\frac{d}{dt}\left( \frac{\theta_\eps^\lambda \cdot D_tD_j L_\eps^\lambda}{4\dot{\phi}_\eps^\lambda} \right)\cos(2\eps^{-1}\phi^\lambda_\eps)-\frac{d}{dt}\left( [p^j_2]^\eps_1 - \sum_{\lambda=1}^r \frac{\theta_\eps^\lambda\cdot D_tD_jL_\eps^\lambda}{4\dot{\phi}_\eps^\lambda}\cos(2\eps^{-1}\phi^\lambda_\eps)\right)\\
        \IEEEeqnarraymulticol{3}{r}{
            \label{eq:25}
            +\>\sum_{\lambda=1}^r \frac{d}{dt}\left( \frac{\theta_\ast^\lambda \cdot D_t L_0^\lambda\cdot D_j \omega_\lambda(y_\eps)}{4\omega^2(y_0)} \right)\cos(2\eps^{-1}\phi^\lambda_0)-\frac{d}{dt}\left( [p^j_2]^\eps_2+\sum_{\lambda=1}^r\frac{\theta_\ast^\lambda \cdot D_t L_0^\lambda\cdot D_j \omega_\lambda(y_\eps)}{4\omega^2_\lambda(y_0)}\cos(2\eps^{-1}\phi^\lambda_0) \right),\IEEEeqnarraynumspace}
    \end{IEEEeqnarray}
    where we used \([p^j_2]^\eps = [p^j_2]^\eps_1 + [p^j_2]^\eps_2\) with
    \begin{equation*}
        [p^j_2]^\eps_1 \coloneqq \sum_{\lambda=1}^r \frac{\theta_\ast^\lambda\cdot D_t D_j
            L_0^\lambda}{4\omega_\lambda(y_0)}\cos(2\eps^{-1}\phi^\lambda_0),\qquad [p^j_2]^\eps_2 \coloneqq -
        \sum_{\lambda=1}^r\frac{\theta_\ast^\lambda \cdot D_t L_0^\lambda \cdot
            D_j L_0^\lambda}{4\omega_\lambda(y_0) }\cos(2\eps^{-1}\phi^\lambda_0).
    \end{equation*}
    For the derivation of Equation~\eqref{eq:25} note that in the evaluation of \(\dot{p}_2^{j\eps}\) we need
    to evaluate the expression
    \begin{equation*}
        \sum_{\lambda=1}^r \frac{\theta_\eps^\lambda \cdot D_j\omega_\lambda(y_\eps)-\theta_\ast^\lambda
            \cdot D_j \omega_\lambda(y_0)}{\eps^2} = \sum_{\lambda=1}^r \frac{\theta_\eps^\lambda
            - \theta_\ast^\lambda}{\eps^2}D_j \omega_\lambda(y_\eps)+\sum_{\lambda=1}^r \theta_\ast^\lambda\frac{D_j
            \omega_\lambda(y_\eps)-D_j \omega_\lambda(y_0)}{\eps^2},
    \end{equation*}
    in which we rewrite the first term on the right-hand side by introducing
    \([\theta_1^\lambda]^\eps\),
    i.e.,
    \begin{IEEEeqnarray*}{rCl}
        \frac{\theta_\eps^\lambda - \theta_\ast^\lambda}{\eps^2}D_j \omega_\lambda(y_\eps) &=&\frac{\theta_1^{\lambda\eps}-[\theta_1^\lambda]^\eps}{\eps}D_j \omega_\lambda(y_\eps) +\frac{d}{dt}\left( \frac{\theta_\ast^\lambda \cdot D_t L_0^\lambda \cdot D_j \omega_\lambda(y_\eps)}{4\omega^2_\lambda(y_0)}\cos(2\eps^{-1}\phi^\lambda_0) \right)\\
        && -\>\frac{d}{dt}\left(\frac{\theta_\ast^\lambda \cdot D_t L_0^\lambda  \cdot  D_j \omega_\lambda(y_\eps)}{4\omega^2(y_0)} \right)\cos(2\eps^{-1}\phi^\lambda_0).
    \end{IEEEeqnarray*}
    By Lemmas~\ref{lemma:5},~\ref{lemma:4} and~\ref{lemma:6} the sequence \(\{\phi_2^\eps -[\phi_2]^\eps\}\)
    is bounded in \(C^{0,1}([0,T], \R^r)\) and the sequences \(\{y_2^\eps -[y_2]^\eps\}\) and
    \(\{p_2^\eps -[p_2]^\eps\}\) are bounded in \(C^{0,1}([0,T],\R^n)\). The claim follows after
    successive applications of~\cite[Principle~4]{Bornemann}.
\end{proof}

\begin{lemma}
    \label{lemma:9}
    There exists a further subsequence \(\{\eps^\prime\}\) and a function
    \(\bar{\theta}_2\in C^\infty([0,T],\R^r)\) such that the convergence~\eqref{eq:26} holds. Moreover, the
    component functions \(\bar{\theta}_2^\lambda \;(\lambda = 1, \ldots, r)\) satisfy~\eqref{eq:27b}.
\end{lemma}

\begin{proof}
    We write \(\theta_2^{\lambda\eps}= \sum_{i=1}^2 \theta_{2i}^{\lambda\eps}\) as in the proof of
    Lemma~\ref{lemma:6} and \([\theta_2^\lambda]^\eps = \sum_{i=1}^2[\theta_2^\lambda]_i^\eps\), where
    \([\theta_2^\lambda]_i^\eps\) \((i=1,2)\) will be defined later in this proof. We then show that there exist
    component functions \(\bar{\theta}_2^\lambda\coloneqq\sum_{i=1}^2\bar{\theta}_{2i}^\lambda\) and a
    subsequence \(\{\eps^\prime\}\), not relabelled, such that for \(i=1\)
    \begin{equation}
        \label{eq:56}
        \theta_{21}^{\lambda\eps} -[\theta^\lambda_2]_1^\eps\to \bar{\theta}^\lambda_{21} \quad \text{in} \quad C([0,T])
    \end{equation}
    and for \(i=2\)
    \begin{equation}
        \label{eq:58}
        \theta_{22}^{\lambda\eps} -[\theta^\lambda_2]_2^\eps \to \bar{\theta}^\lambda_{22} \quad \text{in} \quad C([0,T]),
        \qquad \frac{d}{dt}\left(  \theta_{22}^{\lambda\eps} -[\theta^\lambda_2]_2^\eps \right)\weak{\ast}
        \frac{d \bar{\theta}^\lambda_{22}}{dt} \quad \text{in} \quad L^\infty([0,T]).
    \end{equation}
    The convergence result~\eqref{eq:26} then follows immediately.

    \paragraph{Part \(\bm{i=1}\):}
    To prove~\eqref{eq:56}, we expand \(\theta_{21}^{\lambda\eps}\) in~\eqref{eq:74}, by replacing
    \(\theta_\eps^\lambda \to \theta_\ast^\lambda + \left( \theta_\eps^\lambda - \theta_\ast^\lambda
    \right)\),
    \(D_t L_\eps^\lambda \to D_t L_0^\lambda + \left( D_t L_\eps^\lambda - D_t L_0^\lambda \right)\) and
    \(\sin(2\eps^{-1}\phi_\eps^\lambda) \to \sin(2\eps^{-1}\phi_0^\lambda) + \left(
    \sin(2\eps^{-1}\phi_\eps^\lambda) - \sin(2\eps^{-1}\phi_0^\lambda) \right)\), and assign the resulting
    terms to the functions \(\theta_{21j}^{\lambda\eps}\) \((j=1,\ldots,5)\). That is, we derive
    \(\theta_{21}^{\lambda\eps}= \sum_{j=1}^5\theta_{21j}^{\lambda\eps}\), where
    \begin{IEEEeqnarray*}{rCl}
        \theta_{211}^{\lambda\eps}&\coloneqq& \frac{1}{\eps}\left( \frac{1}{\omega_\lambda(y_0)}-\frac{1}{\dot{\phi}_\eps^\lambda} \right)\frac{\theta_\ast^\lambda\cdot D_t L_0^\lambda}{2}\sin(2\eps^{-1}\phi^\lambda_0),\\
        \theta_{212}^{\lambda\eps}&\coloneqq& -\frac{1}{\eps}\frac{\left( \theta_\eps^\lambda - \theta_\ast^\lambda \right)\cdot D_t L_0^\lambda}{2\dot{\phi}_\eps^\lambda}\sin(2\eps^{-1}\phi^\lambda_0),\\
        \theta_{213}^{\lambda\eps}&\coloneqq& -\frac{1}{\eps}\frac{\theta_\ast^\lambda \cdot \left( D_t L_\eps^\lambda - D_t L_0^\lambda\right)}{2\dot{\phi}_\eps^\lambda}\sin(2\eps^{-1}\phi^\lambda_0),\\
        \theta_{214}^{\lambda\eps}&\coloneqq& -\frac{1}{\eps}\frac{\theta_\ast^\lambda \cdot D_t L_0^\lambda}{2\dot{\phi}_\eps^\lambda}\left( \sin(2\eps^{-1}\phi^\lambda_\eps)-\sin(2\eps^{-1}\phi^\lambda_0) \right),\\
        \theta_{215}^{\lambda\eps}&\coloneqq& - \>\frac{1}{\eps}\frac{\left( \theta_\eps^\lambda - \theta_\ast^\lambda \right)\cdot \left( D_t L_\eps^\lambda- D_t L_0^\lambda\right)}{2\dot{\phi}_\eps^\lambda}\sin(2\eps^{-1}\phi^\lambda_0)\\
        &&-\> \frac{1}{\eps}\frac{\left( \theta_\eps^\lambda - \theta_\ast^\lambda \right)\cdot D_t L_0^\lambda}{2\dot{\phi}_\eps^\lambda}\left( \sin(2\eps^{-1}\phi^\lambda_\eps)-\sin(2\eps^{-1}\phi^\lambda_0) \right)\\
        && -\frac{1}{\eps}\frac{\theta_\ast^\lambda \cdot \left( D_t L_\eps^\lambda - D_t L_0^\lambda\right)}{2\dot{\phi}_\eps^\lambda}\left( \sin(2\eps^{-1}\phi^\lambda_\eps)-\sin(2\eps^{-1}\phi^\lambda_0) \right)\\
        &&- \>\frac{1}{\eps}\frac{\left( \theta_\eps^\lambda - \theta_\ast^\lambda \right)\cdot \left( D_t L_\eps^\lambda- D_t L_0^\lambda\right)}{2\dot{\phi}_\eps^\lambda}\left( \sin(2\eps^{-1}\phi^\lambda_\eps)-\sin(2\eps^{-1}\phi^\lambda_0) \right).
    \end{IEEEeqnarray*}
    Notice that by~\eqref{eq:12} and Lemma~\ref{lemma:4}, we have
    \(\theta_{21j}^{\lambda\eps}=\mathcal{O}(1)\) for \(j=1,\ldots,4\) and
    \(\theta_{215}^{\lambda\eps}=\mathcal{O}(\eps)\). The function
    \(\theta_{21}^{\lambda\eps}= \sum_{j=1}^5\theta_{21j}^{\lambda\eps}\) is composed of oscillatory and
    non-oscillatory (averaged) terms. For the proof of~\eqref{eq:56}, we therefore define the corresponding
    oscillatory term \([\theta_2^\lambda]_1^\eps \coloneqq\sum_{j=1}^4[\theta_2^\lambda]_{1j}^\eps\), where
    \begin{IEEEeqnarray*}{rCl}
        [\theta_2^\lambda]^\eps_{11} &\coloneqq& -\frac{\theta_\ast^\lambda (D_t L_0^\lambda)^2}{8\omega_\lambda^2(y_0)}\cos(4\eps^{-1}\phi^\lambda_0),\\
        \phantom{}[\theta_2^\lambda]^\eps_{12}&\coloneqq&-\frac{\theta_\ast^\lambda (D_t L_0^\lambda)^2}{8\omega_\lambda^2(y_0)}\cos(4\eps^{-1}\phi^\lambda_0),\\
        \phantom{}[\theta_2^\lambda]^\eps_{13}&\coloneqq& \frac{(\theta_\ast^\lambda)^2 \left\vert DL_0^\lambda \right\vert^2}{8\omega_\lambda(y_0)}\cos(4\eps^{-1}\phi^\lambda_0)-\sum_{\substack{\mu=1\\\mu \neq \lambda}}^r\frac{\theta_\ast^\lambda\theta_\ast^\mu\left\langle DL_0^\mu,DL_0^\lambda \right\rangle }{4\omega_\lambda(y_0)}\sin(2\eps^{-1}\phi^\lambda_0)\sin(2\eps^{-1}\phi^\mu_0),\\
        \phantom{}[\theta_2^\lambda]^\eps_{14}&\coloneqq&-\frac{\theta_\ast^\lambda \cdot D_t L_0^\lambda}{\omega_\lambda(y_0)}\bar{\phi}_2^\lambda \cos(2\eps^{-1}\phi^\lambda_0)+\frac{\theta_\ast^\lambda (D_t L_0^\lambda)^2}{8\omega_\lambda^2(y_0)}\cos(4\eps^{-1}\phi^\lambda_0).
    \end{IEEEeqnarray*}
    We now define the averaged functions \(\bar{\theta}_{21j}^\lambda\) \((j=1,\ldots,4)\) such that for
    \(\bar{\theta}_{21}^\lambda\coloneqq \sum_{j=1}^4 \bar{\theta}_{21j}^\lambda\) the statement
    in~\eqref{eq:56} holds. More precisely, we will show that in the case of \(j=1,2,3\)
    \begin{equation}
        \label{eq:28}
        \theta_{21j}^{\lambda\eps} -[\theta^\lambda_2]_{1j}^\eps\to \bar{\theta}^\lambda_{21j} \quad \text{in} \quad C([0,T]),
        \qquad \frac{d}{dt}\left(  \theta_{21j}^{\lambda\eps} -[\theta^\lambda_2]_{1j}^\eps \right)\weak{\ast}
        \frac{d \bar{\theta}^\lambda_{21j}}{dt} \quad \text{in} \quad L^\infty([0,T]),
    \end{equation}
    and in the case of \(j=4\)
    \begin{equation}
        \label{eq:57}
        \theta_{214}^{\lambda\eps} -[\theta^\lambda_2]_{14}^\eps\to \bar{\theta}^\lambda_{214} \quad \text{in} \quad C([0,T]).
    \end{equation}
    In the following, we give the detailed proof of the convergences
    in~\eqref{eq:28} for the case \(j=1\). The
    other cases follow along a similar line of arguments.

    \paragraph{Case \(\bm{j=1}\):}
    We start by defining \(\bar{\theta}_{221}^\lambda\), i.e.,
    \begin{equation*}
        \bar{\theta}_{211}^\lambda \coloneqq \frac{\theta_\ast^\lambda (D_t L_0^\lambda)^2}{8\omega_\lambda^2(y_0)},
    \end{equation*}
    for \(\lambda = 1,\ldots, r\) and use the trigonometric identity \(1 - \cos(4x) = 2\sin^2(2x)\) to derive
    \begin{equation}
        \label{eq:68}
        \bar{\theta}_{211}^\lambda + [\theta_2^\lambda]^\eps_{11} = \frac{\theta_\ast^\lambda
            (D_t L_0^\lambda)^2}{8\omega_\lambda^2(y_0)} -\frac{\theta_\ast^\lambda
            (D_t L_0^\lambda)^2}{8\omega_\lambda^2(y_0)}\cos(4\eps^{-1}\phi^\lambda_0)
        = \frac{\theta_\ast^\lambda (D_t L_0^\lambda)^2}{4\omega_\lambda^2(y_0)} \sin^2(2\eps^{-1}\phi_0^\lambda).
    \end{equation}
    Moreover, with Equation~\eqref{eq:12a} we can write
    \begin{IEEEeqnarray*}{rCl}
        \theta_{211}^{\lambda\eps} &=& \frac{1}{\eps} \left(\frac{\dot{\phi}_\eps^\lambda
            - \omega_\lambda(y_0)}{\omega_\lambda(y_0)\dot{\phi}_\eps^\lambda}\right)\frac{\theta_\ast^\lambda\cdot
            D_t L_0^\lambda}{2}\sin(2\eps^{-1}\phi_0^\lambda) \\
        \IEEEyesnumber
        \label{eq:69}
        &=&  \left(\frac{1}{\eps}\frac{\omega_\lambda(y_\eps) - \omega_\lambda(y_0)}{\dot{\phi}_\eps^\lambda}
        + \frac{D_t L_\eps^\lambda}{2\dot{\phi}_\eps^\lambda}\sin(2\eps^{-1}\phi_\eps^\lambda)\right)\frac{\theta_\ast^\lambda\cdot
            D_t L_0^\lambda}{2\omega_\lambda(y_0)}\sin(2\eps^{-1}\phi_0^\lambda).
    \end{IEEEeqnarray*}
    Now, we use Equations~\eqref{eq:68} and~\eqref{eq:69} to write
    \begin{equation}
        \label{eq:70}
        \theta_{211}^{\lambda\eps} - [\theta_2^\lambda]^\eps_{11} - \bar{\theta}_{211}^\lambda =\left( u_1^{\lambda\eps}
        + v_1^{\lambda\eps} \right)w_1^{\lambda\eps},
    \end{equation}
    where
    \begin{equation*}
        u_1^{\lambda\eps} \coloneqq \frac{1}{\eps}\frac{\omega_\lambda(y_\eps) - \omega_\lambda(y_0)}{\dot{\phi}_\eps^\lambda},
        \quad v_1^{\lambda\eps} \coloneqq  \frac{D_t L_\eps^\lambda}{2\dot{\phi}_\eps^\lambda}\sin(2\eps^{-1}\phi^\lambda_\eps)
        -\frac{D_t L_0^\lambda}{2\omega_\lambda(y_0)}\sin(2\eps^{-1}\phi^\lambda_0), \quad w_1^{\lambda\eps}
        \coloneqq \frac{\theta_\ast^\lambda \cdot D_t L_0^\lambda}{2\omega_\lambda(y_0)}\sin(2\eps^{-1}\phi^\lambda_0).
    \end{equation*}
    It follows from Lemmas~\ref{lemma:5} and~\ref{lemma:4}, and the system of differential
    equations~\eqref{eq:12} that the sequences \(\{\eps^{-1}u_1^{\lambda\eps}\}\),
    \(\{\dot{u}_1^{\lambda\eps}\}\), \(\{\eps^{-1}v_1^{\lambda\eps}\}\), \(\{\dot{v}_j^{\lambda\eps}\}\),
    \(\{w_1^{\lambda\eps}\}\) and \(\{\eps \dot{w}_1^{\lambda\eps}\}\) are bounded in
    \(L^\infty([0,T])\). This implies the uniform convergence, and after an application of~\cite[Principle
        1]{Bornemann}, the weak\(^\ast\) convergence in~\eqref{eq:28}.

    For the cases \(j = 2,3\), we only summarise the equations corresponding to~\eqref{eq:70}, from which
    \(u_j^{\lambda\eps}\), \(v_3^{\lambda\eps}\) and \(w_j^{\lambda\eps}\) can be read off. The convergences
    as in~\eqref{eq:28} are then proven similarly to the case \(j=1\), by
    applying Lemmas~\ref{lemma:5}
    and~\ref{lemma:4} (and~\ref{lemma:6} in the case \(j=2\)). The case \(j=4\) requires more explanation and
    is thus again described in greater detail.
    \paragraph{Case \(\bm{j=2}\):}
    With
    \begin{equation*}
        \bar{\theta}_{212}^\lambda \coloneqq \frac{\theta_\ast^\lambda (D_t L_0^\lambda)^2}{8\omega_\lambda^2(y_0)},
    \end{equation*}
    we can write
    \begin{equation*}
        \theta_{212}^{\lambda\eps} -[\theta_2^\lambda]^\eps_{12}- \bar{\theta}_{212}^\lambda = u_2^{\lambda\eps} w_2^{\lambda\eps},
    \end{equation*}
    where
    \begin{equation*}
        u_2^{\lambda\eps} \coloneqq \frac{[\theta_1^\lambda]^\eps}{\omega_\lambda(y_0)} -
        \frac{\theta_1^{\lambda\eps}}{\dot{\phi}_\eps^\lambda},\qquad w_2^{\lambda\eps}
        \coloneqq \frac{D_t L_0^\lambda}{2}\sin(2\eps^{-1}\phi_0^\lambda).
    \end{equation*}
    \paragraph{Case \(\bm{j=3}\):}
    Analogously, with
    \begin{equation*}
        \bar{\theta}_{213}^\lambda \coloneqq -\frac{(\theta_\ast^\lambda)^2 \left\vert DL_0^\lambda \right\vert^2}{8\omega_\lambda(y_0)},
    \end{equation*}
    we write
    \begin{equation*}
        \theta_{213}^{\lambda\eps} -[\theta_2^\lambda]^\eps_{13}- \bar{\theta}_{213}^\lambda = \left( u_3^{\lambda\eps}
        + v_3^{\lambda\eps} \right)w_3^{\lambda\eps};
    \end{equation*}
    here
    \begin{equation*}
        u_3^{\lambda\eps} \coloneqq -\frac{1}{\eps} \frac{ \left\langle p_\eps, D L_\eps^\lambda \right\rangle
            - \left\langle p_0, DL_0^\lambda \right\rangle }{\dot{\phi}_\eps^\lambda}, \qquad w_3^{\lambda\eps}
        \coloneqq \frac{\theta_\ast^\lambda}{2}\sin(2\eps^{-1}\phi_0^\lambda),
    \end{equation*}
    \begin{equation*}
        v_3^{\lambda\eps} \coloneqq \sum_{\mu=1}^r \frac{\theta_\ast^\mu\left\langle DL_0^\mu ,
            DL_0^\lambda \right\rangle }{2 \omega_\lambda(y_0)}\sin(2\eps^{-1}\phi_0^\mu) - \frac{\theta_\eps^\mu \left\langle
            DL_\eps^\mu , DL_\eps^\lambda  \right\rangle }{2\dot{\phi}_\eps^\lambda}\sin(2\eps^{-1}\phi_\eps^\mu).
    \end{equation*}
    \paragraph{Case \(\bm{j=4}\):}
    For this final case we first define
    \begin{equation*}
        \bar{\theta}_{214}^\lambda \coloneqq \frac{\theta_\ast^\lambda (D_t L_0^\lambda)^2}{8\omega_\lambda^2(y_0)}
    \end{equation*}
    and use \(\phi_\eps^\lambda = \phi_0^\lambda + \eps^2 \phi_2^{\lambda\eps}\) (see Definition~\ref{def:1})
    with a trigonometric identity to write
    \begin{equation*}
        \sin(2\eps^{-1}\phi_\eps^\lambda) =
        \sin(2\eps^{-1}\phi_0^\lambda)\cos(2\eps\phi_2^{\lambda\eps})+\cos(2\eps^{-1}\phi_0^\lambda)\sin(2\eps\phi_2^{\lambda\eps}).
    \end{equation*}
    This allows us to derive the equation
    \begin{equation*}
        \theta_{214}^{\lambda\eps} - [\theta_2^\lambda]^\eps_{14} - \bar{\theta}_{214}^\lambda
        = u_{41}^{\lambda\eps} w_{41}^{\lambda\eps} + u_{42}^{\lambda\eps} w_{42}^{\lambda\eps},
    \end{equation*}
    where
    \begin{IEEEeqnarray*}{rCl+rCl}
        u_{41}^{\lambda\eps} &\coloneqq& \frac{[\phi_2^\lambda]^\eps + \bar{\phi}_2^\lambda}{\omega_\lambda(y_0)} - \frac{1}{\eps}\frac{\sin(2\eps \phi_2^{\lambda\eps})}{2\dot{\phi}_\eps^\lambda},& w_{41}^{\lambda\eps} &\coloneqq& \theta_\ast^\lambda \cdot D_t L_0^\lambda\cos(2\eps^{-1}\phi^\lambda_0),\\
        u_{42}^{\lambda\eps} &\coloneqq& \frac{1}{\eps}\frac{1-\cos(2\eps\phi_2^{\lambda\eps})}{2\dot{\phi}_\eps^\lambda},& w_{42}^{\lambda\eps} &\coloneqq& \theta_\ast^\lambda \cdot D_t L_0^\lambda \sin(2\eps^{-1}\phi_0^\lambda).
    \end{IEEEeqnarray*}
    By~\eqref{eq:12},~\eqref{eq:20} and Lemma~\ref{lemma:4} we obtain
    \begin{equation*}
        \left\vert u_{41}^{\lambda\eps} \right\vert \leq \left\vert \frac{[\phi_2^\lambda]^\eps
            + \bar{\phi}_2^\lambda}{\omega_\lambda(y_0)} -\frac{\phi_2^{\lambda\eps}}{\dot{\phi}_\eps^\lambda} \right\vert
        + \left\vert  \frac{\phi_2^{\lambda\eps}}{\dot{\phi}_\eps^\lambda} \sum_{k=1}^\infty
        \frac{(-1)^k}{(2k+1)!}(2\eps \phi_2^{\lambda\eps})^{2k}\right\vert\to 0 \quad \text{in} \quad C([0,T])
    \end{equation*}
    and
    \begin{equation*}
        \left\vert u_{42}^{\lambda\eps} \right\vert \leq  \left\vert \frac{\phi_2^{\lambda\eps}}{\dot{\phi}_\eps^\lambda}
        \sum_{k=1}^\infty \frac{(-1)^k}{(2k)!}(2\eps\phi_2^{\lambda\eps})^{2k-1} \right\vert = \mathcal{O}(\eps),
    \end{equation*}
    which implies the uniform convergence in~\eqref{eq:57}.

    \paragraph{Part \(\bm{i=2}\):}
    To prove~\eqref{eq:58} we expand \(\theta_{22}^{\lambda\eps}\) in~\eqref{eq:74}, by writing out the time
    derivative and using the equations
    \begin{equation*}
        \ddot{\phi}_\eps^\lambda = D_t \omega_\lambda(y_\eps) + \frac{\eps}{2} D^2_t L_\eps^\lambda \sin(2\eps^{-1}\phi^\lambda_\eps)
        +  D_t L_\eps^\lambda \cdot \dot{\phi}_\eps^\lambda \cos(2\eps^{-1}\phi^\lambda_\eps)
    \end{equation*}
    and
    \begin{equation*}
        D_t^2 L_\eps^\lambda = \left\langle D^2 L_\eps^\lambda \dot{y}_\eps, \dot{y}_\eps \right\rangle
        -\left\langle DV(y_\eps),DL_\eps^\lambda \right\rangle -\sum_{\mu=1}^r\theta_\eps^\mu \left\langle
        D\omega_\mu(y_\eps), DL_\eps^\lambda \right\rangle + \sum_{\mu=1}^r\theta_\eps^\mu
        \left\langle D\omega_\mu(y_\eps), DL_\eps^\lambda \right\rangle\cos(2\eps^{-1}\phi_\eps^\mu).
    \end{equation*}
    In this way, we can write \(\theta_{22}^{\lambda\eps} = \sum_{j=1}^8\theta^{\lambda\eps}_{22j}\), where
    \begin{IEEEeqnarray*}{rCl}
        \theta^{\lambda\eps}_{221} &\coloneqq& -\frac{1}{\eps}\int_0^\cdot \frac{\theta_\eps^\lambda(D_t L_\eps^\lambda)^2}{\dot{\phi}_\eps^\lambda}\cos(2\eps^{-1}\phi^\lambda_\eps)\sin(2\eps^{-1}\phi^\lambda_\eps)\ud t,\\
        \theta^{\lambda\eps}_{222} &\coloneqq& -\frac{1}{\eps}\int_0^\cdot \frac{\theta_\eps^\lambda \cdot D_t L_\eps^\lambda \cdot D_t \omega_\lambda(y_\eps)}{2(\dot{\phi}_\eps^\lambda)^2}\sin(2\eps^{-1}\phi^\lambda_\eps)\ud t,\\
        \theta^{\lambda\eps}_{223} &\coloneqq& \frac{1}{\eps}\int_0^\cdot \frac{\theta_\eps^\lambda\left\langle D^2L_\eps^\lambda \dot{y}_\eps,\dot{y}_\eps\right\rangle }{2\dot{\phi}_\eps^\lambda}\sin(2\eps^{-1}\phi^\lambda_\eps)\ud t,\\
        \theta_{224}^{\lambda\eps}&\coloneqq& -\frac{1}{\eps}\int_{0}^{\cdot}\frac{\theta_\eps^\lambda \left\langle DV(y_\eps), DL_\eps^\lambda \right\rangle }{2\dot{\phi}_\eps^\lambda}\sin(2\eps^{-1}\phi_0^\lambda)\ud t,\\
        \theta^{\lambda\eps}_{225} &\coloneqq& -\frac{1}{\eps}\sum_{\mu=1}^r\int_0^\cdot  \frac{\theta_\eps^\lambda\theta_\eps^\mu \left\langle D \omega_\mu (y_\eps), DL_\eps^\lambda   \right\rangle }{2\dot{\phi}_\eps^\lambda}\sin(2\eps^{-1}\phi^\lambda_\eps) \ud t,\\
        \theta^{\lambda\eps}_{226} &\coloneqq& \frac{1}{\eps} \int_0^\cdot\frac{(\theta_\eps^\lambda)^2 \left\langle D \omega_\lambda(y_\eps), DL_\eps^\lambda \right\rangle }{2\dot{\phi}_\eps^\lambda}\cos(2\eps^{-1}\phi^\lambda_\eps)\sin(2\eps^{-1}\phi^\lambda_\eps)\ud t,\\
        \theta^{\lambda\eps}_{227} &\coloneqq& \frac{1}{\eps}\sum_{\substack{\mu=1\\\mu \neq \lambda}}^r \int_0^\cdot\frac{\theta_\eps^\lambda \theta_\eps^\mu \left\langle D \omega_\mu(y_\eps), DL_\eps^\lambda \right\rangle }{2\dot{\phi}_\eps^\lambda}\cos(2\eps^{-1}\phi^\mu_\eps)\sin(2\eps^{-1}\phi^\lambda_\eps)\ud t,\\
        \theta^{\lambda\eps}_{228} &\coloneqq& -\int_0^\cdot \frac{\theta_\eps^\lambda \cdot D_t L_\eps^\lambda \cdot D_t^2 L_\eps^\lambda}{4(\dot{\phi}_\eps^\lambda)^2}\sin^2(2\eps^{-1}\phi^\lambda_\eps)\ud t.\\
    \end{IEEEeqnarray*}
    Again, the function \(\theta_{22}^{\lambda\eps}=\sum_{j=1}^8 \theta_{22j}^{\lambda\eps}\) consists of
    oscillatory and non-oscillatory terms. To derive the statement in~\eqref{eq:58}, we therefore define the
    corresponding oscillatory term
    \([\theta_2^\lambda]_2^\eps \coloneqq \sum_{j=1}^7[\theta_2^\lambda]_{2j}^\eps\), where
    \begin{IEEEeqnarray*}{rCl}
        [\theta_2^\lambda]_{21}^\eps &\coloneqq&\frac{\theta_\ast^\lambda(D_t L_0^\lambda)^2}{8\omega^2_\lambda(y_0)}\cos(4\eps^{-1}\phi_0^\lambda),\\
        \phantom{}[\theta_2^\lambda]_{22}^\eps &\coloneqq&\frac{\theta_\ast^\lambda(D_t L_0^\lambda)^2 }{4\omega^2_\lambda(y_0)}\cos(2\eps^{-1}\phi^\lambda_0),\\
        \phantom{}[\theta_2^\lambda]_{23}^\eps &\coloneqq&-\frac{\theta_\ast^\lambda \left\langle D^2L_0^\lambda \dot{y}_0,\dot{y}_0\right\rangle }{4\omega^2_\lambda(y_0)}\cos(2\eps^{-1}\phi^\lambda_0),\\
        \phantom{}[\theta_2^\lambda]_{24}^\eps &\coloneqq& \frac{\theta_\ast^\lambda\left\langle DV(y_0), DL_0^\lambda\right\rangle }{4\omega_\lambda^2(y_0)}\cos(2\eps^{-1}\phi_0^\lambda),\\
        \phantom{}[\theta_2^\lambda]_{25}^\eps &\coloneqq&\sum_{\mu=1}^r \frac{\theta_\ast^\lambda\theta_\ast^\mu \left\langle D\omega_\mu(y_0), DL_0^\lambda \right\rangle }{4\omega_\lambda^2(y_0)}\cos(2\eps^{-1}\phi^\lambda_0),\\
        \phantom{}[\theta_2^\lambda]_{26}^\eps &\coloneqq&-\frac{(\theta_\ast^\lambda)^2 \vert DL_0^\lambda \vert^2 }{16\omega_\lambda(y_0)}\cos(4\eps^{-1}\phi^\lambda_0),\\
        \phantom{}[\theta_2^\lambda]_{27}^\eps &\coloneqq&\sum_{\substack{\mu=1\\\mu \neq \lambda}}^r \frac{\theta_\ast^\lambda\theta_\ast^\mu \left\langle D \omega_\mu(y_0), DL_0^\lambda \right\rangle }{8\omega_\lambda(y_0)}\left\{ \frac{\cos\left( 2\eps^{-1}(\phi^\mu_0-\phi^\lambda_0) \right)}{\omega_\mu(y_0) - \omega_\lambda(y_0)}-\frac{\cos\left( 2\eps^{-1}(\phi^\mu_0+\phi^\lambda_0) \right)}{\omega_\mu(y_0) + \omega_\lambda(y_0)}\right\}.
    \end{IEEEeqnarray*}
    We show that for a subsequence \(\{\eps^\prime\}\) (not relabelled) there exist non-oscillatory
    functions \(\bar{\theta}_{22j}^\lambda\) \({(j=1,\ldots,8)}\) such that for
    \(\bar{\theta}_{22}^\lambda\coloneqq \sum_{j=1}^8 \bar{\theta}_{22j}^\lambda\), the statement
    in~\eqref{eq:58} holds. More precisely, we will prove that for \(j=1,\ldots,7\)
    \begin{equation}
        \label{eq:29}
        \theta_{22j}^{\lambda\eps} -[\theta^\lambda_2]_{2j}^\eps\to \bar{\theta}^\lambda_{22j} \quad \text{in} \quad C([0,T]),
        \qquad \frac{d}{dt} \left( \theta_{22j}^{\lambda\eps} -[\theta^\lambda_2]_{2j}^\eps \right) \weak{\ast}
        \frac{d \bar{\theta}^\lambda_{22j}}{dt} \quad \text{in} \quad L^\infty([0,T]),
    \end{equation}
    and for \(j=8\)
    \begin{equation}
        \label{eq:29a}
        \theta_{228}^{\lambda\eps}\to \bar{\theta}^\lambda_{228} \quad \text{in} \quad C([0,T]),\qquad
        \frac{d \theta_{228}^{\lambda\eps}}{dt} \weak{\ast} \frac{d \bar{\theta}^\lambda_{228}}{dt} \quad \text{in}
        \quad L^\infty([0,T]).
    \end{equation}
    Note that the scaling in \(\theta_{228}^{\lambda\eps}\) is different from the scaling in
    \(\theta_{22j}^{\lambda\eps}\) \((j=1,\ldots,7)\). As a consequence, there is no non-converging oscillatory
    component that we would have to subtract from \(\theta_{228}^{\lambda\eps}\) in order to analyse the limit
    \(\eps\to0\).

    We now give a detailed proof of the convergences in~\eqref{eq:29} for the case \(j=1\). The other cases
    are dealt with similarly.
    \paragraph{Case \(\bm{j=1}\):}
    For \(\lambda = 1,\ldots, r\) we start by writing
    \begin{IEEEeqnarray*}{rCl}
        \frac{d \theta_{221}^{\lambda\eps}}{dt} &=& -\frac{1}{\eps}\frac{\theta_\eps^\lambda (D_t L_\eps^\lambda)^2}{2\dot{\phi}_\eps^\lambda}\sin(4\eps^{-1}\phi^\lambda_\eps) = \frac{\theta_\eps^\lambda (D_t L_\eps^\lambda)^2}{8(\dot{\phi}_\eps^\lambda)^2}\frac{d}{dt}\cos(4\eps^{-1}\phi^\lambda_\eps)\\
        &=&\frac{d}{dt}\left( \frac{\theta_\eps^\lambda (D_t L_\eps^\lambda)^2}{8(\dot{\phi}_\eps^\lambda)^2}\cos(4\eps^{-1}\phi^\lambda_\eps)\right)-\frac{d}{dt}
        \left( \frac{\theta_\eps^\lambda (D_t L_\eps^\lambda)^2}{8(\dot{\phi}_\eps^\lambda)^2} \right)
        \cos(4\eps^{-1}\phi^\lambda_\eps),
    \end{IEEEeqnarray*}
    which we use to derive
    \begin{IEEEeqnarray*}{rCl}
        \frac{d}{dt}\left( \theta^{\lambda\eps}_{221}-[\theta_2^\lambda]_{21}^\eps \right)&=& \frac{d}{dt}\left( \frac{\theta_\eps^\lambda (D_t L_\eps^\lambda)^2}{8(\dot{\phi}_\eps^\lambda)^2}\cos(4\eps^{-1}\phi^\lambda_\eps) -[\theta_2^\lambda]_{21}^\eps\right)-\frac{d}{dt} \left( \frac{\theta_\eps^\lambda (D_t L_\eps^\lambda)^2}{8(\dot{\phi}_\eps^\lambda)^2} \right)\cos(4\eps^{-1}\phi^\lambda_\eps)\\
        &=& \dot{v}_1^{\lambda\eps} - \dot{u}_1^{\lambda\eps} \cos(\eps^{-1}\psi_1^{\lambda\eps}).
    \end{IEEEeqnarray*}
    Here, we identified
    \(v_1^{\lambda\eps} \coloneqq u_1^{\lambda\eps} \cos(\eps^{-1}\psi_1^{\lambda\eps}) - u_1^{\lambda 0}
    \cos(\eps^{-1}\psi_1^{\lambda 0})\), where
    \begin{equation*}
        u_1^{\lambda\eps} \coloneqq \frac{\theta_\eps^\lambda (D_t L_\eps^\lambda)^2}{8(\dot{\phi}_\eps^\lambda)^2},
        \qquad \psi_1^{\lambda\eps} \coloneqq 4\phi_\eps^\lambda,
    \end{equation*}
    according to Lemma~\ref{lemma:5}. The necessary assumptions on \(u_1^{\lambda\eps}\) and
    \(\psi_1^{\lambda\eps}\) are satisfied by~\eqref{eq:10} and Lemma~\ref{lemma:4}. Consequently, it follows
    that \(\dot{v}_1^{\lambda\eps}\weak{\ast}0\) in \(L^\infty([0,T])\). Moreover, since
    \(\{\dot{u}_1^{\lambda\eps} \cos(\eps^{-1}\psi_1^{\lambda\eps})\}\) is a bounded sequence in
    \(L^\infty([0,T])\),~\cite[Chapter I. Lemma 1]{Bornemann} implies the equivalence of the weak\(^\ast\)
    convergence of the sequence \(\{\dot{u}_1^{\lambda\eps} \cos(\eps^{-1}\psi_1^{\lambda\eps})\}\) and the
    integral convergence as in Lemma~\ref{lemma:3}. Hence, we reason that
    \(\dot{u}_1^{\lambda\eps} \cos(\eps^{-1}\psi_1^{\lambda\eps})\weak{\ast}0\) in \(L^\infty([0,T])\). We
    therefore conclude that
    \begin{equation*}
        \frac{d}{dt}\left( \theta^{\lambda\eps}_{221}-[\theta_2^\lambda]_{21}^\eps \right)
        \weak{\ast}\frac{d\bar{\theta}_{221}^\lambda}{dt}\coloneqq 0 \quad \text{in} \quad L^\infty([0,T]).
    \end{equation*}
    Finally, with \(\{\dot{v}_1^{\lambda\eps}\}\) and \(\{\dot{u}_1^{\lambda\eps}\}\) being bounded sequences
    in \(L^\infty([0,T])\), the convergence in~\eqref{eq:29} follows for a subsequence \(\{\eps^\prime\}\)
    from~\cite[Principle 4]{Bornemann}, i.e.,~an extended version of the Arzel\`a--Ascoli theorem.

    For \(j=2,\ldots,6\) we only summarise the results, since the arguments follow along the same lines as in the case
    \(j=1\) above. In particular, we always identify terms \(\dot{v}_j^{\lambda\eps}\weak{\ast}0\)
    in \(L^\infty([0,T])\), according to Lemma~\ref{lemma:5}, and terms
    \(\dot{u}_j^{\lambda\eps} \cos(\eps^{-1}\psi_j^{\lambda\eps})\weak{\ast}\dot{\bar{\theta}}_{22j}^\lambda\)
    \((j=2,\ldots,6)\) according to Lemma~\ref{lemma:3}. The cases \(j=7,8\) require some different reasoning
    and are thus explained in more detail.
    \paragraph{Case \(\bm{j=2}\):}
    \begin{IEEEeqnarray*}{rcl}
        \frac{d}{dt}\left( \theta^{\lambda\eps}_{222}-[\theta_2^\lambda]_{22}^\eps \right) &=&\>\frac{d}{dt}\left( \frac{\theta_\eps^\lambda \cdot D_t L_\eps^\lambda \cdot D_t \omega_\lambda(y_\eps)}{4(\dot{\phi}_\eps^\lambda)^3}\cos(2\eps^{-1}\phi^\lambda_\eps)-[\theta_2^\lambda]_{22}^\eps \right)-\frac{d}{dt}\left( \frac{\theta_\eps^\lambda \cdot D_t L_\eps^\lambda \cdot D_t \omega_\lambda(y_\eps)}{4(\dot{\phi}_\eps^\lambda)^3} \right) \cos(2\eps^{-1}\phi^\lambda_\eps) \\
        &\weak{\ast}&\> \frac{d\bar{\theta}^{\lambda}_{222}}{dt}\coloneqq \frac{\theta_\ast^\lambda(D_t L_0^\lambda)^3}{2\omega^2_\lambda(y_0)}- \frac{(\theta_\ast^\lambda)^2 \vert  DL_0^\lambda \vert^2 \cdot D_t L_0^\lambda}{4\omega_\lambda(y_0)} \quad \text{in} \quad L^\infty([0,T]).
    \end{IEEEeqnarray*}
    \paragraph{Case \(\bm{j=3}\):}
    \begin{IEEEeqnarray*}{rCl}
        \frac{d}{dt}\left( \theta^{\lambda\eps}_{223}-[\theta_2^\lambda]_{23}^\eps \right)&=& -\frac{d}{dt}\left( \frac{\theta_\eps^\lambda \left\langle D^2L_\eps^\lambda \dot{y}_\eps,\dot{y}_\eps\right\rangle }{4(\dot{\phi}_\eps^\lambda)^2}\cos(2\eps^{-1}\phi^\lambda_\eps)+[\theta_2^\lambda]_{23}^\eps \right)+ \frac{d}{dt}\left( \frac{\theta_\eps^\lambda \left\langle D^2L_\eps^\lambda \dot{y}_\eps,\dot{y}_\eps\right\rangle }{4(\dot{\phi}_\eps^\lambda)^2} \right)\cos(2\eps^{-1}\phi^\lambda_\eps) \\
        &\weak{\ast}& \frac{d\bar{\theta}^{\lambda}_{223}}{dt}\coloneqq \frac{(\theta_\ast^\lambda)^2 \left\langle D_t DL_0^\lambda, DL_0^\lambda \right\rangle }{4\omega_\lambda(y_0)}- \frac{3\theta_\ast^\lambda\left\langle D^2 L_0^\lambda \dot{y}_0, \dot{y}_0 \right\rangle \cdot D_t L_0^\lambda }{8\omega^2_\lambda(y_0)}\quad \text{in} \quad L^\infty([0,T]).
    \end{IEEEeqnarray*}
    \paragraph{Case \(\bm{j=4}\):}
    \begin{IEEEeqnarray*}{rCl}
        \frac{d}{dt}\left( \theta_{224}^{\lambda\eps} - [\theta_2^\lambda]_{24}^\eps\right)&=&\frac{d}{dt}\left( \frac{\theta_\eps^\lambda \left\langle DV(y_\eps), DL_\eps^\lambda\right\rangle }{4(\dot{\phi}_\eps^\lambda)^2}\cos(2\eps^{-1}\phi_\eps^\lambda) - [\theta_2^\lambda]_{24}^\eps \right) - \frac{d}{dt}\left( \frac{\theta_\eps^\lambda \left\langle DV(y_\eps), DL_\eps^\lambda \right\rangle }{4(\dot{\phi}_\eps^\lambda)^2} \right)\cos(2\eps^{-1}\phi_\eps^\lambda)\\
        &\weak{\ast}&\frac{d\bar{\theta}_{224}^\lambda}{dt}\coloneqq \frac{3 \theta_\ast^\lambda \left\langle DV(y_0), DL_0^\lambda \right\rangle \cdot D_t L_0^\lambda }{8 \omega_\lambda^2(y_0)} \quad \text{in} \quad L^\infty([0,T]).
    \end{IEEEeqnarray*}
    \paragraph{Case \(\bm{j=5}\):}
    \begin{IEEEeqnarray*}{rCl}
        \frac{d}{dt}\left( \theta^{\lambda\eps}_{225}-[\theta_2^\lambda]_{25}^\eps \right)&=& \frac{d}{dt}\left( \sum_{\mu=1}^r \frac{\theta_\eps^\lambda \theta_\eps^\mu \left\langle D \omega_\mu (y_\eps), DL_\eps^\lambda \right\rangle }{4(\dot{\phi}_\eps^\lambda)^2}\cos(2\eps^{-1}\phi^\lambda_\eps)-[\theta_2^\lambda]_{25}^\eps \right)\\
        &&-\>\sum_{\mu=1}^r \frac{d}{dt}\left( \frac{\theta_\eps^\lambda \theta_\eps^\mu \left\langle D \omega_\mu (y_\eps), DL_\eps^\lambda \right\rangle }{4(\dot{\phi}_\eps^\lambda)^2} \right)\cos(2\eps^{-1}\phi^\lambda_\eps)\\
        &\weak{\ast}&\frac{d\bar{\theta}^{\lambda}_{225}}{dt} \coloneqq \sum_{\mu=1}^r \frac{3\theta_\ast^\lambda\theta_\ast^\mu \left\langle D\omega_\mu(y_0), DL_0^\lambda \right\rangle \cdot D_t L_0^\lambda }{8\omega^2_\lambda(y_0)} +\frac{(\theta_\ast^\lambda)^2\vert DL_0^\lambda\vert^2 \cdot D_t L_0^\lambda}{8\omega_\lambda(y_0)}  \quad \text{in} \quad L^\infty([0,T]).
    \end{IEEEeqnarray*}
    \paragraph{Case \(\bm{j=6}\):}
    \begin{IEEEeqnarray*}{rCl}
        \frac{d}{dt}\left( \theta^{\lambda\eps}_{226}-[\theta_2^\lambda]_{26}^\eps \right)&=&-\frac{d}{dt} \left( \frac{(\theta_\eps^\lambda)^2 \left\langle D \omega_\lambda(y_\eps), DL_\eps^\lambda \right\rangle }{16(\dot{\phi}_\eps^\lambda)^2}\cos(4\eps^{-1}\phi^\lambda_\eps)+[\theta_2^\lambda]_{26}^\eps \right)\\
        &&+\> \frac{d}{dt} \left( \frac{(\theta_\eps^\lambda)^2 \left\langle D \omega_\lambda(y_\eps), DL_\eps^\lambda \right\rangle }{16(\dot{\phi}_\eps^\lambda)^2} \right) \cos(4\eps^{-1}\phi^\lambda_\eps)\\
        &\weak{\ast}& \frac{d\bar{\theta}_{226}^\lambda}{dt}\coloneqq 0 \quad \text{in} \quad L^\infty([0,T]).
    \end{IEEEeqnarray*}
    \paragraph{Case \(\bm{j=7}\):}
    Analogous to the previous cases we write
    \begin{IEEEeqnarray*}{rCl}
        \frac{d}{dt}\left( \theta^{\lambda\eps}_{227}-[\theta_2^\lambda]_{27}^\eps \right)&=&\frac{d}{dt}\left( \sum_{\substack{\mu=1\\\mu \neq \lambda}}^r \frac{\theta_\eps^\lambda\theta_\eps^\mu \left\langle D \omega_\mu(y_\eps), DL_\eps^\lambda \right\rangle }{8\dot{\phi}_\eps^\lambda}\left\{ \frac{\cos\left( 2\eps^{-1}(\phi^\mu_\eps-\phi^\lambda_\eps) \right)}{\dot{\phi}^\mu_\eps - \dot{\phi}^\lambda_\eps}-\frac{\cos\left( 2\eps^{-1}(\phi^\mu_\eps+\phi^\lambda_\eps) \right)}{\dot{\phi}^\mu_\eps + \dot{\phi}^\lambda_\eps}\right\}-[\theta_2^\lambda]_{27}^\eps \right)\\
        &&+\>\sum_{\substack{\mu=1\\\mu \neq \lambda}}^r \frac{d}{dt}\left( \frac{\theta_\eps^\lambda\theta_\eps^\mu \left\langle D \omega_\mu(y_\eps), DL_\eps^\lambda \right\rangle }{8\dot{\phi}_\eps^\lambda (\dot{\phi}^\mu_\eps + \dot{\phi}^\lambda_\eps)} \right) \cos\left( 2\eps^{-1}(\phi^\mu_\eps + \phi^\lambda_\eps) \right)\\
        &&-\>\sum_{\substack{\mu=1\\\mu \neq \lambda}}^r \frac{d}{dt}\left( \frac{\theta_\eps^\lambda\theta_\eps^\mu \left\langle D \omega_\mu(y_\eps), DL_\eps^\lambda \right\rangle }{8\dot{\phi}_\eps^\lambda (\dot{\phi}^\mu_\eps - \dot{\phi}^\lambda_\eps)} \right) \cos\left( 2\eps^{-1}(\phi^\mu_\eps - \phi^\lambda_\eps) \right).
    \end{IEEEeqnarray*}
    We identify, similar to the case \(j=1\), the first term on the right-hand
    side with \(\dot{v}_7^{\lambda\eps}\). Then, it follows from Lemma~\ref{lemma:5} that \(\dot{v}_7^{\lambda\eps}\weak{\ast}0\) in \(L^\infty([0,T])\). Moreover, we identify the summands in the remaining two sums on the right-hand side with functions
    \begin{equation*}
        \dot{u}_{7\pm}^{\lambda\mu\eps} \cos(\eps^{-1}\psi_{7\pm}^{\lambda\mu\eps}) \coloneqq \frac{d}{dt}\left( \frac{\theta_\eps^\lambda\theta_\eps^\mu \left\langle D \omega_\mu(y_\eps), DL_\eps^\lambda \right\rangle }{8\dot{\phi}_\eps^\lambda (\dot{\phi}^\mu_\eps \pm \dot{\phi}^\lambda_\eps)} \right) \cos\left( 2\eps^{-1}(\phi^\mu_\eps \pm \phi^\lambda_\eps) \right),
    \end{equation*}
    for \(\lambda,\mu = 1,\ldots, r\), \(\lambda\neq \mu\). Together with~\eqref{eq:12} we expand the time
    derivative in \(\dot{u}_{7\pm}^{\lambda\mu\eps}\) and find, based on the non-resonance Assumptions~\ref{ass:1}
    and~\ref{ass:2}, and Lemmas~\ref{lemma:2} and~\ref{lemma:7}, that
    \(\dot{u}_{7\pm}^{\lambda\mu\eps} \cos(\eps^{-1}\psi_{7\pm}^{\lambda\mu\eps})\weak{\ast}0\) in
    \(L^\infty([0,T])\). All together, we conclude that
    \begin{equation*}
        \frac{d}{dt}\left( \theta^{\lambda\eps}_{227}-[\theta_2^\lambda]_{27}^\eps \right)\weak{\ast} \frac{d\bar{\theta}_{227}^\lambda}{dt}\coloneqq 0 \quad \text{in} \quad L^\infty([0,T]).
    \end{equation*}
    Finally, we use~\cite[Principle 4]{Bornemann} to derive the uniform convergence for a subsequence
    \(\{\eps^\prime\}\) in~\eqref{eq:29}.
    \paragraph{Case \(\bm{j=8}\):}
    The convergences in~\eqref{eq:29a} follow for a subsequence \(\{\eps^\prime\}\) (not relabelled)
    from~\cite[Principle 4]{Bornemann} and
    \begin{equation*}
        \frac{d\theta^{\lambda\eps}_{228}}{dt}  =  -\frac{\theta_\eps^\lambda \cdot D_t L_\eps^\lambda \cdot D_t^2
            L_\eps^\lambda}{8(\dot{\phi}_\eps^\lambda)^2} +\frac{\theta_\eps^\lambda \cdot D_t L_\eps^\lambda \cdot D_t^2
            L_\eps^\lambda}{8(\dot{\phi}_\eps^\lambda)^2}\cos(4\eps^{-1}\phi^\lambda_\eps) \weak{\ast}
        \frac{d\bar{\theta}^{\lambda}_{228}}{dt}\coloneqq - \frac{\theta_\ast^\lambda \cdot D_t L_0^\lambda \cdot D_t^2
            L_0^\lambda}{8\omega_\lambda^2(y_0)}  \quad \text{in} \quad L^\infty([0,T]),
    \end{equation*}
    where we used Lemma~\ref{lemma:2} for the weak\(^\ast\) convergence.

    Now, by combining the cases \(j=1,\ldots, 8\) notice that
    \begin{IEEEeqnarray*}{rCl}
        \sum_{j=1}^8\frac{d\bar{\theta}_{22j}^\lambda}{dt}&=& \frac{\theta_\ast^\lambda(D_t L_0^\lambda)^3}{2\omega^2_\lambda(y_0)}- \frac{\theta_\ast^\lambda \cdot D_t L_0^\lambda \cdot D_t^2 L_0^\lambda}{2\omega_\lambda^2(y_0)}+\frac{(\theta_\ast^\lambda)^2 \left\langle D_t DL_0^\lambda, DL_0^\lambda \right\rangle }{4\omega_\lambda(y_0)}- \frac{(\theta_\ast^\lambda)^2 \vert  DL_0^\lambda \vert^2 \cdot D_t L_0^\lambda}{8\omega_\lambda(y_0)}\\
        &=&- \frac{d}{dt}\left( \frac{\theta_\ast^\lambda (D_t L_0^\lambda)^2}{4\omega_\lambda^2(y_0)} -\frac{(\theta_\ast^\lambda)^2 \left\vert DL_0^\lambda \right\vert^2}{8\omega_\lambda(y_0)}\right),
    \end{IEEEeqnarray*}
    where we used
    \begin{equation*}
        D_t^2 L_0^\lambda = \left\langle D^2 L_0^\lambda \dot{y}_0, \dot{y}_0 \right\rangle -\left\langle DV(y_0),DL_0^\lambda
        \right\rangle -\sum_{\mu=1}^r\theta_\ast^\mu \left\langle D\omega_\mu(y_0), DL_0^\lambda \right\rangle.
    \end{equation*}
    Finally, Equation~\eqref{eq:27b} follows with
    \begin{equation}
        \label{eq:31}
        \frac{d\bar{\theta}_2^\lambda}{dt}=\sum_{i=1}^2\frac{d\bar{\theta}_{2i}^\lambda}{dt} = \sum_{j=1}^4
        \frac{d\bar{\theta}_{21j}^\lambda}{dt}+\sum_{j=1}^8\frac{d\bar{\theta}_{22j}^\lambda}{dt}.
    \end{equation}
\end{proof}

\begin{remark}
    It would be desirable to complement the uniform convergence result
    in~\eqref{eq:57} by a weak\(^\ast\)
    convergence result of the form
    \begin{equation*}
        \frac{d}{dt}\left(  \theta_{214}^{\lambda\eps} -[\theta^\lambda_2]_{14}^\eps \right)\weak{\ast} \frac{d
            \bar{\theta}^\lambda_{214}}{dt} \quad \text{in} \quad L^\infty([0,T]).
    \end{equation*}
    This would allow us to extend the uniform convergence result in~\eqref{eq:26} by a weak\(^\ast\) convergence
    as in~\eqref{eq:19}--\eqref{eq:22}. To this end one would need to show in the proof of Lemma~\ref{lemma:9},
    part \(i=1\), case \(j=4\), that \(u_4^{\lambda\eps} = \mathcal{O}(\eps)\). To do so one would need to
    extend Lemma~\ref{lemma:4} and show that the sequence
    \(\{\eps^{-1}\left( \phi_2^{\lambda\eps}-\left(\bar{\phi}_2^\lambda+[\phi_2^\lambda]^\eps\right)\right)\}\)
    is bounded in \(L^\infty([0,T])\). This would require more notation and would significantly increase the
    complexity of this article. We therefore do not pursue this analysis further.
\end{remark}

\begin{lemma}
    \label{lemma:10}
    The extraction of a subsequence in Lemmas~\ref{lemma:8} and~\ref{lemma:9} can be discarded altogether and
    \((\bar{\phi}_2, \bar{\theta}_2, \bar{y}_2,\bar{p}_2)\) is the unique solution to the initial value
    problem~\eqref{eq:27}--\eqref{eq:30}.
\end{lemma}

\begin{proof}
    The differential equations~\eqref{eq:27a},~\eqref{eq:27c} and~\eqref{eq:27d} follow
    from~\eqref{eq:23}--\eqref{eq:25} by taking the weak\(^\ast\) limit in
    combination with Lemmas~\ref{lemma:3}
    and~\ref{lemma:5}, and~\cite[Lemma~1]{Bornemann}. Formula~\eqref{eq:27b} follows from~\eqref{eq:31}. The
    initial values~\eqref{eq:30} can be derived from the uniform convergences in~\eqref{eq:20}--\eqref{eq:26}.
    Furthermore, since the right-hand side of~\eqref{eq:27} --- and therefore the solution
    \( (\bar{\phi}_2, \bar{\theta}_2, \bar{y}_2, \bar{p}_2) \in C^\infty([0,T],\R^{2m})\) --- does not depend on
    the chosen subsequence,~\cite[Principle~5]{Bornemann} allows us to discard the extraction of a subsequence
    altogether.
\end{proof}

\subsection{Higher-order asymptotic expansion and restrictions on the timescale}

\label{sec:High-order-asympt}
In the following we summarise how to derive higher-order asymptotic expansions
of the solution
to~\eqref{eq:8}--\eqref{eq:9}. Let us assume that we know the asymptotic
expansion up to order \(k-1\) and we want to
derive the asymptotic expansion to \(k\)th order, i.e.,~for \(u_\eps\)
representing the functions
\(\phi_\eps\), \(\theta_\eps\), \(y_\eps\) or \(p_\eps\), we are looking for an
asymptotic expansion of the
form
\begin{equation*}
    u_\eps = u_0 + \sum_{\ell=1}^{k-1} \eps^\ell[\bar{u}_\ell]^\eps + \eps^k[\bar{u}_k]^\eps + \eps^k u_{k+1}^\eps,\end{equation*}
where for \(\ell = 1, \ldots, k\),
\begin{equation*}
    [\bar{u}_\ell]^\eps \coloneqq \bar{u}_\ell + [u_\ell]^\eps\weak{\ast} \bar{u}_\ell \quad \text{in}
    \quad L^\infty([0,T]),\qquad u_{k+1}^\eps \to 0 \quad \text{in}\quad C([0,T]).
\end{equation*}
Two approaches can be used to derive the function \([\bar{u}_k]^\eps\). They both rely
on analysing the
leading-order asymptotic expansion of
\begin{equation*}
    u_k^{\eps}\coloneqq\frac{u_\eps-u_0}{\eps^k} - \sum_{\ell=1}^{k-1}\eps^{\ell-k} [\bar{u}_\ell]^\eps.
\end{equation*}
The first approach relies on deriving \([\bar{u}_k]^\eps\) directly from \(u_k^\eps\) by
applying the
fundamental theorem of calculus to the function \(u_\eps - u_0\) and subsequently
integrating the oscillatory
component of the integrand \(\dot{u}_\eps - \dot{u}_0\) by parts to lower the exponent of the
denominator
\(\eps^k\). After \(k\) iterations by parts and corresponding
expansions of the resulting terms, \(\bar{u}_k\)
and \([u_k]^\eps\) can then be derived such that
\begin{equation*}
    u_k^{\eps} - [u_k]^\eps \to \bar{u}_k \quad \text{in} \quad C([0,T]).
\end{equation*}
This method was used to derive the leading-order asymptotic expansion of
\(\theta_{21}^{\eps}\) in
Lemma~\ref{lemma:9}.

Another approach for the derivation of \([\bar{u}_k]^\eps\) is based on an application
of the extended
Arzel\`a--Ascoli theorem. Analogous to  Lemma~\ref{lemma:4}, one shows first that the
sequence \(\{u_k^\eps\}\)
is uniformly bounded in \(L^\infty([0,T])\). Then, by Alaoglu's
theorem~\cite[Principle 3]{Bornemann}, there
exists a subsequence \(\{\eps^\prime\}\) such that \(u_k^\eps \weak{\ast} \bar{u}_k\) in
\(L^\infty([0,T])\). To determine \(\bar{u}_k\), one chooses \([u_k]^\eps\) such that the
sequence
\(\{u_k^\eps-[u_k]^\eps\}\) is uniformly bounded in \(C^{0,1}([0,T])\). Then, according to the
extended
Arzel\`a--Ascoli theorem~\cite[Chapter~I~\textsection 1]{Bornemann}, there exists
a
subsequence such that
\begin{equation*}
    u_k^{\eps} - [u_k]^\eps \to \bar{u}_k \quad \text{in} \quad C([0,T]), \qquad \frac{d}{dt}\left( u_k^{\eps}
    - [u_k]^\eps \right)\weak{\ast} \frac{d \bar{u}_k}{dt} \quad \text{in} \quad L^\infty([0,T]),\end{equation*}
from which \(\bar{u}_k\) can be determined as the solution to a system of
differential equations. This
approach was used to derive the leading-order asymptotic expansion of
\(\phi_2^\eps\), \(y_2^\eps\),
\(p_2^\eps\) in  Lemma~\ref{lemma:8} and \(\theta_{22}^{\eps}\) in  Lemma~\ref{lemma:9}.

\begin{remark}Theorem~\ref{thm:2} provides immediately quantitative estimates on the difference
    between the original
    system \((\phi_\eps, \theta_\eps, y_\eps, p_\eps)\) and the limit system \((\phi_0,\theta_0,y_0,p_0)\) of
    order \(\mathcal{O}(\eps)\) for times up to arbitrary, but fixed \(T\). With the
    second-order asymptotic expansions
    \(\phi_0 +\eps^2 (\bar{\phi}_2+[\phi_2])\), \(\theta_0+\eps [\theta_1]+\eps^2(\bar{\theta}_2+[\theta_2])\), \(y_0+
    \eps^2 (\bar{y}_2+[y_2])\) and \(p_0+ \eps^2 (\bar{p}_2+[p_2])\) the result provides
    error estimates of order better than
    \(\mathcal{O}(\eps^2)\) over the same timescale. There are other averaging approaches which
    deal with the differential
    equation only. A formal expansion in \(\eps\) can also derive the equations
    for the averaged second-order corrections
    \(\bar{\phi}_2, \bar{\theta}_2,\bar{y}_2,\bar{p}_2\), then error estimates need to be obtained separately, e.g.,~using
    some Gronwall and integration by parts arguments as described
    in~\cite{SVM07}. The general
    restriction to finite timescales cannot be avoided unless the averaged
    correction terms vanish~\cite[Chap. 2]{SVM07}, which
    does not hold in our situation.
\end{remark}

\section{Thermodynamic interpretation}
\label{sec:4}

We now give a thermodynamic interpretation of the analytic result presented in Theorem~\ref{thm:2}. The model
problem in Section~\ref{sec:8} describes the interaction of \(r\) (in general non-ergodic) fast and \(n\) slow
degrees of freedom \((n,r\in \mathbb{N})\). A simplified model of one fast (hence, ergodic) and one slow
degree of freedom was already studied in~\cite{Klar2020}, where the authors similarly interpret a fast--slow
system of the kind presented in Section~\ref{sec:8} from a thermodynamic point of view. Since the
thermodynamic interpretation of the model studied in~\cite{Klar2020} includes arguments that are similarly
applicable to the more general model considered in this article, we will focus here on the differences and
refer the interested reader for a detailed thermodynamic discussion to~\cite{Klar2020}.

Fundamental in the theory of classical equilibrium thermodynamics is the transfer of energy in the form of work and heat in
thermodynamic processes. This energy transfer is described by the energy
relation
\begin{equation}
    \label{eq:61}
    d E = dW + dQ = \sum_{j=1}^n F^j dy^j + T dS.
\end{equation}
In more detail, let \(E\) be the energy of a generic thermodynamic system
composed of many fast particles,
such as gas particles trapped in a container with a piston. Then, the change of the system's energy \(dE\) is the sum of
external work done on the system, \(dW = \sum_{j=1}^n F^j dy^j\),
where \(F^j\) are external forces exerted on the system by infinitesimal
displacements of some external slow
variables \(dy^j\), and a change of heat, \(dQ = TdS\), where \(T\)
is the system's temperature and \(dS\) a
change of entropy. Classical statistical mechanics provides the derivation of
thermodynamic quantities such as
temperature, entropy and external forces as the slow, average macroscale
observations from the microscale
dynamics in the system.

The energy transfer within a thermodynamic system in the form of work and heat also
applies to mechanical systems
which evolve within an environment that allow for thermodynamic interactions. A
suitable thermodynamic theory
for such mechanical systems was developed by L.~Boltzmann and later refined by
G.~W.~Gibbs~\cite{Gibbs1901},
which was subsequently rederived by Hertz~\cite{Hertz1910}. We will follow
Hertz' line of thought. His
formalisation is based on fast Hamiltonian systems that are slowly perturbed by
external agents. In this
setting, his theory describes how to define temperature, entropy and external
forces such that the fundamental
thermodynamic energy relation~\eqref{eq:61} is satisfied.

Applying Hertz' thermodynamic formalism to the model problem introduced in
Section~\ref{sec:8}, we regard,
similar to~\cite{Klar2020}, the subsystem composed of the fast degrees of
freedom \(z_\eps^\lambda\)
\((\lambda=1, \ldots, r)\) as a thermodynamic system that is slowly perturbed by the
interactions with the
slow subsystem composed of \(y_\eps^j\) \((j=1,\ldots, n)\). Note that the ergodicity
assumption for
thermodynamic systems is not given for the fast subsystem. Nevertheless, one can still
derive thermodynamic properties
if one replaces time-averages with ensemble-averages (see~\cite{Berdichevsky}),
which can be derived by
averaging the trajectories not only with respect to time but also with respect
to initial values assumed to be
uniformly distributed over the energy surface. A more detailed explanation can
be found in
Appendix~\ref{app:1}.

In contrast to classical thermodynamic theory, which mainly focuses on the
thermodynamic analysis of some fast
dynamics that experiences some slow external influence, such as gas particles trapped in a container with a piston,
our focus lies in analysing
the slow dynamics that experiences some external thermodynamic effects through
its interaction with the fast
subsystem. This focus is motivated, for instance, by the conformal motion of a
molecule in a solvent.

We will focus our attention on the energy associated to the fast degrees of
freedom \(E_\eps^\perp\) and the
residual energy \(E_\eps^\parallel\), which are given by
\begin{equation}
    \label{eq:72}
    E_\eps^\perp = \frac{1}{2}\vert \dot{z}_\eps \vert^2 + \frac{1}{2}\eps^{-2}\sum_{\lambda=1}^r \omega_\lambda^2(y_\eps)
    (z_\eps^\lambda)^2, \qquad E_\eps^\parallel = E_\eps - E_\eps^\perp.
\end{equation}
The evolution of the fast degrees of freedom is governed by the energy
\(E_\eps^\perp = E_\eps^\perp(z_\eps, \dot{z}_\eps; y_\eps)\) which is subject to slowly varying external
parameters given by \(y_\eps^j\) \((j=1, \ldots, n)\). As pointed out
in~\cite{Klar2020}, this framework
allows us to interpret the model problem from a thermodynamic point of view by
applying the thermodynamic
theory of Hertz~\cite{Hertz1910}.

By applying Hertz' thermodynamic formalism to the fast subsystem, which is
governed by the energy function
\(E_\eps^\perp\), we derive in Appendix~\ref{app:1}, provided that \(\theta_\ast^\lambda\neq 0\) for at
least one
\(\lambda=1,\ldots, r\), the following expressions for the temperature \(T_\eps\), the
entropy \(S_\eps\) and the external force \(F_\eps\):
\begin{equation}
    \label{eq:32}
    T_\eps = \frac{1}{r}\sum_{\lambda=1}^r \theta_\eps^\lambda \omega_\lambda(y_\eps),\qquad
    S_\eps =\sum_{\lambda=1}^r \log\left( \sum_{\mu=1}^r \theta_\eps^\mu \frac{\omega_\mu(y_\eps)}{\omega_\lambda(y_\eps)} \right),
    \qquad F_\eps = T_\eps \sum_{\lambda = 1}^r D L_\eps^\lambda,
\end{equation}
where, according to the notation introduced in~\eqref{eq:755}, the vector
\(DL_\eps^\lambda\) represents the gradient of \(L_\eps^\lambda=\log(\omega_\lambda(y_\eps))\) with respect to
\(y_\eps\in \R^n\). The classical thermodynamic concepts of temperature and entropy
are commonly described for systems in or near thermodynamic equilibrium,
i.e.,~for an infinite separation of timescales, so in the limit \(\eps\to0\). It
is
noteworthy that we derive these expressions for finite but non-zero
\(\eps\). Note that the assumption on
\(\theta_\ast^\lambda\) ensures, that the system exhibits a genuine scale-separation into
fast and slow
dynamics. Moreover, note that the temperature is the arithmetic mean of the
frequencies
\(\omega_\lambda(y_\eps)\) \((\lambda=1,\ldots, r)\) weighted by their corresponding actions
\(\theta_\eps^\lambda\). The entropy provides a measure for the pairwise weighted frequency ratios
\(\theta_\eps^\lambda \omega_\lambda(y_\eps)/ \omega_\mu(y_\eps)\) \((\lambda,\mu = 1,\ldots, r)\), while the
external force primarily indicates the change of \(\log(\omega_\lambda(y_\eps))\) with respect to the slow
coordinates \(y_\eps\).

In combination with the second-order expansion derived in Theorem~\ref{thm:2} we can expand \(T_\eps\),
\(S_\eps\) and \(F_\eps\), and thus determine their asymptotic properties, i.e.,~\(T_\eps= T_0 +\mathcal{O}(\eps)\), \(F_\eps= F_0+\mathcal{O}(\eps)\) and
\(S_\eps=S_0 + \eps [\bar{S}_1]^\eps +\eps^2 [\bar{S}_2]^\eps + \eps^2 S_3^\eps\) with \(S_3^\eps \to 0\) in
\(C([0,T])\), where
\begin{equation}
    \label{eq:59}
    T_0 \coloneqq  \frac{1}{r}\sum_{\lambda=1}^r \theta_\ast^\lambda \omega_\lambda(y_0), \quad F_0 \coloneqq T_0 \sum_{\lambda=1}^r D L_0^\lambda, \quad S_0 \coloneqq \sum_{\lambda=1}^r \log\left( \sum_{\mu=1}^r \theta_\ast^\mu \frac{\omega_\mu(y_0)}{\omega_\lambda(y_0)} \right),\quad [\bar{S}_1]^\eps \coloneqq \frac{1}{T_0} \sum_{\lambda=1}^r [\theta_1^\lambda]^\eps \omega_\lambda(y_0),
\end{equation}
and
\begin{IEEEeqnarray*}{rCl}
    [\bar{S}_2]^\eps &\coloneqq& \frac{1}{T_0}\sum_{\lambda=1}^r \left(\bar{\theta}_2^\lambda+[\theta_2^\lambda]^\eps \right) \omega_\lambda(y_0)+\frac{1}{T_0}\sum_{\lambda=1}^r \theta_\ast^\lambda \left\langle D\omega_\lambda(y_0),\bar{y}_2+[y_2]^\eps \right\rangle\\
    && -\> \sum_{\lambda=1}^r \left\langle DL_0^\lambda, \bar{y}_2+[y_2]^\eps \right\rangle-\frac{1}{2rT_0^2}  \left(\sum_{\lambda=1}^r [\theta_1^\lambda]^\eps \omega_\lambda(y_0)\right)^2.
\end{IEEEeqnarray*}
We use the expansions derived in Section~\ref{sec:5} to analyse the energy \(E_\eps^\perp\) on different
scales. To this end, we expand \(E_\eps^\perp =\sum_{\lambda=1}^r \theta_\eps^\lambda\omega_\lambda(y_\eps)\)
and write
\(E_\eps^\perp = E_0^\perp + \eps [\bar{E}_1^\perp]^\eps + \eps^2 [\bar{E}_2^\perp]^\eps + \eps^2
E_3^{\perp\eps}\) with \(E_3^{\perp\eps} \to 0 \) in \(C([0,T])\), where
\begin{IEEEeqnarray*}{C}
    E_0^\perp \coloneqq \sum_{\lambda=1}^r \theta_\ast^\lambda \omega_\lambda(y_0),\qquad [\bar{E}_1^\perp]^\eps\coloneqq \sum_{\lambda=1}^r\left[ \theta_1^\lambda \right]^\eps \omega_\lambda(y_0),\\
    \phantom{}[\bar{E}_2^\perp]^\eps \coloneqq \sum_{\lambda=1}^r \left( \bar{\theta}_2^\lambda + [\theta_2^\lambda]^\eps  \right) \omega_\lambda(y_0) + \sum_{\lambda=1}^r \theta_\ast^\lambda \left\langle D \omega_\lambda(y_0),\bar{y}_2 + [y_2]^\eps \right\rangle.
\end{IEEEeqnarray*}

\subsection{Leading-order thermodynamics}
\label{sec:Lead-order-therm}
We now analyse the energy \(E_\eps^\perp\) in the limit \(\eps\to0\) from a thermodynamic perspective. For
\(\eps\to0\), the temperature, entropy and external force are given by the expressions \(T_0\), \(S_0\) and
\(F_0\) as in~\eqref{eq:59}.

While the temperature captures the average collective dynamics of the weighted frequencies
\(\theta_\ast^\lambda\omega_\lambda(y_0)\), the entropy depends on the dynamics of the weighted frequency
ratios \(\theta_\ast^\lambda \omega_\lambda(y_0)/\omega_\mu(y_0)\). In contrast to the simplified model
in~\cite{Klar2020}, which can be regarded as the degenerate case of one fast degree of freedom, the entropy
\(S_0\) is constant if and only if all weighted frequency ratios
\(\theta_\ast^\lambda \omega_\lambda(y_0)/\omega_\mu(y_0)\) \((\lambda, \mu = 1, \ldots, r)\) are constant,
regardless of the number of fast degrees of freedom. In this case, the motion of
the fast degrees of freedom
can be described as a quasi-periodic motion. Thus, the entropy can be considered
as an indicator of the
homogeneity of the frequencies with respect to \(y_0\) and therefore serves
as a measure of chaos for the fast
subsystem. In the case of one fast degree of freedom as in~\cite{Klar2020}, the
weighted frequency ratio is
naturally constant and hence the entropy remains constant. Therefore, we can regard
-- in reference to
classical thermodynamic theory -- the leading-order dynamics of the fast
subsystem in the case of a constant
entropy as an adiabatic thermodynamic process and non-constant entropy as a
non-adiabatic thermodynamic
process. We remark that we make this thermodynamic interpretation despite the
fact that the fast subsystem is
non-ergodic.

Finally, by expressing the leading-order energy of the fast subsystem
\(E_0^\perp =\sum_{\lambda=1}^r \theta_\ast^\lambda\omega_\lambda(y_0)\) as a function of \(S_0\) and \(y_0\),
it can be written as
\begin{equation*}
    E_0^\perp(S_0, y_0)=e^{S_0/r}\prod_{\lambda=1}^r \omega_\lambda^{1/r}(y_0).
\end{equation*}
As a consequence, the differential is given by
\begin{equation}
    \label{eq:first-order-thermo}
    d E_0^\perp = \sum_{j=1} ^n F_0^j dy_0^j + T_0 dS_0,
\end{equation}
which coincides with the fundamental thermodynamic energy relation in~\eqref{eq:61}.

\subsection{Second-order thermodynamics}
\label{sec:Second-order-therm}
In contrast to the \(\eps\)-independent thermodynamic expressions to
leading-order discussed in
Section~\ref{sec:Lead-order-therm}, the asymptotic expansion terms to higher-order are
\(\eps\)-dependent. In particular, they contain terms that rapidly oscillate
around zero, and terms that yield
the average motion of the higher-order asymptotic expansions. As the
thermodynamic theory aims to describe
many-particle systems by their average dynamics, we analyse, similar
to~\cite{Klar2020}, the average
dynamics of the higher-order asymptotic expansion in \(E_\eps^\perp\) and
\(S_\eps\) by studying the
weak\(^\ast\) limit of \( [\bar{E}_1^\perp]^\eps\), \([\bar{E}_2^\perp]^\eps\), \([\bar{S}_1]^\eps\) and
\([\bar{S}_2]^\eps\), i.e.,
\begin{IEEEeqnarray*}{rClCl+rClCl}
    [\bar{E}_1^\perp]^\eps &\weak{\ast}& 0 \quad &\text{in}& \quad L^\infty([0,T]),& [\bar{E}_2^\perp]^\eps &\weak{\ast}& \bar{E}_2^\perp \quad&\text{in}& \quad L^\infty([0,T]),\\
    \phantom{}[\bar{S}_1]^\eps &\weak{\ast}& 0 \quad &\text{in}& \quad L^\infty([0,T]),& [\bar{S}_2]^\eps &\weak{\ast}& \bar{S}_2 \quad &\text{in}& \quad L^\infty([0,T]),
\end{IEEEeqnarray*}
where
\begin{equation*}
    \bar{E}_2^\perp \coloneqq \sum_{\lambda=1}^r  \bar{\theta}_2^\lambda  \omega_\lambda(y_0) + \sum_{\lambda=1}^r \theta_\ast^\lambda \left\langle D \omega_\lambda(y_0),\bar{y}_2  \right\rangle
\end{equation*}
and
\begin{IEEEeqnarray*}{rCl}
    \bar{S}_2&\coloneqq&\frac{1}{T_0}\sum_{\lambda=1}^r \bar{\theta}_2^\lambda \omega_\lambda(y_0)+\frac{1}{T_0}\sum_{\lambda=1}^r \theta_\ast^\lambda \left\langle D\omega_\lambda(y_0),\bar{y}_2 \right\rangle-  \sum_{\lambda=1}^r \left\langle DL_0^\lambda, \bar{y}_2 \right\rangle -\frac{1}{16rT_0^2} \sum_{\lambda=1}^r \left( \theta_\ast^\lambda\cdot D_t L_0^\lambda \right)^2\\
    \IEEEyesnumber \label{eq:60}
    &=& \frac{\bar{E}_2^\perp}{T_0}-  \sum_{\lambda=1}^r \left\langle DL_0^\lambda, \bar{y}_2 \right\rangle -\frac{1}{16rT_0^2} \sum_{\lambda=1}^r \left( \theta_\ast^\lambda\cdot D_t L_0^\lambda \right)^2.
\end{IEEEeqnarray*}
Note that the expression of the entropy is in this case not constant.
Intuitively, this follows from the
second-order asymptotic expansion of the slow degrees of freedom \(y_\eps^j\)
\((j=1,\ldots,n)\). These
exhibit according to Theorem~\ref{thm:2} a decomposition into slowly varying
components \(\bar{y}_2^j\) and
rapidly varying components \([y_2^j]^\eps\). The existence of this decomposition
gives rise to a non-constant
entropy discussed in more detail in~\cite{Klar2020}. Moreover, we notice that
the last term in \(\bar{S}_2\)
originates from \([\bar{S}_1]^\eps\), the rapidly oscillating first-order correction of
\(S_0\).

Finally, after rearranging~\eqref{eq:60}, we derive for
\(\bar{E}_2^\perp = \bar{E}_2^\perp(\bar{S}_2,\bar{y}_2; y_0, p_0)\) the expression
\begin{equation*}
    \bar{E}_2^\perp \coloneqq \left\langle F_0, \bar{y}_2 \right\rangle + T_0 \bar{S}_2 +  \frac{1}{16rT_0} \sum_{\lambda=1}^r \left( \theta_\ast^\lambda\cdot D_t L_0^\lambda \right)^2.
\end{equation*}
Now, by analysing the differential of \(\bar{E}_2^\perp\) for fixed \((y_0, p_0)\), which
is given by
\begin{equation*}
    d \bar{E}^\perp_2 = \sum_{j=1}^n F_0^j d \bar{y}_2^j + T_0 d \bar{S}_2,\end{equation*}
we find, similar to~\eqref{eq:first-order-thermo}, a remarkable resemblance to
the fundamental
thermodynamic relation as presented in~\eqref{eq:61}.

\subsection{Analysis of the total energy}

Finally, we inspect how the thermodynamic energy transfer in form of work and
heat is realised in the
second-order asymptotic expansion of the total energy \(E_\eps\). Recalling
the analysis above, we split the
total energy \(E_\eps\) into \(E_\eps^\perp\) and \(E_\eps^\parallel\) (compare
with~\eqref{eq:72}), i.e.,~\(E_\eps=E_\eps^\perp+E_\eps^\parallel\), where
\begin{equation*}
    E^\parallel_\eps \coloneqq \frac{1}{2} \vert p_\eps\vert^2 +V(y_\eps) +\frac{\eps}{2}\sum_{\lambda=1}^r \theta_\eps^\lambda \left\langle p_\eps, DL_\eps^\lambda \right\rangle \sin(2\eps^{-1}\phi_\eps^\lambda)+\frac{\eps^2}{8}\sum_{\lambda=1}^r\sum_{\mu=1}^r\theta_\eps^\lambda\theta_\eps^\mu\left\langle DL_\eps^\lambda, DL_\eps^\mu \right\rangle \sin(2\eps^{-1}\phi_\eps^\lambda)\sin(2\eps^{-1}\phi_\eps^\mu).
\end{equation*}
Similar to before, we use the expressions derived in Theorem~\ref{thm:2} to expand the
energy
\(E_\eps^\parallel\), i.e.,~\(E_\eps^\parallel = E_0^\parallel + \eps [\bar{E}_1^\parallel]^\eps + \eps^2 [\bar{E}_2^\parallel]^\eps +
\eps^2 E_3^{\parallel\eps}\) with \(E_3^{\parallel\eps} \to 0 \) in \(C([0,T])\), where
\begin{equation*}
    E_0^\parallel \coloneqq \frac{1}{2}\vert p_0\vert^2 +V(y_0),\qquad [\bar{E}_1^\parallel]^\eps
    \coloneqq \frac{1}{2}\sum_{\lambda=1}^r \theta_\ast^\lambda D_t L_0^\lambda\sin(2\eps^{-1}\phi_0^\lambda)
\end{equation*}
and
\begin{IEEEeqnarray*}{rCl}
    [\bar{E}_2^\parallel]^\eps &\coloneqq& \left\langle p_0, \bar{p}_2+[p_2]^\eps \right\rangle +\left\langle DV(y_0), \bar{y}_2 + [y_2]^\eps \right\rangle +\frac{1}{2}\sum_{\lambda=1}^r [\theta_1^\lambda]^\eps D_tL_0^\lambda \sin(2\eps^{-1}\phi_0^\lambda) \\
    &&+\>\sum_{\lambda=1}^r\theta_\eps^\lambda D_tL_0^\lambda\cos(2\eps^{-1}\phi_0^\lambda)(\bar{\phi}_2+[\phi_2]^\eps)+\frac{1}{8}\sum_{\lambda=1}^r\sum_{\mu=1}^r \theta_\ast^\lambda\theta_\ast^\mu \left\langle DL_0^\mu, DL_0^\lambda \right\rangle \sin(2\eps^{-1}\phi_0^\lambda)\sin(2\eps^{-1}\phi_0^\mu).
\end{IEEEeqnarray*}
To determine the average energy correction at first- and second-order, we take
the weak\(^\ast\) limit and
derive \(E_1^{\parallel\eps}\weak{\ast} 0\) in \(L^\infty([0,T])\) and
\begin{equation*}
    [\bar{E}_2^\parallel]^\eps \weak{\ast} \bar{E}_2^\parallel \coloneqq \left\langle p_0,\bar{p}_2 \right\rangle +\left\langle DV(y_0), \bar{y}_2 \right\rangle -\sum_{\lambda=1}^r\frac{\theta_\ast^\lambda (D_t L_0^\lambda)^2}{4\omega_\lambda(y_0)} + \sum_{\lambda=1}^r \frac{(\theta_\ast^\lambda)^2 \vert DL_0^\lambda\vert^2}{16} \quad \text{in} \quad L^\infty([0,T]).
\end{equation*}

The following theorem shows how the Hamiltonian character of the problem and the
thermodynamic interpretation
materialise for the averaged second-order energy correction
\(\bar{E}_2 = \bar{E}_2^\parallel + \bar{E}_2^\perp\).
\begin{theorem}
    \label{thm:3}
    Let \((y_0, p_0)\) be as in~\eqref{eq:13} and \((\bar{y}_2, \bar{p}_2)\) be as in Theorem~\ref{thm:2}. Let
    \(\bar{E}_2\) be the averaged second-order energy correction
    \(\bar{E}_2= \bar{E}_2^\parallel + \bar{E}_2^\perp\), where
    \begin{equation*}
        \bar{E}_2^\parallel(\bar{y}_2,\bar{p}_2; y_0, p_0) = \left\langle p_0, \bar{p}_2 \right\rangle +\left\langle DV(y_0),\bar{y}_2
        \right\rangle -\sum_{\lambda=1}^r \frac{\theta_\ast^\lambda \left\langle p_0,
            D\omega_\lambda(y_0)\right\rangle^2}{4\omega^3_\lambda(y_0)}+\sum_{\lambda=1}^r \frac{(\theta_\ast^\lambda)^2 \left\vert
            D\omega_\lambda(y_0)\right\vert^2}{16\omega_\lambda^2(y_0)}
    \end{equation*}
    and
    \begin{equation*}
        \bar{E}_2^\perp(\bar{y}_2;y_0, p_0) = \sum_{\lambda=1}^r \bar{\theta}_2^\lambda(y_0, p_0) \omega_\lambda(y_0)+\sum_{\lambda=1}^r
        \theta_\ast^\lambda \left\langle D\omega_\lambda(y_0), \bar{y}_2 \right\rangle,
    \end{equation*}
    with
    \begin{equation*}
        \bar{\theta}_2^\lambda(y_0, p_0)=\frac{\theta_\ast^\lambda \left\langle p_0, D\omega_\lambda(y_0) \right\rangle^2 }{8\omega^4_\lambda(y_0)}+C_{\bar{\theta}^\lambda_2}, \qquad     C_{\bar{\theta}_2^\lambda} = -\frac{\theta_\ast^\lambda \left\langle p_\ast, D\omega_\lambda(y_\ast) \right\rangle^2 }{8\omega^4_\lambda(y_\ast)} -[\theta_2^\lambda]^\eps(0).
    \end{equation*}
    Then the differential equations~\eqref{eq:27c} and~\eqref{eq:27d} take the form
    \begin{equation}
        \label{eq:63}
        \frac{d \bar{y}_2^j}{dt} = \frac{\partial \bar{E}_2}{\partial p_0^j},\qquad \frac{d \bar{p}_2^j}{dt} = -\frac{\partial \bar{E}_2}{\partial y_0^j},
    \end{equation}
    for \(j=1,\ldots, n\). Moreover, with the functions \(T_0\), \(\bar{S}_2\) and \(F_0\) given
    in~\eqref{eq:59} and~\eqref{eq:60}, which can be interpreted as the temperature, entropy and external
    force in the fast subsystem, the energy \(\bar{E}_2^\perp\) can be written as
    \begin{equation*}
        \bar{E}_2^\perp(\bar{S}_2, \bar{y}_2; y_0, p_0) = \left\langle F_0(y_0), \bar{y}_2 \right\rangle +T_0(y_0) \bar{S}_2 + \frac{1}{16rT_0(y_0)} \sum_{\lambda=1}^r \frac{(\theta_\ast^\lambda)^2 \left\langle p_0, D \omega_\lambda(y_0) \right\rangle^2 }{\omega_\lambda^2(y_0)}.
    \end{equation*}
    With this notation, the energy \(\bar{E}_2^\perp\) satisfies the constituent equations
    \begin{equation}
        \label{eq:73}
        T_0 = \frac{\partial \bar{E}^\perp_2}{\partial \bar{S}_2},\qquad F_0^j = \frac{\partial \bar{E}^\perp_2}{\partial \bar{y}_2^j}.
    \end{equation}
\end{theorem}

\begin{proof}
    The evolution equations~\eqref{eq:63} follow directly from~\eqref{eq:27c} and~\eqref{eq:27d}. The
    constituent equations~\eqref{eq:73} follow from~\eqref{eq:61}.
\end{proof}

\begin{remark}
    With \(E_\eps =E_\ast= E_0\), according to~\eqref{eq:4}, the expansion of the energy
    \(E_\eps = E_0 + \eps [\bar{E}_1]^\eps + \eps^2 [\bar{E}_2]^\eps + \eps^2 E_3^\eps\) implies that
    \([\bar{E}_1]^\eps =[\bar{E}_2]^\eps = E_3^\eps\equiv 0 \) and thus \(\bar{E}_2\equiv 0\). As a consequence,
    the averaged energy function \(\bar{E}_2\) acts as a constraint on the system and \(\eps^2\) can be regarded
    as a Lagrange multiplier. Note that the evolution equations~\eqref{eq:63} resemble
    Hamilton's canonical
    equations.
\end{remark}

\section{Simulations}
\label{sec:3}
Fast--slow Hamiltonian systems model, for example, the evolution of molecular
systems, where the slow degrees
of freedom represent the conformal motion of a molecule and the fast degrees of
freedom represent the
molecular vibrations. A crucial component in the fast--slow Hamiltonian system
with the Lagrangian of
Section~\ref{sec:8} is the scale parameter \(\eps\). It often represents a fixed
parameter determined by the
problem in terms of the ratio of the typical timescales of the fast (here
\(z_\eps\)) and slow (here
\(y_\eps\)) degrees of freedom.

In the analysis of molecular systems, one is often primarily interested in the
slow conformal motion of
molecules. As such, a small scale parameter \(\eps\) causes costly overhead
in the numerical derivation of
\(y_\eps\) from~\eqref{eq:2}, since the step size has to be chosen
sufficiently small to account for the fast,
oscillatory motion of \(z_\eps\). Theorem~\ref{thm:1} provides a possible solution
to this problem by deriving
the homogenised system~\eqref{eq:2}.

The homogenised system describes the evolution of the slow degrees of freedom
\(y_0\) only, which can be used
to approximate the evolution of \(y_\eps\). The approximation of \(y_\eps\)
by \(y_0\) comes, however, with a
trade-off. On the one hand, one can choose a larger step size for the
computation of \(y_0\)
from~\eqref{eq:62} than for that of \(y_\eps\) from~\eqref{eq:2}. This
significantly reduces the computational
cost of the numerical integration. On the other hand, approximating \(y_\eps\)
by \(y_0\) introduces an
approximation error which depends on the scale parameter \(\eps\), namely
\(\left\Vert y_\eps - y_0 \right\Vert_{L^\infty([0,T], \R^n)} = \mathcal{O}(\eps^2)\).  Therefore, we extend
in this article the leading-order asymptotic expansion and derive in Theorem~\ref{thm:2} the second-order
correction \([\bar{y}_2]^\eps\) to \(y_0\) such that
\(\left\Vert  \eps^{-2}(y_\eps - y_0) -[\bar{y}_2]^\eps \right\Vert _{L^\infty([0,T], \R^n)}\to 0\) as
\(\eps\to0\). Here, \([\bar{y}_2]^\eps\) takes the form \([\bar{y}_2]^\eps=\bar{y}_2 + [y_2]^\eps\), where
\(\bar{y}_2\) traces the average motion of the second-order correction and can be
derived as the solution to a
slow system of differential equations~\eqref{eq:63} and \([y_2]^\eps \) is the
explicitly given rapidly
oscillating term of the second-order correction.

We compare the global error of approximating \(y_\eps\) by \(y_0\) and by
\(y_0 + \eps^2 [\bar{y}_2]^\eps\)
both on a short and a long time interval, and the associated computation times
for a specific fast--slow
Hamiltonian system described in the next paragraph. The key finding is that the
computation of \(y_0\) and
\([\bar{y}_2]^\eps\), which can be done in parallel, is up to two orders of magnitude
faster than the
computation of \(y_\eps\) of similar accuracy. Moreover, the total computation
time for \(y_0\) and
\(y_0 + \eps^2 [\bar{y}_2]^\eps\) is practically identical, while the global error
\(\Vert y_\eps-y_0-\eps^2 [\bar{y}_2]^\eps \Vert_{L^\infty([0,T], \R^n)}\) is significantly smaller than the
global error \(\Vert y_\eps - y_0\Vert_{L^\infty([0,T], \R^n)}\) on short as well as on long time intervals.

\paragraph*{The test model.}
We consider a fast--slow Hamiltonian system as described in Section~\ref{sec:8},
defined on the Euclidean
configuration space \(M=\R^4\). The test model describes the evolution of two
fast and two slow degrees of
freedom such that \(x = (y,z)\in \R^2 \times \R^2 = \R^4\). Their dynamics is governed by the Lagrangian as
described in~\eqref{eq:67}, with
\begin{equation*}
    V(y_\eps) = \tfrac{1}{2} (y_\eps^1)^4 + \tfrac{1}{2} (y_\eps^2)^4, \qquad \omega_1(y_\eps) = 4 + (y_\eps^1y_\eps^2)^2,\qquad \omega_2(y_\eps) = 2 + \sin(y_\eps^1)
\end{equation*}
and initial values
\begin{equation*}
    y_\eps(0) = (1,-0.5), \qquad \dot{y}_\eps(0) = (1,1.2),\qquad z_\eps(0) = (0,0), \qquad \dot{z}_\eps(0) = (3,2).
\end{equation*}
We simulate the full solution \(y_\eps=(y_\eps^1, y_\eps^2)\) and homogenised approximations, in
particular
the second-order approximation based on Theorems~\ref{thm:1} and~\ref{thm:2}. Specifically, we
compare the
second-order asymptotic expansion \(y_0^1+ \eps^2 [\bar{y}_2^1]^\eps\) with the full trajectory of
\(y_\eps^1\) for short and long time intervals. Here, the superscript
\(1\) denotes the index of the first
component of \(y_\eps\) and indicates the first slow degree of freedom in the
system. A similar comparison of the
second slow degree of freedom, \(y_\eps^2\) and \(y_0^2+ \eps^2 [\bar{y}_2^2]^\eps\), is analogous.

Figure~\ref{fig:13} displays the trajectory of \(y_\eps^1 - y_0^1 - \eps^2 [\bar{y}_2^1]^\eps\) superimposed
on \(y_\eps^1 - y_0^1\), for a long time interval with final time \(T=\eps^{-2}\), where
\(\eps=0.5^3\). It is
evident that the second-order error \(y_\eps^1 - y_0^1 - \eps^2 [\bar{y}_2^1]^\eps\) is significantly smaller
throughout the entire time interval than the leading-order error \(y_\eps^1 - y_0^1\).
This becomes even
clearer in Figure~\ref{fig:14}. There, we observe that the leading-order error grows
faster than the
second-order error, illustrating the increased importance of the second-order
correction
\([\bar{y}_2^1]^\eps\) with time.

\begin{figure}[ht!]
    \captionsetup[subfigure]{width=0.92\textwidth, format=hang, justification=centering}
    \centering
    \begin{subfigure}{.5\textwidth}
        \centering
        \includegraphics[scale=0.42]{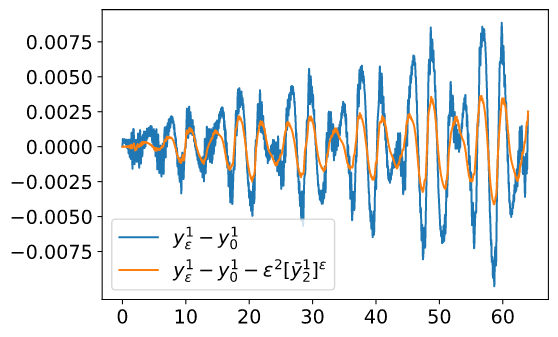}
        \caption{\(y_\eps^1(t) - y_0^1(t)\) and \(y_\eps^1(t) - y_0^1(t) - \eps^2 [\bar{y}_2^1]^\eps(t)\) \\ for
        \(t \in[0,\eps^{-2}]\)}
        \label{fig:13}
    \end{subfigure}%
    \begin{subfigure}{.50\textwidth}
        \centering
        \includegraphics[scale=0.42]{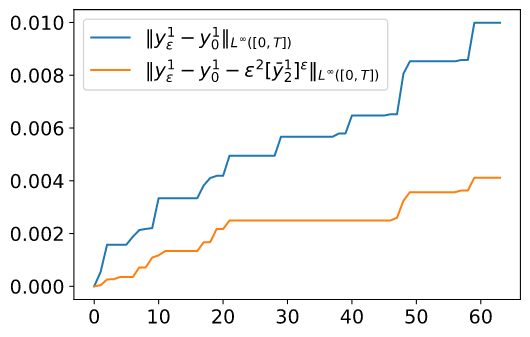}
        \caption{\( \left\Vert y_\eps^1 - y_0^1 \right\Vert_{L^\infty([0,T])}\) and
        \(\left\Vert y_\eps^1 - y_0^1 - \eps^2 [\bar{y}_2^1]^\eps \right\Vert_{L^\infty([0,T])}\) for
        \(T\in [0,\eps^{-2}]\)}
        \label{fig:14}
    \end{subfigure}
    \caption{Comparison of the full dynamics \(y_\eps^1\) as the solution to~\eqref{eq:2}, \(y_0^1\) as the
        solution to the homogenised limit equation~\eqref{eq:62} and \([\bar{y}_2^1]\) as second-order
        approximation, where \(\bar{y}_2^1\) is derived from~\eqref{eq:63}. The parameter choice is \(\eps=0.5^3\).}
    \label{fig:10}
\end{figure}
The reason why an approximation of \(y_\eps^1\) by \(y_0^1\) performs worse
than an approximation by
\(y_0^1 + \eps^2 [\bar{y}_2^1]^\eps\) on long time intervals is that \(y_\eps^1\) is highly oscillatory
at
higher-orders, which is not captured by \(y_0^1\). This difference becomes
evident only to
higher-order. Figure~\ref{fig:15} illustrates the oscillatory behaviour of
\(y_\eps^1\) to second-order. Here,
we superimpose the second-order correction \([\bar{y}_2^1]^\eps = \bar{y}_2^1 + [y_2^1]^\eps\) on top of
\(y_2^{1\eps}= \eps^{-2}(y_\eps^1 - y_0^1)\) to visualise the oscillatory dynamics to higher-order and
illustrate the approximation quality of \([\bar{y}_2^1]^\eps\) for short time intervals.
For this purpose, we
integrate system~\eqref{eq:2},~\eqref{eq:62} and~\eqref{eq:63} for the test model with
\(\eps=0.5^5\), on a
short time interval \(t\in [0,1]\). The trajectories of \(y_2^{1\eps}\) and
\([\bar{y}_2^1]^\eps\) are almost
indistinguishable; the difference becomes visible only at third-order, as shown
in Figure~\ref{fig:16}.

\begin{figure}[ht!]
    \captionsetup[subfigure]{width=0.9\textwidth, justification=centering, format=hang}
    \centering
    \begin{subfigure}{.5\textwidth}
        \centering
        \includegraphics[scale=0.42]{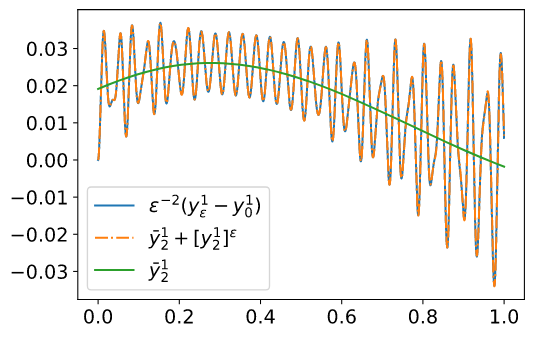}
        \caption{\(\eps^{-2}(y_\eps^1(t) - y_0^1(t))\), \(\bar{y}_2^1(t) + [y_2^1]^\eps(t)\) and
        \(\bar{y}_2^1(t)\) \\ for \(t\in [0,1]\)}
        \label{fig:15}
    \end{subfigure}%
    \begin{subfigure}{.5\textwidth}
        \centering
        \includegraphics[scale=0.42]{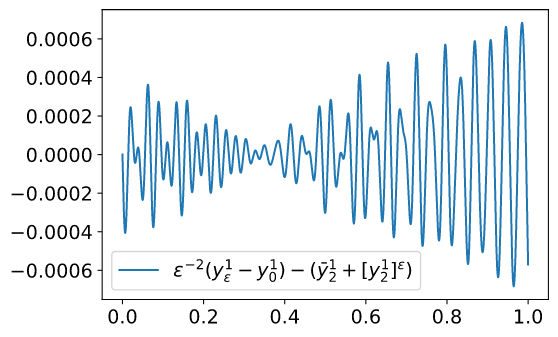}
        \caption{\(\eps^{-2}(y_\eps^1(t) - y_0^1(t)) - (\bar{y}_2^1(t) + [y_2^1]^\eps(t))\) \\for
        \(t\in [0,1]\)}
        \label{fig:16}
    \end{subfigure}
    \caption{Comparison of \(\eps^{-2}(y_\eps^1(t) - y_0^1(t))\) and \(\bar{y}_2^1(t) + [y_2^1]^\eps(t)\) on a
    short time interval \(t\in [0,1]\), with \(y_\eps^1\), \(y_0^1\) and \(\bar{y}_2^1\) as the solutions
    of~\eqref{eq:2},~\eqref{eq:62} and~\eqref{eq:63}, with \(\eps=0.5^5\). The function
    \([y_2^1]^\eps\) is given explicitly in Definition~\ref{def:1}.}
    \label{fig:11}
\end{figure}

Although the error \(y_\eps^1 - y_0^1 - \eps^2[\bar{y}_2^1]^\eps\) looks very accurate on short time
intervals, the accuracy decreases, as Figure~\ref{fig:16} suggest, for long time
intervals. Figure~\ref{fig:12} reveals how this error increases for long time
intervals. Here, we
integrated~\eqref{eq:2},~\eqref{eq:62} and~\eqref{eq:63} for \(\eps=0.5^3\), where \(T=\eps^{-2}\).

\begin{figure}[ht!]
    \centering
    \includegraphics[scale=0.42]{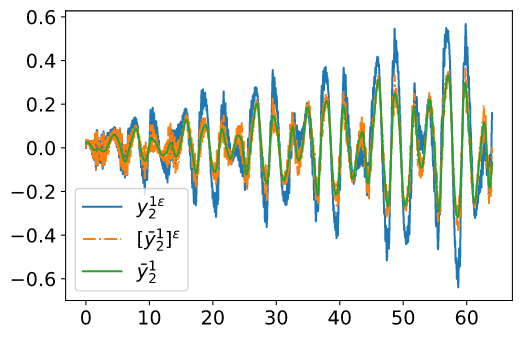}
    \caption{The second-order correction \([\bar{y}_2^1]^\eps(t) = \bar{y}_2^1(t) + [y_2^1]^\eps(t)\) and its
    average motion \(\bar{y}_2^1(t)\) superimposed on \(y_2^{1\eps}(t)=\eps^{-2}(y_\eps^1(t) - y_0^1(t))\) for
    \(t\in [0,\eps^{-2}]\) where \(\eps=0.5^3\).}
    \label{fig:12}
\end{figure}

\subsection{Comparison of the execution time}
\label{sec:2}
As mentioned earlier, the approximation of \(y_\eps^1\) by the homogenisation
limit \(y_0^1\) comes with a
trade-off. The simulation of \(y_0^1\) is faster than that of \(y_\eps^1\)
but introduces an approximation
error of order \(\mathcal{O}(\eps^2)\). This error can be reduced by approximating
\(y_\eps^1\) by
\(y_0^1 +\eps^2[\bar{y}_2^1]^\eps\), i.e.,~the second-order asymptotic expansion derived in Theorem~\ref{thm:2}. The leading-order and second-order errors are discussed in the
previous section. We now
discuss the computational costs of the simulations in this section.

Before comparing the total runtime for deriving \(y_\eps^1\), \(y_0^1\) and
\(y_0^1 + \eps^2[\bar{y}_2^1]^\eps\), we note that in \([\bar{y}_2^1]^\eps= \bar{y}_2^1 + [y_2^1]^\eps\), the
function \(\bar{y}_2^1\) traces the slow, average motion of the second-order
correction term and is given as
the solution to~\eqref{eq:63}, while \([y_2^1]^\eps\) is the explicitly given
rapidly oscillating term of the
second-order correction. Moreover, we point out that the derivation of
\(y_0^1\) and \(\bar{y}_2^1\) can be
carried out in parallel. As such, there is little additional simulational
overhead in computing the
second-order asymptotic expansion \([\bar{y}_2^1]^\eps\).

To analyse the execution time for simulating \(y_\eps^1\), \(y_0^1\) and
\(y_0^1 + \eps^2[\bar{y}_2^1]^\eps\),
we determine the maximal step sizes such that certain convergence properties are
still satisfied. More
precisely, we determine the maximal step size \(dt_{0, y_\eps^1}^{\max}\) to compute
\(y_\eps^1\) and
\(dt_{0, y_0^1}^{\max}\) to compute \(y_0^1\) such that
\begin{equation}
    \label{eq:66}
    \left\Vert y_\eps^1 - y_0^1 \right\Vert_{L^\infty([0,T])} = \mathcal{O}(\eps^2)
\end{equation}
and the maximal step size \(dt_{2, y_\eps^1}^{\max}\) to compute \(y_\eps^1\) and
\(dt_{2, y_0^1,\bar{y}_2^1}^{\max}\) to compute \(y_0^1\) and \(\bar{y}_2^1\) such that
\begin{equation}
    \label{eq:65}
    \left\Vert y_\eps^1 - y_0^1 - \eps^2 [\bar{y}^1_2]^\eps\right\Vert_{L^\infty([0,T])} = \mathcal{O}(\eps^3).
\end{equation}
To determine, for instance, \(dt_{2, y_\eps^1}^{\max}\), we fix \(\eps\) and derive
\(y_0^1\) and
\(\bar{y}_2^1\) with a small but fixed step size \(dt_{2,y_0^1, \bar{y}_2^1}\), and solve
system~\eqref{eq:2}
for increasingly larger \(dt_{2,y_\eps^1}\). This process results in an error plot as
shown in Figure~\ref{fig:4}. The error is constant for small step sizes \(dt_{2,y_\eps^1}\) and
increases after
crossing an \(\eps\)-dependent threshold value. This value expresses the
maximal step size
\(dt_{2, y_\eps^1}^{\max}\) which still ensures that property~\eqref{eq:65} holds. A similar
procedure was
applied to determine the step size \(dt_{0, y_\eps^1}^{\max}\) such that
property~\eqref{eq:66} holds, and
conversely to determine \(dt_{0, y_0^1}^{\max}\) and \(dt_{2, y_0^1,\bar{y}_2^1}^{\max}\).

\begin{figure}[ht!]
    \centering
    \includegraphics[scale=0.42]{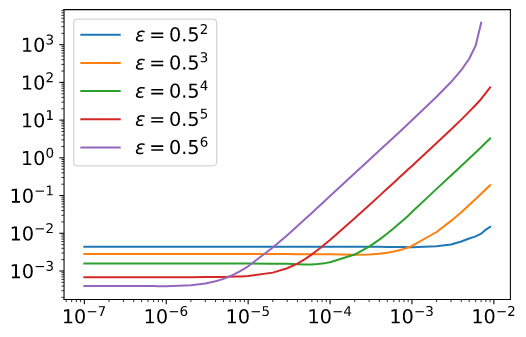}
    \caption{Graphs of \(\left\Vert y_2^{1\eps} - [\bar{y}_2^1]^\eps \right\Vert_{L^\infty([0,T])}\) versus
    step-size \(dt_{2, y_\eps^1}\) for different values of \(\eps\). The start of an upwards slope indicates
    the maximal step-size.}
    \label{fig:4}
\end{figure}

With this procedure, we find for \(\eps=0.5^k\) \((k=2,\ldots, 7)\) the maximal step
size such that the
properties~\eqref{eq:66} and~\eqref{eq:65} are still satisfied. That is, for the
leading-order
error~\eqref{eq:66} we derive \(d_{0,y_\eps^1}^{\max}=\mathcal{O}(\eps^2)\) and
\(d_{0,y_0^1}^{\max}=\mathcal{O}(\eps)\), and for the second-order error~\eqref{eq:65} we obtain
\(dt_{2, y_\eps^1}^{\max}=\mathcal{O}(\eps^3)\) and
\(dt_{2, y_0^1, \bar{y}_2^1}^{\max}=\mathcal{O}(\eps^{3/2})\). The exact maximal step sizes are listed in Tables~\ref{tab:3} and~\ref{tab:4}
in Appendix~\ref{app:2}.

With these maximal step sizes, we can determine the average runtime for
simulating \(y_\eps^1\), \(y_0^1\) and
\(y_0^1 + \eps^2[\bar{y}_2^1]^\eps\) as depicted in Figures~\ref{fig:5} and~\ref{fig:6}. Most significantly, we
see that the derivation of the leading-order asymptotic expansion in Figure~\ref{fig:5} and the second-order
asymptotic expansion in Figure~\ref{fig:6} are up to two orders of magnitude faster
than the simulation of
\(y_\eps^1\) via~\eqref{eq:2}. This directly reflects the differences in the
maximal step sizes as explained
above. Moreover, Figure~\ref{fig:6} shows that the execution time for simulating
\(y_0^1\) as the
leading-order approximation and \(\bar{y}_2^1\) as the averaged second-order
correction are comparable. The
minimal difference can be explained by the evolution
equation~\eqref{eq:63}, which is more
complicated than in~\eqref{eq:62}, resulting in an increase of floating points
operations per time-step.

\begin{figure}[ht!]
    \captionsetup[subfigure]{width=0.89\textwidth, format=hang}
    \centering
    \begin{subfigure}{.5\textwidth}
        \centering
        \includegraphics[scale=0.42]{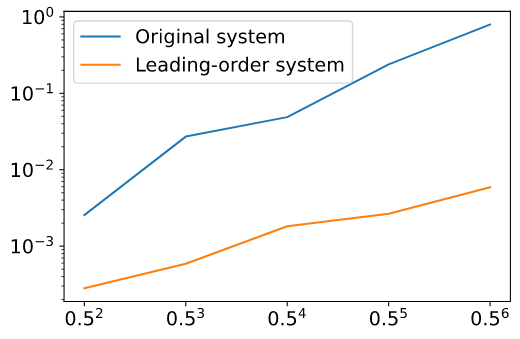}
        \caption{Average runtime in seconds to solve~\eqref{eq:2} and~\eqref{eq:62} for different values of
            \(\eps\).}
        \label{fig:5}
    \end{subfigure}%
    \begin{subfigure}{.5\textwidth}
        \centering
        \includegraphics[scale=0.42]{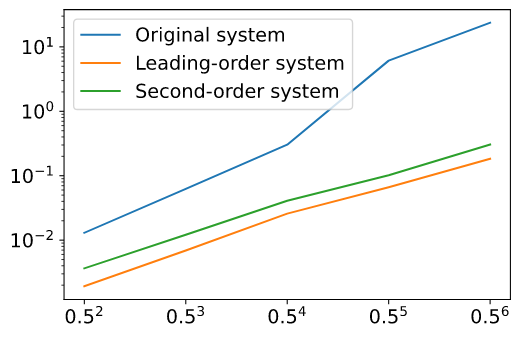}
        \caption{Average runtime in seconds to solve~\eqref{eq:2},~\eqref{eq:62} and~\eqref{eq:63} for different
            values of \(\eps\).}
        \label{fig:6}
    \end{subfigure}
    \caption{Total runtime to simulate \(y_\eps^1\) from~\eqref{eq:2},~\(y_0^1\) from~\eqref{eq:62} and
        \(\bar{y}^1_2\) from~\eqref{eq:63}. The exact computation times are given in Appendix~\ref{app:2}.}
    \label{fig:b}
\end{figure}

\subsection{Details of the implementation}

To compare the runtime of solving system~\eqref{eq:2} for \(y_\eps\)
with an accuracy that describes its
evolution up to second-order with both the leading-order approximation
(Theorem~\ref{thm:1}) and second-order
approximation (Theorem~\ref{thm:2}), one needs a numerical integration scheme that allows to solve each of the
three systems of differential equations~\eqref{eq:2},~\eqref{eq:62} and~\eqref{eq:63}.

We note that system~\eqref{eq:2} and~\eqref{eq:62} are given as two autonomous,
second-order systems of
differential equations. As such, a simple Velocity-Verlet algorithm, which is
frequently used in the numerical
integration of molecular dynamic systems, can be used to integrate these
systems. However,
system~\eqref{eq:63} is non-autonomous. Thus, a numerical integration scheme
from the family of Runge--Kutta
methods could be used to integrate each of the three systems of differential
equations. We notice that
system~\eqref{eq:63} resembles Hamilton's canonical equations. In particular,
the system is separable, which
allows for the implementation of efficient partitioned Runge--Kutta methods.
Furthermore, because of the
Hamiltonian structure of systems~\eqref{eq:2} and~\eqref{eq:62}, it seems natural to apply a
symplectic
partitioned Runge--Kutta method as an integration scheme for solving the three
systems. On that account, the
simulations in this article were derived on the basis of a second-order symplectic
partitioned Runge--Kutta method
which combines the following Lobatto IIIA (Table~\ref{tab:1}) and Lobatto IIIB
(Table~\ref{tab:2}) tableaux
(taken from~\cite[Chapter IV.5]{Hairer1991}). Sun~\cite{Sun93} proved (also
see~\cite{Jay1996}), that this
specific method is symplectic. A detailed description of the implementation can
be found in~\cite[Chapter 8
    and 14]{SanzSerna2018}.

\begin{table}[ht!]
    \centering
    \aboverulesep=0ex
    \belowrulesep=0ex
    \renewcommand{\arraystretch}{1.2}
    \parbox{.30\linewidth}{
        \centering
        \begin{tabular}{@{}c|cc@{}}
            \(0\) & \(0\)   & \(0\)   \\
            \(1\) & \(1/2\) & \(1/2\) \\
            \midrule
                  & \(1/2\) & \(1/2\)
        \end{tabular}
        \caption{Lobatto IIIA}
        \label{tab:1}
    }
    \parbox{.30\linewidth}{
        \centering
        \begin{tabular}{@{}c|cc@{}}
            \(0\) & \(1/2\) & \(0\)   \\
            \(1\) & \(1/2\) & \(0\)   \\
            \midrule
                  & \(1/2\) & \(1/2\)
        \end{tabular}
        \caption{Lobatto IIIB}
        \label{tab:2}
    }
\end{table}

\section{Conclusion}
\label{sec:10}
In this article, we studied a class of fast--slow Hamiltonian systems with energy
functions given by
\begin{equation*}
    E_\eps = \frac{1}{2}\vert \dot{y}_\eps \vert^2 + \frac{1}{2}\vert \dot{z}_\eps \vert^2 +V(y_\eps)+
    \frac{1}{2}\eps^{-2}\sum_{\lambda=1}^r \omega_\lambda^2(y_\eps)(z_\eps^\lambda)^2,
\end{equation*}
where \(y_\eps^j\) \((j=1, \ldots, n)\) are the slow and \(z_\eps^\lambda\) \((\lambda=1,\ldots, r)\)
are the
non-ergodic fast degrees of freedom and \(0<\eps<\eps_0<\infty\) is a parameter
characterising their typical
timescale ratio. A simplified version of one fast and one slow degree of freedom
was already studied
in~\cite{Klar2020}.

In the first part of this article, we introduced a transformation of the fast
degrees of freedom into
action--angle variables \((z_\eps, \dot{z}_\eps)\mapsto (\theta_\eps, \phi_\eps)\), which also required a
transformation of the momenta \(\dot{y}_\eps \mapsto p_\eps\). We derived subsequently the
second-order
asymptotic expansion of the transformed degrees of freedom. Furthermore, we
showed that these expansions can
be decomposed into terms that oscillate rapidly around zero and slow terms that
trace the average motion of
the expansion. While the rapidly oscillating terms are given explicitly, the
slow, average terms are given as
solutions to an inhomogeneous linear system of differential equations.

In the second part of this article, we studied the fast subsystem characterised
by the energy function
\begin{equation*}
    E_\eps^\perp = \frac{1}{2}\vert \dot{z}_\eps \vert^2 + \frac{1}{2}\eps^{-2}\sum_{\lambda=1}^r
    \omega_\lambda^2(y_\eps)(z_\eps^\lambda)^2.
\end{equation*}
Guided by the thermodynamic theory for ergodic Hamiltonian systems described by
Hertz, we regard the dynamics
of the fast degrees of freedom \(z_\eps^\lambda\) \((\lambda=1,\ldots, r)\) as a system that is
slowly
perturbed by the interaction with the slow degrees of freedom \(y_\eps^j\)
\((j=1,\ldots, n)\). Because the
fast subsystem is not ergodic, we followed along the lines
of~\cite{Berdichevsky} and replaced the
time-average in classical statistical mechanics by an ensemble-average and defined
otherwise, following Hertz, the
temperature \(T_\eps\), the entropy \(S_\eps\) and the external force
\(F_\eps\) of the fast subsystem.

Together with the second-order asymptotic expansion derived in the first part of
this article, we expanded
\(E_\eps^\perp\), \(T_\eps\), \(S_\eps\) and \(F_\eps\). After analysing
the leading-order asymptotic
expansion of these terms, we found that they obey an energy relation akin to the
first and second law of
thermodynamics (in the sense of Carath\'eodory)
\begin{equation*}
    d E_0^\perp = \sum_{j=1} ^n F_0^j dy_0^j + T_0 dS_0.
\end{equation*}
In contrast to the case studied in~\cite{Klar2020}, the entropy is not always
constant. Indeed, the entropy is
constant if and only if all weighted frequency ratios
\(\theta_\ast^\lambda \omega_\lambda(y_0)/\omega_\mu(y_0)\) \((\lambda, \mu = 1,\ldots, r)\) are constant. In
this case, the fast subsystem's dynamics is a rigid (quasi-)periodic motion. We
infer that, in the case of a
constant entropy, the fast subsystem can be regarded as an adiabatic
thermodynamic system, while in the case
of a non-constant entropy, it can be interpreted as a non-adiabatic thermodynamic
system.

Remarkably, for the second-order asymptotic expansion we find, for fixed \((y_0, p_0)\), a thermodynamic
energy relation of the form
\begin{equation*}
    d \bar{E}^\perp_2 = \sum_{j=1}^n F_0^j d \bar{y}_2^j + T_0 d \bar{S}_2.
\end{equation*}
With a second-order entropy expression \(\bar{S}_2\) that is not constant, we
can interpret the averaged
second-order asymptotic dynamics as a non-adiabatic thermodynamic process.

Finally, in the third part of this article, we analysed the model problem from a
numerical point of view. In
particular, we compared by means of a specific test model the quality of the
short- and long-term
approximation of \(y_\eps\) by the leading-order asymptotic expansion
\(y_0\) and by the second-order
asymptotic expansion \(y_0 + \eps^2 [\bar{y}_2]^\eps\). Most importantly we found that the time
interval for
which \(y_0+ \eps^2 [\bar{y}_2]^\eps\) ceases to be a viable approximation of \(y_\eps\) is
significantly
longer than for an approximation by \( y_0 \) alone. Moreover, we analysed
in a series of tests how the total
runtime of numerically computing \(y_\eps\), \(y_0\) and \(\bar{y}_2\)
depends on the value of the scale
parameter \(\eps\). We derived experimentally the largest step size so that
certain convergence properties are
still satisfied. In contrast to system~\eqref{eq:62} and~\eqref{eq:63}, which only require
the integration of
slow degrees of freedom and thus allow for choosing a relatively large step
size, the integration of
system~\eqref{eq:2} requires the choice of a relatively small step size to
accurately replicate small-scale
oscillations in the numerical solution. As a consequence, we found that the
runtime for simulating \(y_0\) and
\(\bar{y}_2\), and thus for simulating the second-order asymptotic expansion
\(y_0 + \eps^2 [\bar{y}_2]^\eps\), is up to two orders of magnitude faster than the simulation of
\(y_\eps\)
from the original system, for a similar accuracy.

The analysis of this article is restricted to a simple Hamiltonian. A
significant limitation of the current
analysis is the choice of the interaction potential \(U\)
in~\eqref{eq:U-ass}. The diagonal structure implies
that fast modes interact only indirectly, through slow modes as intermediaries,
via multiplicative
coupling. Such a coupling appears in the Caldeira--Leggett
Hamiltonian~\cite{Caldeira1981a}, with Lagrangian
\begin{equation*}
    \mathscr{L}(y, z, \dot{y}, \dot{z}) = \frac 1 2 M  \dot y^2 - V(y) + \frac 1 2 \sum_{\lambda=1}^r m_\lambda \dot z_\lambda^2 - \frac 1 2 \sum_{\lambda=1}^r m_\lambda c_\lambda z_\lambda^2
    - y \sum_{\lambda=1}^r \omega_\lambda z_\lambda\end{equation*}
with \(c_\lambda >0\). (Note in the framework of this article, the small parameter
would here not be the mass ratio
\(m_\lambda/M\), but the limit of increasing coupling \(\omega_\lambda\).) For direct
practical applications such as chemical
reactions, for example, the evolution of the butane molecule, an extension of the
results presented here to
more complex potentials is required. One of the key insights of this paper is
the existence of thermodynamic
potentials far from equilibrium, albeit in the special situation of diagonal, or
diagonalisable, interaction
potentials \(U\). If this observation holds in greater generality, then
this can lead to a better understanding
and better computational approaches away from equilibrium, such as a chain of
atoms linked to two reservoirs assigning the outer atoms different temperatures.
This is a matter of future investigation.

\subsubsection*{Acknowledgements}
We thank Ben Leimkuhler for stimulating discussions. MK is supported by a scholarship from the EPSRC Centre
for Doctoral Training in Statistical Applied Mathematics at Bath (SAMBa), under the project EP/L015684/1. JZ
gratefully acknowledges funding by a Royal Society Wolfson Research Merit Award. CR acknowledges support from NSF CAREER Award, United States, CMMI-2047506.

\appendix

\section{Hertz' approach to thermodynamics}
\label{app:1}

As mentioned earlier, the authors in~\cite{Klar2020} analyse a simplified
version of the model problem as
presented in Section~\eqref{sec:8} from a thermodynamic point of view. More
precisely, they focus on a system
of one fast and one slow degree of freedom, i.e.,~\(n=r=1\), whose fast
subsystem is by construction
ergodic. That analysis builds on the thermodynamic theory described by Hertz as
presented
in~\cite{Berdichevsky}. Because of the similarity of the two models, we will
focus on the differences in the
derivation of the temperature, entropy and external force as given
in~\eqref{eq:32} and refer the interested
reader for a detailed discussion to~\cite{Klar2020}.

\subsection{Introduction to thermodynamics for non-ergodic systems}

To illustrate the difference in the derivation of the thermodynamic quantities
in~\cite{Klar2020} and here, we
recall how the temperature is derived for the ergodic system studied
in~\cite{Klar2020} and explain why the
same approach fails for non-ergodic systems as studied in this article.

Let us start by analysing the dynamics of a generalised position \(z_\eps\in \R^r\)
and momentum
\(\zeta_\eps\in \R^r\) governed by a Hamiltonian of the form
\begin{equation*}
    H_\eps^\perp(z_\eps, \zeta_\eps;y_\eps) \coloneqq \sum_{\lambda=1}^r H_\eps^\lambda(z_\eps, \zeta_\eps; y_\eps),
    \qquad \text{where} \qquad H_\eps^\lambda(z_\eps, \zeta_\eps; y_\eps) \coloneqq \frac{1}{2} (\zeta_\eps^\lambda)^2 +
    \frac{1}{2}\omega_\lambda^2(y_\eps) (z_\eps^\lambda)^2, \qquad \lambda=1,\ldots, r,\end{equation*}
and \(y_\eps(t)=y(\eps t)\in \R^n\) are slow external parameters with \(\dot{y}_\eps=
\mathcal{O}(\eps)\). This setting
of a Hamiltonian system which is slowly perturbed by an external parameter
is fundamental in the thermodynamic formulation derived by Hertz. For
\(\eps=0\), the unperturbed Hamiltonian
is given by
\begin{equation*}
    H_0^\perp(z_0, \zeta_0;y_0) = \sum_{\lambda=1}^r H_0^\lambda(z_0, \zeta_0; y_0), \qquad \text{where}
    \qquad H_0^\lambda(z_0, \zeta_0; y_0) = \frac{1}{2} (\zeta_0^\lambda)^2 + \frac{1}{2}\omega_\lambda^2(y_0)
    ( z_0^\lambda)^2, \qquad \lambda=1,\ldots, r,
\end{equation*}
and \(y_0(t)=y(0)\equiv y_\ast\). With initial values of the form \(z_0^\lambda(0) = 0\) and
\(\zeta_0^\lambda(0) = \sqrt{2E_\ast^\lambda}\) the solutions to the corresponding Hamilton's equations are
then given for \(\lambda = 1,\ldots, r\) by
\begin{equation}
    \label{eq:41}
    z_0^\lambda(t) = \sqrt{\frac{2 E^\lambda_\ast}{\omega^2_\lambda(y_\ast)}}\sin\left(\omega_\lambda(y_\ast)t \right),
    \qquad \zeta^\lambda_0(t) = \sqrt{2E^\lambda_\ast}\cos\left( \omega_\lambda(y_\ast)t \right).
\end{equation}
Moreover, we define the constant total energy
\begin{equation*}
    E_\ast^\perp \coloneqq \sum_{\lambda=1}^r E_\ast^\lambda, \qquad \text{where} \qquad E_\ast^\lambda
    \coloneqq \frac{1}{2}( \zeta_0^\lambda)^2 + \frac{1}{2}\omega_\lambda^2(y_\ast)( z_0^\lambda)^2,
    \qquad \lambda=1,\ldots, r.
\end{equation*}
If \(r=1\), the trajectory of \((z_0, \zeta_0)\) covers the entire energy
surface
\(\{(z_0, \zeta_0)\in \R^2\colon H_0^\perp(z_0, \zeta_0;y_0)=E_\ast^\perp\}\). Hence, the system is
ergodic. In this case, which corresponds to the model studied
in~\cite{Klar2020}, the temperature in thermal
equilibrium is defined via the time average, indicated by angle brackets
\(\left\langle \cdot \right\rangle \), of twice the kinetic energy.  More precisely, we obtain
\begin{equation*}
    T_\lambda(E_\ast^\lambda, y_\ast)\coloneqq \left\langle \zeta_0^\lambda \frac{\partial H_0^\lambda}{\partial
        \zeta_0^\lambda} \right\rangle = \lim_{\theta\to\infty} \frac{1}{\theta}
    \int_0^\theta 2E_\ast^\lambda \cos^2(\omega_\lambda(y_\ast)t)\ud t  = E_\ast^\lambda,
\end{equation*}
which is unique in the case \(r=1\).

However, if \(r>1\), the energies \(E_\ast^\lambda\) \((\lambda=1,\ldots, r)\) form
distinct integrals of
motion. This implies that the system is non-ergodic. A na\"\i ve application of
the definition of temperature
above results in distinct temperature expressions that are unsuitable to
describe the thermodynamic state of
the whole system, because their values are in general path-dependent,
i.e.,~\(T_\lambda = E_\ast^\lambda\neq E_\ast^\mu = T_\mu\) for \(\lambda\neq \mu\) \((\lambda,\mu=1,\ldots,r)\). Therefore, we define as
in~\cite{Berdichevsky} the temperature for non-ergodic systems via the
ensemble-average. This gives a unique measure for the thermodynamic state of the
whole system.

\subsubsection{The Birkhoff--Khinchin theorem for non-ergodic systems}

A suitable expression for the temperature, which provides a unique measure for
the whole system, can be derived
if, in addition to averaging with respect to time, one averages with respect to
all uniformly distributed
initial values on the energy surface, making the temperature path-independent.
This ensemble average allows us
to define a temperature expression as a measure of the average kinetic motion of the whole
system.  We
follow~\cite{Berdichevsky} for the definition of the ensemble average and its
application to Hamiltonian
systems. Let \(x=x(t,x_0)\) be the parametric form of the trajectory in
phase-space starting at the point
\(x_0\). Then, the average value of some function \(\phi\) with
respect to any phase trajectory \(x(t)\),
i.e.,
\begin{equation*}
    \left\langle \phi(x) \right\rangle = \lim_{\theta\to\infty}\frac{1}{\theta}\int_0^\theta \phi(x(t,x_0))\ud t,
\end{equation*}
depends, in general, on \(x_0\). An ensemble of systems is given by
varying initial data \(x_0\), independent
and identically distributed over the phase region \(E\leq H(x)\leq E+\Delta E\). The probability
density of
\(x_0\) in this region is constant and is equal to \(\left( \Gamma(E+\Delta E)-\Gamma(E) \right)^{-1}\).

The ensemble average, \(E.A.\left\langle \phi \right\rangle \), of the function \(\phi\) is defined
by
\begin{equation*}
    E.A.\left\langle \phi \right\rangle \coloneqq\lim_{\Delta E\to 0}\frac{1}{\Gamma(E+\Delta E)
        -\Gamma(E)}\int_{E\leq H(x_0)\leq E+\Delta E}  \lim_{\theta\to \infty}\frac{1}{\theta}\int_0^\theta
    \phi(x(t,x_0))\ud t \ud x_0.
\end{equation*}
Suppose that the order of calculation of the integral over \(x_0\) and
\(\lim_{\theta\to \infty}\) can be
changed. Then
\begin{equation}
    \label{eq:40}
    E.A.\left\langle \phi \right\rangle =\lim_{\Delta E\to 0}\frac{1}{\Gamma(E+\Delta E)-\Gamma(E)}
    \lim_{\theta\to \infty}\frac{1}{\theta}\int_0^\theta \int_{E\leq H(x_0)\leq E+\Delta E}  \phi(x(t,x_0))\ud x_0\ud t.
\end{equation}
The region \(E\leq H(x_0)\leq E+\Delta E\) is invariant under the action of the phase flow
\(x(t,x_0)\). Hence, in calculating the integral over \(x_0\)
in~\eqref{eq:40}, one can make a change of the
variables \(x_0\mapsto x\). Since the determinant of this transformation is
\(1\) by Liouville's theorem, we can write
\begin{equation*}
    \int_{E\leq H(x_0)\leq E+\Delta E} \phi(x(t,x_0))\ud x_0 = \int_{E\leq H(x)\leq E+\Delta E}\phi(x)\ud x.
\end{equation*}
Thus, the integral does not depend on time. For small \(\Delta E\), this
integral is given by
\begin{equation*}
    \int_{E\leq H(x)\leq E+\Delta E} \phi(x)\ud x \approx \Delta E\int \phi(x)\frac{\ud \sigma}{\left\vert
        \nabla H \right\vert}.
\end{equation*}
Therefore, we arrive at an ``analogous'' version of the Birkhoff--Khinchin
theorem: for any Hamiltonian system
\begin{equation}
    \label{eq:50}
    E.A.\left\langle \phi \right\rangle =\frac{\int_\Sigma \phi(x)\frac{\ud \sigma}{\left\vert \nabla H
            \right\vert}}{\int_\Sigma \frac{\ud \sigma}{\left\vert \nabla H \right\vert}}.
\end{equation}
This version of the Birkhoff--Khinchin theorem reflects the ``average'' (with
respect to initial data)
behaviour of non-ergodic Hamiltonian systems and is thus used to define the
temperature in non-ergodic
systems.

\subsection{Derivation of thermodynamic relations in non-ergodic systems}

As we saw in the previous section, the ensemble average of a function can be
derived from the equality
\begin{equation*}
    E.A.\left\langle \phi \right\rangle =\frac{\int_\Sigma \phi(x)\frac{\ud \sigma}{\left\vert
            \nabla H \right\vert}}{\int_\Sigma \frac{\ud \sigma}{\left\vert \nabla H \right\vert}},
\end{equation*}
where
\(\Sigma= \{ (z_0, \zeta_0) \in \mathbb{R}^{2r}\colon H^\perp_0(z_0, \zeta_0; y_\ast)= E^\perp_\ast \}\), \(\ud \sigma\) is a surface element on the energy surface and
\begin{equation*}
    \vert \nabla H^\perp_0 \vert =  \left[ \sum_{\lambda=1}^r\left( \frac{\partial H^\perp_0}{\partial
            \zeta^\lambda_0} \right)^2+ \left( \frac{\partial H^\perp_0}{\partial z^\lambda_0} \right)^2\right]^{1/2}.
\end{equation*}
The \emph{temperature} for non-ergodic Hamiltonian systems is defined via the
ensemble average by
\begin{equation}
    \label{eq:45}
    T(E^\perp_\ast, y_\ast)\coloneqq E.A.\left\langle \zeta^\lambda_0 \frac{\partial H^\perp_0(z_0,
        \zeta_0; y_\ast)}{\partial \zeta^\lambda_0} \right\rangle
    = \frac{\displaystyle \int_{\Sigma} \zeta^\lambda_0 \frac{\partial H^\perp_0}{\partial \zeta^\lambda_0}\dfrac{\ud
            \sigma}{\vert \nabla H^\perp_0 \vert }}
    {\displaystyle \int_{\Sigma} \frac{\ud \sigma}{\vert \nabla H^\perp_0 \vert}}.
\end{equation}
The numerator can be evaluated by noting that \(\partial H^\perp_0 / \partial \zeta^\lambda_0\) is the
\(\lambda\)th component of the vector \(\nabla H^\perp_0\) and hence
\begin{equation*}
    n^\lambda_\zeta \coloneqq \frac{\partial H^\perp_0/\partial \zeta^\lambda_0}{\vert \nabla H^\perp_0 \vert }
\end{equation*}
is the \(\lambda\)th component of the outer unit vector
\(n = \nabla H^\perp_0 / \vert \nabla H^\perp_0 \vert \) on the energy surface. Therefore, we can write the
numerator in the form
\begin{equation}
    \label{eq:43}
    \int_{\Sigma} \zeta^\lambda_0 \frac{\partial H^\perp_0}{\partial \zeta^\lambda_0}\dfrac{\ud \sigma}{\vert \nabla H^\perp_0 \vert }
    = \int_{\Sigma}\zeta^\lambda_0 n^\lambda_\zeta \ud \sigma = \int_{H^\perp_0(z_0, \zeta_0;y_\ast)\leq E^\perp_\ast} \ud^n (z_0, \zeta_0)
    \eqqcolon \Gamma(E^\perp_\ast, y_\ast) ,
\end{equation}
which follows from Gauss' theorem, where \(\Gamma(E^\perp_\ast, y_\ast)\) is the phase-space volume
enclosed
by the trajectories of~\eqref{eq:41}. To derive the denominator
in~\eqref{eq:45}, we calculate the derivative
of \(\Gamma(E^\perp_\ast, y_\ast)\) with respect to \(E^\perp_\ast\) and find
\begin{equation*}
    \Gamma(E^\perp_\ast + \Delta E^\perp_\ast, y_\ast)-\Gamma(E^\perp_\ast, y_\ast)
    = \int_{E^\perp_\ast \leq H^\perp_0(z_0, \zeta_0; y_\ast) \leq E^\perp_\ast + \Delta E^\perp_\ast}\ud^n (z_0, \zeta_0)
    \approx \int_{H^\perp_0(z_0, \zeta_0;y_\ast)=E^\perp_\ast} \Delta n \ud \sigma,
\end{equation*}
where \(\Delta n\) is the distance between the energy surface \(H^\perp_0 (z_0 + n_z \Delta n, \zeta_0 + n_\zeta \Delta n; y_\ast)=E^\perp_\ast + \Delta E^\perp_\ast\) and
\(H^\perp_0(z_0, \zeta_0;y_\ast) = E^\perp_\ast\). A
Taylor expansion gives \(\Delta n = \Delta E^\perp_\ast / \vert \nabla H^\perp_0 \vert \) and hence
\begin{equation}
    \label{eq:44}
    \frac{\partial \Gamma(E^\perp_\ast, y_\ast)}{\partial E^\perp_\ast}
    = \int_{H^\perp_0(z_0, \zeta_0;y_\ast)= E^\perp_\ast}\frac{\ud  \sigma}{\vert \nabla H^\perp_0 \vert }.
\end{equation}
Combining Equations~\eqref{eq:45}--\eqref{eq:44}, the temperature \(T\) can
thus be expressed in terms of the
phase-space volume \(\Gamma(E^\perp_\ast,y_\ast)\):
\begin{equation}
    \label{eq:46}
    T(E^\perp_\ast, y_\ast) = \frac{\Gamma(E^\perp_\ast,y_\ast)}{\partial \Gamma(E^\perp_\ast,y_\ast)/\partial E^\perp_\ast}.
\end{equation}
Similar to~\cite{Klar2020} we integrate~\eqref{eq:46} with respect to
\(E^\perp_\ast\) and obtain for the
entropy
\begin{equation*}
    S(E^\perp_\ast, y_\ast) = \log \left( \Gamma(E^\perp_\ast, y_\ast) \right) + f(y_\ast),
\end{equation*}
where \(f(y_\ast)\) is a constant of integration with respect to \(E_\ast^\perp\).
To find the dependence of
\(S\) on \(y_\ast\) we follow again the derivation presented
in~\cite{Klar2020}. Using~\eqref{eq:50}, we
calculate for \(j=1,\ldots, n\) the external force
\begin{equation}
    \label{eq:48}
    F_j(E^\perp_\ast, y_\ast) =E.A. \left\langle \frac{\partial H^\perp_0(z_0, \zeta_0; y_\ast)}{\partial y^j_\ast} \right\rangle
    = \frac{\displaystyle \int_{\Sigma} \frac{\partial H^\perp_0}{\partial y^j_\ast}\dfrac{\ud \sigma}{\vert \nabla H^\perp_0 \vert }}
    {\displaystyle \int_{\Sigma} \frac{\ud \sigma}{\vert \nabla H^\perp_0 \vert }}.
\end{equation}

For the numerator, we calculate the derivative of \(\Gamma(E^\perp_\ast, y_\ast)\) with respect to
\(y^j_\ast\). Similarly to before, we have
\begin{IEEEeqnarray*}{rCl}
    \Gamma(E^\perp_\ast, y_\ast + \Delta y_\ast)-\Gamma(E^\perp_\ast, y_\ast)
    &=& \int_{H^\perp_0(z_0, \zeta_0;y_\ast + \Delta y_\ast)\leq E^\perp_\ast}\ud^n (z_0, \zeta_0)
    - \int_{H^\perp_0(z_0, \zeta_0;y_\ast)\leq E^\perp_\ast}\ud^n (z_0, \zeta_0) \\
    &\approx&\int_{H^\perp_0(z_0, \zeta_0;y_\ast)= E^\perp_\ast} \Delta n \ud \sigma,
\end{IEEEeqnarray*}
where \(\Delta n\) indicates the distance between the energy surface
\(H^\perp_0 (z_0 + n_z \Delta n, \zeta_0 + n_\zeta \Delta n; y^j_\ast + \Delta y^j_\ast)= E^\perp_\ast\) and \(H^\perp_0(z_0, \zeta_0;y_\ast)= E^\perp_\ast\). A Taylor expansion gives
\begin{equation*}
    \Delta n = -\frac{1}{\vert \nabla H^\perp_0 \vert }\frac{\partial H^\perp_0}{\partial y^j_\ast} \Delta y^j_\ast
\end{equation*}
and we obtain
\begin{equation}
    \label{eq:49}
    \frac{\partial \Gamma(E^\perp_\ast, y_\ast)}{\partial y^j_\ast}
    = -\int_{H^\perp_0(z_0, \zeta_0;y_\ast)= E^\perp_\ast} \frac{\partial H^\perp_0}{\partial y^j_\ast}
    \frac{\ud \sigma}{\vert \nabla H^\perp_0 \vert }.
\end{equation}
Combining Equations~\eqref{eq:44},~\eqref{eq:48} and~\eqref{eq:49} we obtain
\begin{equation}
    \label{eq:54}
    F_j(E^\perp_\ast, y_\ast) = E.A.\left\langle \frac{\partial H^\perp_0}{\partial y^j_\ast} \right\rangle
    = -\frac{\partial \Gamma(E^\perp_\ast, y_\ast)/\partial y^j_\ast}{\partial \Gamma(E^\perp_\ast, y_\ast)/\partial E^\perp_\ast}.
\end{equation}
We thus find
\begin{equation}
    \label{eq:53}
    S(E^\perp_\ast, y_\ast) = \log \left( \Gamma(E^\perp_\ast, y_\ast) \right) + C.
\end{equation}
The constant \(C\) is chosen such that the entropy is dimensionless.
This is the key result of Hertz'
thermodynamic formulation: the explicit derivation of the entropy of a
Hamiltonian system under the influence
of a slowly varying parameter is (up to a constant) the logarithm of the
phase-space volume.

\subsection{Application to the model problem}

The analysis of the previous section reveals that thermodynamic properties of
Hamiltonian systems are intrinsically connected to the
phase-space volume.

In general, the set \(\{x\in \R^{d}: x^T\Sigma^{-1}x\leq R^2\}\), where
\(\Sigma =\mathrm{diag} (a_1^2, a_2^2,\ldots, a_d^2)\) with \(a_1, \ldots, a_d\in \R\), describes a hyperellipsoid in \(\R^{d}\). Its \(d\)-dimensional volume is given by
\begin{equation}
    \label{eq:52}
    \Gamma=\Gamma_d \vert \Sigma\vert^{1/2}R^d,
\end{equation}
where \(\Gamma_d\) is the volume of the \(d\)-dimensional hypersphere.

To calculate the phase-space volume for the model problem as presented in
Section~\ref{sec:8} note, that the
set \(\{ (z_0, \zeta_0) \in \mathbb{R}^{2r} \colon H^\perp_0(z_0, \zeta_0; y_\ast)=E^\perp_\ast \}\) with
\begin{equation}
    \label{eq:51}
    E^\perp_\ast = \sum_{\lambda=1}^r E_\ast^\lambda = \sum_{\lambda=1}^r \frac{1}{2}(\zeta_0^\lambda)^2
    + \frac{1}{2} \omega_\lambda^2(y_\ast)(z_0^\lambda)^2,
\end{equation}
describes a hyperellipsoid in \(\R^{2r}\). Equation~\eqref{eq:51} can be written
in the form \(E^\perp_\ast = x^T\Sigma^{-1}x\) with \(x = (z_0^1,z_0^2,\ldots,z_0^r,\zeta_0^1,\zeta_0^2, \ldots,\zeta_0^r)\) and
\begin{equation*}
    \Sigma = \mathrm{diag}\left( 2, 2,\ldots, 2,\frac{2}{\omega^2_1(y_\ast)},\frac{2}{\omega^2_2(y_\ast)},
    \ldots , \frac{2}{\omega^2_r(y_\ast)} \right).
\end{equation*}
Therefore, with \(d=2r\) and \(R^2=E_\ast^\perp\), the volume of the
hyperellipsoid~\eqref{eq:51} is according
to~\eqref{eq:52} given by
\begin{equation*}
    \Gamma(E_\ast^\perp, y_\ast)= \Gamma_{2r}\frac{\left( 2E_\ast^\perp \right)^r}{\prod_{\lambda=1}^r\omega_\lambda(y_\ast)}.
\end{equation*}
We reason by analogy that the \(\eps\)-dependent phase-space volume, characterised by
the energy of the fast subsystem
\begin{equation*}
    E_\eps^\perp = \sum_{\lambda=1}^r E_\eps^\lambda = \sum_{\lambda=1}^r \frac{1}{2}(\zeta_\eps^\lambda)^2
    + \frac{1}{2} {\eps}^{-2}\omega_\lambda^2(y_\eps)(z_\eps^\lambda)^2,
\end{equation*}
is given by
\begin{equation}
    \label{eq:55}
    \Gamma_\eps(E_\eps^\perp, y_\eps) = \eps^r \Gamma_{2r}
    \frac{\left( 2E_\eps^\perp \right)^r}{\prod_{\lambda=1}^r\omega_\lambda(y_\eps)}.
\end{equation}
We therefore define, provided that \(\Gamma_\eps(E^\perp_\eps, y_\eps) \neq 0\), the \emph{temperature},
\emph{normalised entropy} and \emph{external force}, in analogy
to~\eqref{eq:46},~\eqref{eq:54} and~\eqref{eq:53}, as
\begin{equation*}
    T_\eps(E_\eps^\perp, y_\eps) \coloneqq\frac{\Gamma_\eps(E^\perp_\eps,y_\eps)}{\partial
        \Gamma_\eps(E^\perp_\eps,y_\eps)/\partial E^\perp_\eps},\quad
    S_\eps(E_\eps^\perp, y_\eps) \coloneqq  \log \left( \Gamma_\eps(E^\perp_\eps, y_\eps) \right) +C_\eps,\quad
    F_\eps^j(E_\eps^\perp, y_\eps) \coloneqq  -\frac{\partial \Gamma_\eps(E^\perp_\eps,
        y_\eps)/\partial y^j_\eps}{\partial \Gamma_\eps(E^\perp_\eps, y_\eps)/\partial E^\perp_\eps},
\end{equation*}
which become with~\eqref{eq:55} and \(C_\eps=-\log((2\eps)^r \Gamma_{2r})\) for \(j=1,\ldots, n\)
\begin{equation*}
    T_\eps =\frac{1}{r}\sum_{\lambda=1}^r \theta_\eps^\lambda \omega_\lambda(y_\eps),\qquad
    S_\eps = \sum_{\lambda=1}^r \log \left( \frac{E_\eps^\perp}{\omega_\lambda(y_\eps)} \right),\qquad
    F_\eps^j =  T_\eps \sum_{\lambda = 1}^r D_j L_\eps^\lambda.
\end{equation*}

\section{Computation times for numerical simulations}
\label{app:2}

For completeness, we present in this section the total computation times corresponding to the maximal step sizes used in the simulations presented in this
article. Tables~\ref{tab:3} and~\ref{tab:4} illustrate in column \(y_\eps^1\)
the total runtime for simulations of
system~\eqref{eq:2} with respect to distinct values of \(\eps\) and a
corresponding maximal step size as
discussed in  Section~\ref{sec:2}. Similarly, the columns \(y_0^1\) and
\(\bar{y}_2^1\) indicate the total runtime for
simulating systems~\eqref{eq:62} and~\eqref{eq:63}. We recall that the maximal step size as
discussed in
Section~\ref{sec:2} is given for the leading-order approximation under the theoretical
global
error~\eqref{eq:66} by \(dt_{0, y_\eps^1}^{\max} = \mathcal{O}(\eps^2)\) and
\(dt_{0, y_0^1}^{\max} = \mathcal{O}(\eps)\), and for the second-order approximation under the theoretical
global error~\eqref{eq:65} by \(dt_{2, y_\eps^1}^{\max} = \mathcal{O}(\eps^3)\) and
\(dt_{2, y_0^1, \bar{y}_2^1}^{\max} = \mathcal{O}(\eps^{3/2})\). Note that we always chose identical step sizes for the derivation
of \(y_0^1\) and \(\bar{y}_2^1\).

The source code for the numerical integration scheme was written in Python
version~3.8.5. The simulations of
the systems~\eqref{eq:2},~\eqref{eq:62} and~\eqref{eq:63} as presented
in Tables~\ref{tab:3} and~\ref{tab:4} were performed on a single core Intel{\textregistered}
Core{\texttrademark}
i5-8250U CPU.

\begin{table}[ht!]
    \centering
    \begin{tabular}{@{}ccccc@{}}
        \toprule
                  & \multicolumn{4}{c}{Computation times (s) and maximal step-sizes}                                                                                                          \\ \cmidrule(r){2-5}
        \(\eps\)  & \phantom{00}\(y_\eps^1\)                                         & \phantom{00}\(dt_{0, y_\eps^1}^{\max}\) & \phantom{00}\(y_0^1\) & \phantom{00}\(dt_{0, y_0^1}^{\max}\) \\ \midrule
        \(0.5^2\) & \phantom{00}0.0026                                               & \phantom{00}\(1\times 10^{-2}\)         & \phantom{00}0.00028   & \phantom{00}\(6\times 10^{-2}\)      \\
        \(0.5^3\) & \phantom{00}0.0271                                               & \phantom{00}\(1\times 10^{-3}\)         & \phantom{00}0.00059   & \phantom{00}\(3\times 10^{-2}\)      \\
        \(0.5^4\) & \phantom{00}0.0488                                               & \phantom{00}\(5\times 10^{-4}\)         & \phantom{00}0.00182   & \phantom{00}\(1\times 10^{-2}\)      \\
        \(0.5^5\) & \phantom{00}0.2392                                               & \phantom{00}\(1\times 10^{-4}\)         & \phantom{00}0.00266   & \phantom{00}\(7\times 10^{-3}\)      \\
        \(0.5^6\) & \phantom{00}0.7945                                               & \phantom{00}\(3\times 10^{-5}\)         & \phantom{00}0.00590   & \phantom{00}\(3\times 10^{-3}\)      \\
        \bottomrule
    \end{tabular}
    \caption{Computation times in seconds for \(y_\eps^1\) and its leading-order asymptotic expansion \(y_0^1\) for
        maximally viable step-sizes that satisfy the theoretical global error~\eqref{eq:66}.}
    \label{tab:3}
\end{table}

\begin{table}[ht!]
    \centering
    \begin{tabular}{@{}cccccc@{}}
        \toprule
                  & \multicolumn{5}{l}{Computation times (s) and maximal step-sizes}                                                                                                    \\ \cmidrule(r){2-6}
        \(\eps\)  & \(y_\eps^1\)                                                     & \(dt_{2, y_\eps^1}^{\max}\) & \(y_0^1\) & \(\bar{y}_2^1\) & \(dt_{2, y_0^1,\bar{y}_2^1}^{\max}\) \\ \midrule
        \(0.5^2\) & \phantom{0}0.013                                                 & \(2\times 10^{-3}\)         & 0.0019    & 0.0037          & \(1\times 10^{-2}\)                  \\
        \(0.5^3\) & \phantom{0}0.062                                                 & \(4\times 10^{-4}\)         & 0.0069    & 0.0121          & \(3\times 10^{-3}\)                  \\
        \(0.5^4\) & \phantom{0}0.303                                                 & \(8\times 10^{-5}\)         & 0.0258    & 0.0408          & \(8\times 10^{-4}\)                  \\
        \(0.5^5\) & \phantom{0}6.122                                                 & \(4\times 10^{-6}\)         & 0.0663    & 0.1017          & \(3\times 10^{-4}\)                  \\
        \(0.5^6\) & 23.670                                                           & \(1\times 10^{-6}\)         & 0.1838    & 0.3041          & \(1\times 10^{-4}\)                  \\
        \bottomrule
    \end{tabular}
    \caption{Computation times in seconds for \(y_\eps^1\), its leading-order asymptotic expansion \(y_0^1\) and
        averaged second-order correction \(\bar{y}_2^1\) for maximally viable step-sizes that satisfy the
        theoretical global error~\eqref{eq:65}. The step-sizes for deriving \(y_0^1\) and \(\bar{y}_2^1\) are
        chosen the same.}
    \label{tab:4}
\end{table}

\printbibliography

\end{document}